\documentclass[letterpaper,11pt,final]{article}

\usepackage[utf8]{inputenc}
\usepackage[dvipsnames]{xcolor}
\usepackage{fullpage}
\usepackage[english]{babel}
\usepackage{amssymb,mathtools}
\usepackage{xspace}
\usepackage{bm}
\usepackage[pdfpagelabels]{hyperref}
\usepackage{csquotes}
\usepackage{color}
\usepackage{microtype}
\usepackage{nicefrac}
\usepackage{tasks}
\usepackage[sort&compress,numbers]{natbib}
\usepackage{paralist}
\usepackage[inline]{enumitem}
\usepackage[ruled,vlined]{algorithm2e}
\usepackage{comment}
\usepackage{interval}
\intervalconfig{
  soft open fences,
}

\usepackage{tikz}
\usetikzlibrary{patterns,decorations.pathreplacing,backgrounds}

\usepackage{empheq}

\newcommand{\Oh}{\mathcal{O}}
\newcommand{\OhOp}[1]{\Oh\mathopen{}\mathclose\bgroup\left( #1 \aftergroup\egroup\right)}
\DeclareMathOperator{\poly}{{\rm poly}}

\usepackage{xspace}
\newcommand{\FPT}{{\sf FPT}\xspace}
\newcommand{\NP}{{\sf NP}\xspace}
\newcommand{\NPh}{\hbox{{\sf NP}-hard}\xspace}

\newcommand{\Wh}[1]{$\mathsf{W[#1]}$-hard\xspace}

\usepackage{tabularx}
\newcommand{\prob}[3]{
\begin{center}
\begin{tabularx}{\textwidth}{lX}
	\multicolumn{2}{l}{#1}\\
	{\bf Input:}&{#2}\\
	{\bf Find:}&{#3}
\end{tabularx}
\end{center}
}

\usepackage[textwidth=20mm,obeyFinal,colorinlistoftodos]{todonotes}
\newcommand{\mytodo}[2]{\todo[size=\tiny, color=#1!50!white]{#2}\xspace}
\newcommand{\myinlinetodo}[2]{\todo[size=\small, color=#1!50!white, inline]{#2}\xspace}

\newcommand{\dkcom}[1]{\mytodo{green}{#1}}
\newcommand{\dkinline}[1]{\myinlinetodo{green}{#1}}
\newcommand{\mkcom}[1]{\mytodo{blue}{#1}}
\newcommand{\mkinline}[1]{\myinlinetodo{blue}{#1}}

\newcommand{\hy}{\hbox{-}\nobreak\hskip0pt}

\newcommand{\suppo}{\textrm{supp}}
\newcommand{\conv}{\textrm{conv}}
\newcommand{\fract}{\mathrm{frac}}
\newcommand{\BB}{\mathcal{B}}
\newcommand{\CC}{\mathcal{C}}
\newcommand{\DD}{\mathcal{D}}
\newcommand{\RR}{\mathcal{R}}
\newcommand{\SSS}{\mathcal{S}}
\newcommand{\JJ}{\mathcal{J}}
\newcommand{\II}{\mathcal{I}}
\newcommand{\LL}{\mathcal{L}}

\DeclarePairedDelimiter\ceil{\lceil}{\rceil}
\DeclarePairedDelimiter\floor{\lfloor}{\rfloor}

\usepackage{stackengine}
\stackMath

\DeclareMathOperator{\leftCritical}{\operatorname{left}}
\DeclareMathOperator{\rightCritical}{\operatorname{right}}
\DeclareMathOperator{\pred}{\operatorname{pred}}
\DeclareMathOperator{\successor}{\operatorname{succ}}

\def\ve#1{\mathchoice{\mbox{\boldmath$\displaystyle\bf#1$}}
{\mbox{\boldmath$\textstyle\bf#1$}}
{\mbox{\boldmath$\scriptstyle\bf#1$}}
{\mbox{\boldmath$\scriptscriptstyle\bf#1$}}}
\newcommand\vea{{\ve a}}
\newcommand\veb{{\ve b}}
\newcommand\vecc{{\ve c}}
\newcommand\veC{{\ve C}}
\newcommand\ved{{\ve d}}
\newcommand\vece{{\ve e}}

\newcommand\veg{{\ve g}}
\newcommand\veh{{\ve h}}

\newcommand\vel{{\ve l}}
\newcommand\veL{{\ve L}}
\newcommand\ven{{\ve n}}

\newcommand\veo{{\ve o}}
\newcommand\vep{{\ve p}}
\newcommand\veq{{\ve q}}
\newcommand\ver{{\ve r}}
\newcommand\ves{{\ve s}}
\newcommand\vet{{\ve t}}
\newcommand\veu{{\ve u}}
\newcommand\vev{{\ve v}}
\newcommand\vew{{\ve w}}
\newcommand\vex{{\ve x}}
\newcommand\vey{{\ve y}}
\newcommand\vez{{\ve z}}
\newcommand\vealpha{{\ve \alpha}}

\newcommand\vegamma{{\ve \gamma}}
\newcommand\vedelta{{\ve \delta}}

\newcommand\velambda{{\ve \lambda}}
\newcommand\vezeta{{\ve \zeta}}
\newcommand\vemu{{\ve \mu}}
\newcommand\vezero{{\ve 0}}

\def\R{\mathbb{R}}
\def\Z{\mathbb{Z}}
\def\N{\mathbb{N}}
\def\G{\mathcal{G}}
\def\Ker{\textrm{Ker}}
\def \la {\langle}
\def \ra {\rangle}

\usepackage{amsthm}

\theoremstyle{plain}
\newtheorem{theorem}{Theorem}
\newtheorem{lemma}[theorem]{Lemma}
\newtheorem{proposition}[theorem]{Proposition}
\newtheorem{corollary}[theorem]{Corollary}
\newtheorem*{fact*}{Fact}

\theoremstyle{definition}
\newtheorem{definition}{Definition}
\newtheorem*{example*}{Example}

\makeatletter
\newtheorem*{rep@theorem}{\rep@title}
\newcommand{\newreptheorem}[2]{%
\newenvironment{rep#1}[1]{%
 \def\rep@title{#2 \ref{##1}}%
 \begin{rep@theorem}}%
 {\end{rep@theorem}}}
\makeatother

\newreptheorem{theorem}{Theorem}
\newreptheorem{corollary}{Corollary}
\newreptheorem{lemma}{Lemma}
\newreptheorem{proposition}{Proposition}

\theoremstyle{remark}

\newtheorem*{claim*}{Claim}

\begin{document}
  \title{Multitype Integer Monoid Optimization and Applications}

  \author{Du\v{s}an Knop\thanks{Algorithmics and Computational Complexity, Faculty~IV, TU Berlin, and Department of Theoretical Computer Science, Faculty of Information Technology, Czech Technical University in Prague, Prague, Czech Republic. \texttt{dusan.knop@fit.cvut.cz}. Partially supported by DFG, project ``MaMu'', NI 369/19 and by the OP VVV MEYS funded project CZ.02.1.01/0.0/0.0/16\_019/0000765 "Research Center for Informatics".}
     \and Martin Kouteck{\'y}\thanks{Technion - Israel Institute of Technology, Haifa, Israel, and Charles University, Prague, Czech Republic. \texttt{koutecky@kam.mff.cuni.cz}. Partially supported by a postdoctoral fellowship at the Technion funded by the Israel Science Foundation grant 308/18, by Charles University project UNCE/SCI/004, and by the project 17-0914-2S of GA ČR.}
     \and Asaf Levin\thanks{Technion - Israel Institute of Technology, Haifa, Israel. \texttt{levinas@ie.technion.ac.il}. Partially supported by a grant from the GIF, the German-Israeli Foundation for Scientific Research and Development (grant number I-1366-407.6/2016), and Israel Science Foundation grant 308/18.}
     \and Matthias Mnich\thanks{Universit{\"a}t Bonn, Bonn, Germany. \texttt{mmnich@uni-bonn.de}. Supported by DFG grant MN 59/4-1.}
     \and Shmuel Onn\thanks{Technion - Israel Institute of Technology, Haifa, Israel. \texttt{onn@ie.technion.ac.il}. Partially supported by the Dresner chair and Israel Science Foundation grant 308/18.}
  }

  \pagenumbering{gobble}
  \maketitle
  \thispagestyle{empty}

\begin{abstract}
  Configuration integer programs (IP) have been key in the design of algorithms for \NPh high-multiplicity problems since the pioneering work of Gilmore and Gomory [Oper. Res., 1961].
  Configuration IPs have one variable for each possible \emph{configuration}, which describes a placement of items into a location, and whose value corresponds to the number of locations with that placement.
  In high multiplicity problems items come in \emph{types}, and are represented \emph{succinctly} by a vector of multiplicities; solving the configuration IP then amounts to deciding whether the input vector of multiplicities of items of each type can be decomposed into a given number of configurations.

  We make this typically implicit notion explicit by observing that the set of all input vectors which can be decomposed into configurations forms a \emph{monoid} of configurations, and the problem corresponding to solving the configuration IP is the \textsc{Monoid Decomposition} problem.
  Then, motivated by applications, we enrich this problem in two ways.
  First, in certain problems each configuration additionally has an objective value, and the problem becomes an \emph{optimization} problem of finding a ``best'' decomposition under the given objective.
  Second, there are often different types of configurations derived from different \emph{types of locations}.
  The resulting problem is then to optimize over decompositions of the input multiplicity vector into configurations of several types, and we call it \textsc{Multitype Integer Monoid Optimization}, or simply MIMO.

  We develop fast exact (exponential-time) algorithms for various MIMO with few or many location types and with various objectives.
  Our algorithms build on a novel proximity theorem which connects the solutions of a certain configuration IP to those of its continuous relaxation.
  We then cast several fundamental scheduling and bin packing problems as MIMOs, and thereby obtain new or substantially faster algorithms for them.

  We complement our positive algorithmic results by hardness results which show that, under common complexity assumptions, the algorithms cannot be extended into more relaxed regimes.
\end{abstract}


\clearpage
\pagebreak
\pagenumbering{arabic}

\section{Introduction}
In this paper we introduce a very general polyhedral optimization problem related to configuration integer programs (IPs), such as those which often appear in the design of algorithms for scheduling, packing, facility location, and other problems.
\dkcom{can we add some surveys/examples of use}

For motivation and illustration, consider the \emph{makespan minimization on identical machines} problem (also denoted $P || C_{\max}$~\cite{LawlerEtAl1993}).
There, we have $m$ identical machines and $n$ jobs of sizes $p_1, \dots, p_n \in \N$, and the goal is to assign jobs to machines such that the maximum total size of jobs at each machine, the \emph{makespan}, is minimized.
A natural input encoding does not list the job sizes one by one; rather, the $n$ jobs are classified into $d \ll n$ \emph{job types} and the input specifies the size $p_j \in \N$ and the number $n_j$ of jobs of type $j$.
The so-called \emph{high multiplicity} encoding~\cite{HochbaumShamir1991} thus gives a~size vector $\mathbf{p} = (p_1,\hdots,p_d)$ and a \emph{multiplicity vector} $\ven = (n_1, \dots, n_d)$ with $\|\ven\|_1 = n_1 + \cdots + n_d = n$.
Observe now that, for any schedule, considering a single machine defines a vector $\vex = (x_1,\dots,x_d)$, with $x_j$ being the number of jobs of type $j$ assigned to this machine; such a vector $\vex$ is called a~\emph{configuration}.

Already in 1961, Gilmore and Gomory~\cite{GilmoreGomory1961} discovered (in the context of \emph{Cutting Stock}) that the $P||C_{\max}$ problem (and many others) can be rephrased in the following way.
\dkcom{possibly: mention that Cutting Stock is HM variant of bin packing}
Let $\mathbb{N}$ be the set $\{ 0,1,\ldots \}$.
Fix a target makespan~$T$; the optimal makespan can be found using binary search in polynomial time, as it is an integer in the range $1,\dots,p_{\max}\cdot n$ (cf. \emph{dual approximation}~\cite{HochbaumShmoys1987}).
Let $\CC_T = \left\{\vex \in \N^d \mid \vep \vex \leq T \right\}$ denote the set of configurations of size at most $T$.
Then, deciding whether the given instance admits a schedule of makespan at most $T$ amounts to deciding whether $\ven$ can be written as a sum of $m$ configurations from~$\CC_T$, i.e., $\ven = \sum_{i=1}^m \vex^i$, with $\vex^i \in \mathcal C_T$ for all $i = 1, \ldots, m$.
Viewed as an instance of \emph{Integer Programming}, we have an integral variable $\lambda_{\vex} \in \N$ for each configuration $\vex \in \CC_T$, and we ask for a solution with $\sum_{\vex \in \CC_T} \lambda_{\vex} = m$ and $\sum_{\vex \in \CC_T} \lambda_{\vex} \cdot \vex = \ven$.
This formulation is called the \emph{configuration IP}.

Our starting point is the set $M(\mathcal C_T, m) = \left\{ \sum_{\vex \in \mathcal C_T} \lambda_{\vex} \vex \mid \sum_{\vex \in \mathcal C_T} \lambda_{\vex} = m, \velambda \in \N^{\mathcal C_T} \right\}$, which we call the \emph{$m$-monoid of $\mathcal C_T$}.
The corresponding problem becomes deciding whether $\ven \in M(\mathcal C_T, m)$, and we call it the \textsc{Monoid Decomposition} problem.
However, it is not rich enough to encode many other relevant problems.
Consider the \emph{minimum sum of weighted completion times on uniformly related machines} scheduling problem, or $Q || \sum w_j C_j$ in the standard notation~\cite{LawlerEtAl1993}.
There, each machine~$i$ has a rational \emph{speed} $s_i \in (0,1]$ and the time to process a job of type $j$ on machine $i$ is $p_j / s_i$.
Again, taking the high multiplicity perspective makes it natural to deal with $\tau$ \emph{machine types} with~$\mu^i$ being the number of machines of type $i$, where machines of a common type have a common speed.
Moreover, each job has a \emph{weight} $w_j \in \N$.
The time a job $j$ finishes is called its \emph{completion time} $C_j$, and the goal is to minimize $\sum w_j C_j$, where the sum ranges over all jobs.

In order to approach this problem from the perspective of monoid decomposition, we will enrich our model in two ways.
First, we introduce different \emph{types of configurations} ($C_T$s') which correspond to the different types of machines.
Second, each configuration now has a certain \emph{cost} corresponding to the contribution of a given machine to the total cost of a schedule.
The problem then becomes finding a decomposition of the input multiplicity vector $\ven$ into configurations (if it exists) such that \begin{enumerate*}[label={\alph*)},font={\bfseries}]
  \item for all $i$ there are exactly $\mu^i$ configurations of type $i$, and,
  \item the total cost of all configurations is minimized.
\end{enumerate*}
In conclusion, we study the following problem:
\dkcom{Maybe we would like to mention $X^1, \ldots, X^\tau$ before the MIMO definition}

\prob{\textsc{Multitype Integer Monoid Optimization} (MIMO)}
{Dimension $d \in \N$, sets $X^1, \dots, X^\tau \subseteq \Z^d$, objective functions $f^1, \dots, f^\tau \colon \Z^d \to \Z$, numbers $\mu^1, \dots, \mu^\tau \in \N$, and a target vector $\ven \in \Z^d$.}
{$\min \big\{\sum_{i=1}^\tau \sum_{\vex \in X^i} f^i(\vex) \cdot \lambda_{\vex}^i \mid {\sum_{i=1}^\tau \sum_{\vex \in X^i} \vex \cdot \lambda_{\vex}^i = \ven},\, \, {\sum_{\vex \in X^i} \lambda_{\vex}^i = \mu^i},\, \lambda_{\vex}^i \in \N \, \forall \vex \in X^i \, \forall i \in \{1,\hdots,\tau\}$ $\big\}$.}

\paragraph{Parameterized Complexity.}
We are dealing with \NPh problems; thus, we are interested in tractable special cases, whose superpolynomial dependence on the input is confined to some small \emph{parameter}.
A natural framework for our study is the theory of \emph{Parameterized Complexity}: a typical input of a scheduling problem (and consequently our MIMO problem) has many natural parameters, such as the numbers of job types, machine types, the largest size, or, more generally, the description complexity of the set of configurations (i.e., the number of necessary linear inequalities and additional variables, or size of coefficients).
Let $I$ be an instance and $p$ be its parameter.
In parameterized complexity, we aim to obtain algorithms with run time $\phi(p) \cdot |I|^{\Oh(1)}$ for some computable function $\phi$, called \emph{fixed-parameter algorithms} (or \emph{\FPT algorithms}).
In the case when the existence of an \FPT algorithm with parameter $p$ is unlikely, we show the problem to be \emph{\Wh{1} parameterized by $p$}.
For further background on parameterized complexity, we refer to the textbook of Cygan et al.~\cite{CyganEtAl2015}.

\subsection{Our contributions}
Our contribution is two-fold.
First, we study the MIMO problem and provide several fixed-parameter algorithms and hardness results, delineating the complexity landscape with regard to the most natural parameters.
Second, to showcase the usefulness and versatility of our approach, we apply our algorithms to high multiplicity problems in scheduling, bin packing, and surfing, a general model of facility location and multicommodity flows.
Our main focus is on scheduling problems, where we show how to model several fundamental scheduling problems as MIMO, and by applying the presented algorithms for solving MIMO we obtain fast fixed-parameter algorithms for these scheduling problems.

\subsubsection{Complexity of MIMO}
We postpone the precise definitions and give informal statements of our main results.
Typically, the sets $X^1, \dots, X^\tau$ are represented succinctly as (projections of) integer points of polyhedra, as in the example of $P||C_{\max}$ when $\CC_T = \{\vex \in \N^d \mid \vep \vex \leq T\}$.
Thus, we let $\|X\|_\infty$ denote the largest coefficient in such a representation of the sets $X^i$.
In the example above we have $\|X\|_\infty = \|\vep\|_\infty = p_{\max}$, that is, the largest size.
By $\la \cdot \ra$ we denote the binary encoding length of numbers; for vectors and matrices it means the total encoding length.
A function $f\colon \Z^d \to \Z$ is \emph{separable convex} if it is possible to write $f(\vex) = \sum_{i=1}^d f_i(x_i)$, where each $f_i\colon \Z \to \Z$ is convex, and $f$ is \emph{fixed-charge} if there is a constant $c\in\mathbb N$ for which $f(\vex) = c$ for every non-zero $\vex$ and $f(\vezero) = 0$.
The fixed-charge objective is useful in modeling \emph{bin packing}-type problems, where opening a new bin incurs a fixed cost.

\begin{theorem}[informal]
\label{thm:implicitMIMO}
  There is an algorithm that solves any succinctly represented MIMO system $\SSS = (d,X^1,\hdots,X^\tau, f^1,\dots,f^\tau,\mu^1,\dots,\mu^\tau,\ven)$ with the largest coefficient $\|X\|_\infty$ and $N=\sum_{i=1}^d \mu^i$
  \begin{enumerate}
    \item\label{repthm:implicitMIMO:micp} in time $\phi(d, N) \la \|X\|_\infty, \ven \ra^{\Oh(1)}$ when each $f^i$ is convex or fixed-charge;
    \item\label{repthm:implicitMIMO:ch} in time $\phi(d, N) (\|X\|_\infty + \la \ven \ra)^{\Oh(1)}$ when each $f^i$ is concave;
    \item\label{repthm:implicitMIMO:hugenfold} in time $\phi(d, \|X\|_\infty) (\tau +\la \ven, N \ra)^{\Oh(1)}$ when each $f^i$ is linear or separable convex;
    \item\label{repthm:implicitMIMO:gr} in time $\phi(d,\tau)(\|X\|_\infty + \la \ven, N \ra)^{\Oh(1)}$ when each $f^i$ is linear or fixed-charge.
  \end{enumerate}
  for a computable function $\phi$ which is single-exponential in the parameters for parts~\ref{thm:implicitMIMO:micp}--\ref{thm:implicitMIMO:huge} and double-exponential for part~\ref{thm:implicitMIMO:gr}.
\end{theorem}

The proofs of parts~\ref{repthm:implicitMIMO:micp} and~\ref{repthm:implicitMIMO:ch} of Theorem~\ref{thm:implicitMIMO} are relatively straightforward applications of known results from integer programming.
Our most significant technical contribution is the proof of part~\ref{repthm:implicitMIMO:hugenfold}, which is based on showing an \FPT algorithm for the \emph{Huge $n$-fold IP} problem.
The main ingredient is a proximity theorem which states that an integer optimum is not too far from a fractional solution obtained from a certain configuration LP.
Crucially, this proximity bound does \emph{not} have to hold for an arbitrary fractional optimum.
To the best of our knowledge, this is the first proximity theorem regarding configuration LP.
To prove part~\ref{repthm:implicitMIMO:gr}, we build on a structure theorem of Goemans and Rothvo{\ss}~\cite{GoemansRothvoss2014}: we extend their result~\cite[Corollary 5.1]{GoemansRothvoss2014} to the case of optimization (rather than feasibility) and to handle input vectors $\ven$ given in \emph{binary} (rather than unary).

An interesting trade-off is seen between parts~\ref{repthm:implicitMIMO:hugenfold} and~\ref{repthm:implicitMIMO:gr}, which we show to be tight: to obtain an \FPT algorithm, either $\tau$ or $\|X\|_\infty$ have to be parameters.
Similarly, we show that the inability to handle a fixed-charge objective in the case of polynomially many types is inherent.
\begin{proposition}[informal]
\label{prop:hardness}
  Solving MIMO systems is
  \begin{enumerate}
    \item \Wh{1} parameterized by $d$ only, even if $\|X\|_\infty$ is given in unary;
    \item and \NPh with a fixed-charge objective even with $d=1$ and $\|X\|_\infty = 1$.
  \end{enumerate}
\end{proposition}

\subsubsection{Applications to Scheduling and Bin Packing with Few Types} \label{sec:intro:apps}
A long-standing open question for high multiplicity \textsc{Bin Packing} was whether it admits a polynomial-time algorithm for constantly many item sizes (i.e., types).
A positive answer was recently given by Goemans and Rothvo{\ss}~\cite{GoemansRothvoss2014}.
While their result also bears on the complexity of some scheduling problems, it \begin{enumerate*}[label={\alph*)},font={\bfseries}] \item does not handle many important objectives, and, \item does not handle many machine types like for example in $Q || C_{\max}$.\end{enumerate*}
Moreover, from the perspective of parameterized complexity, their result is only an \FPT algorithm when the job sizes are given in unary.
Thus, our main research paradigm which we comprehensively address in this paper, reads:
\begin{quote}
\emph{Design fast exact algorithms for high-multiplicity scheduling problems with few job types.}
\end{quote}
We define a general framework for scheduling problems and show how it is modeled using MIMO.
It is then easy to cast many previously studied scheduling problems in terms of this framework, and thereby derive fast algorithms for them by applying Theorem~\ref{thm:implicitMIMO}.

In our framework we consider scheduling problems whose input consists of a set $\mathcal J$ of $n$ jobs that can be partitioned into $d$ types, and a set $\mathcal M$ of $m$ machines that can be partitioned into $\tau$ types.
Each machine in $\mathcal{M}$ belongs to a certain type $i \in\{1,\hdots,\tau\}$: the \emph{machine type $i$} is defined through the machine \emph{speed} $s_i$---which is a rational scaling factor from the interval $(0,1]$---as well as the machine \emph{kind} $K_i \in \{1, \dots, \kappa\}$, where a kind defines a $d$-dimensional vector of job sizes from $\N \cup\{\infty\}$, one for each job type.
In turn, each \emph{job type} $j \in\{1,\hdots,d\}$ is identified by a length-$\kappa$ vector $\vep_j = (p^1_j, \dots, p^\kappa_j)$ of integer job sizes from $\N \cup\{\infty\}$, a positive integer weight $w_j \in \N$, and length-$\tau$ non-negative rational vectors $(r^1_j, \dots, r^\tau_j)$ and $(d^1_j, \dots, d^\tau_j)$ of release times and due dates.
A (non-preemptive) schedule is then an assignment of jobs to time slots of sufficient size on the machines such that each job starts after its release time and two time slots only possibly overlap in their endpoints.
Our choice to allow different release times and due dates on different machine types follows Goemans and Rothvo{\ss}~\cite{GoemansRothvoss2014} and is useful for example to model the scheduled downtimes or availability of specific resources required by a given job type (e.g., storage, processing power, etc.).

We handle multiple objectives.
Consider a fixed schedule that assigns each job in $\mathcal J$ to some machine in $\mathcal M$.
For a job $j\in\mathcal J$ scheduled on a machine $i\in\mathcal M$, let $C_j$ be its \emph{completion time}, let $F_j = C_j - r_j^i$ be its \emph{flow time}, let $L_j = C_j - d_j^i$ be its \emph{lateness}, let $T_j = \max\{0, L_j\}$ be its \emph{tardiness}, and let $U_j = 0$ if $C_j \leq d_j^i$ and $U_j = 1$ if $C_j > d_j^i$ be its \emph{unit penalty}.
Moreover, let $F_{\max} = \max_{j} F_j$ be the \emph{maximum flow time} and let $L_{\max} = \max_{j} L_j$ be the \emph{maximum lateness}.
Let $C^i$ denote the time that the last job finishes on machine $i$ and define the \emph{load} of a machine $i$ to be the total amount of time machine $i$ spends processing jobs, and denote it $\LL^i$.
The \emph{completion times vector} of a schedule is the vector $\veC = (C^1, \dots, C^m)$, and the \emph{load vector} of a schedule is $\veL = (\LL^1, \dots, \LL^m)$.
The objective $C_{\max}$ asks to minimize $\|\veC\|_\infty$, and the objective $C_{\min}$ (which appears e.g. in the \emph{Santa Claus} problem\dkcom{shall we reference Svenson?}) asks to maximize $\min_{i \in [m]} \LL^i$.
\dkcom{$C_{\min}$ is never discussed! Any ideas how to solve it?}
The $\ell_p$-norm objective is to minimize $\|\veC\|_p$, where $p \geq 2$ is an integer.
In $\sum w_j C_j$, $\sum w_j F_j$, $\sum w_j T_j$, and $\sum w_j U_j$, the goal is to minimize the sum of weighted completion times, flow times, tardiness's, and unit penalties, respectively.
The objective $\min \sum w_j U_j$ is also known as \emph{maximum throughput}.
We partition the discussed objectives into two classes: \emph{linear} objectives $\mathfrak{C}_{\text{lin}} = \{C_{\max}, C_{\min}, F_{\max}, L_{\max}, \sum w_j U_j\}$ and \emph{polynomial} objectives $\mathfrak{C}_{\text{poly}} = \{\sum w_j C_j$, $\sum w_j F_j$, $\sum w_j T_j, \ell_p\text{-norm}\}$.
The distinction is that, in the high multiplicity representation, the linear objectives $\mathfrak{C}_{\text{lin}}$ can be expressed as a linear function of the configurations, while the polynomial objectives $\mathfrak{C}_{\text{poly}}$ seem to only be expressible as (non-linear) polynomials.
Let $p_{\max} = \max_{1 \leq j \leq d} \|\vep_j\|_\infty$ denote the largest size of a job with respect to a kind of machine; note that even if $p_{\max}$ is small the largest job size $p_{\max}$ might be large due to small speeds.

\begin{theorem}
\label{thm:scheduling}
  Any scheduling problem expressible in the framework above is \FPT parameterized by
  \begin{enumerate}
    \item $d$ and $m$ for any objective from $\mathfrak{C}_{\textnormal{lin}} \cup \mathfrak{C}_{\textnormal{poly}}$,
    \item \label{thm:scheduling:huge} $d$ and $p_{\max}$ for any objective from $\mathfrak{C}_{\textnormal{lin}} \cup \mathfrak{C}_{\textnormal{poly}}$,
    \item $d$ and $\tau$ for any objective from $\mathfrak{C}_{\textnormal{lin}}$ and with the largest job size $p_{\max}$ given in unary.
  \end{enumerate}
\end{theorem}
Theorem~\ref{thm:scheduling} gives the first \FPT algorithms for the objectives $\sum w_j F_j$, $\sum w_j T_j$, $\ell_p$-norm, $C_{\min}$, $F_{\max}$ and $L_{\max}$.
Moreover, it gives the first \FPT algorithms for the objective $\sum w_j C_j$ in the case of release times and due dates.
Compared to prior work, Theorem~\ref{thm:scheduling} extends all results of Goemans and Rothvo{\ss}~\cite{GoemansRothvoss2014} to the setting of high multiplicities of jobs, and all results of Knop and Kouteck{\'y}~\cite{KnopKoutecky2017} to the setting of jobs with release times and due dates and of high multiplicities of machines.
It can be seen (with some modeling effort) that the parameterization of Part~\ref{thm:scheduling:huge} of Theorem~\ref{thm:scheduling} is less restrictive than the joint parameter $p_{\max}+r$ with~$r$ being the rank of the processing times matrix considered by Chen et al.~\cite{ChenEtAl2017}, thus strengthening their result and also extending it to the setting of high multiplicities of machines.
\dkcom{do we actually give an explicit statement with proof for this -- maybe it would be nice (all in all it comes ``with some modeling effort'')}

\medskip
Furthermore, we show that several high multiplicity \emph{Bin Packing}-type problems, such as cardinality constrained or vector bin packing, can be modeled as MIMO (Theorems~\ref{thm:vector_bin_packing} and~\ref{thm:bin_packing_gcs}).
Lastly, we also study a problem we call \textsc{Surfing}: large numbers of \emph{surfers} (e.g., internet users) of $\tau$ types make demands on $d'$ \emph{commodities} (e.g., types of content, such as video streaming, email, internet telephony, etc.) and are to be served by $d''$ \emph{servers} (e.g., service providers).
The set of solutions is constrained by capacities and costs.
In Theorem~\ref{thm:surfing}, we show that Part~\ref{thm:implicitMIMO:huge} of Theorem~\ref{thm:implicitMIMO} can solve \textsc{Surfing} in \FPT time parameterized by $d=d' \cdot d''$, with only polynomial dependence on $\tau$ and with all other points of the input (number of users, capacities, costs, etc.) encoded in binary.

\subsection{Related Work}
\paragraph*{Configuration Integer Programs.}
The notion of a configuration IP has been introduced in the seminal paper of Gilmore and Gomory~\cite{GilmoreGomory1961}.
Solving the \emph{relaxation} of a configuration IP (the \emph{configuration LP}) via column generation and typically in combination with a separation oracle and the ellipsoid method has been a major tool in the design of approximation algorithms~\cite{KorteVygen2018}.
The (exact) configuration IP is also a common tool, however, it is typically implicit and not referred to as such.
It is used either to directly obtain an algorithm, or as a subprocedure performed on a simplified (preprocessed, e.g., rounded) input, and its solution then serves as the basis for an approximation algorithm; cf.~\cite{HochbaumShmoys1987,AlonAWY97,VegaLeuker81,KarmarkarKarp82}.

An interesting combination of integer programming and column generation was used by Jansen and Solis-Oba~\cite{JansenSolisOba2011} to find an $OPT+1$ solution of {\sc Bin Packing} with $d$ item types in \FPT time.
Goemans and Rothvo{\ss} implicitely deal with the \textsc{Monoid Decomposition} problem~\cite{GoemansRothvoss2014}.
Onn~\cite{Onn2017} considered the \textsc{Monoid Decomposition} problem in the case when the monoid is defined by a totally unimodular matrix.
Recently, Jansen et al.~\cite{JansenEtAl2018} considered a different extension of the configuration IP notion to multiple ``levels'' of configurations, that is, where placements of items are called \emph{modules} and placements of modules are called configurations.

\paragraph*{Tools from Integer Programming.}
To prove part~\ref{thm:implicitMIMO:micp} of Theorem~\ref{thm:implicitMIMO} we use the fact that convex integer minimization is \FPT parameterized by the dimension.
This was first shown by Gr{\"o}tschel et al.~\cite[Theorem 6.7.10]{GrotschelEtAl1993} and the current fastest algorithm is due to Dadush and Vempala~\cite{DadushVempala2013}.
For part~\ref{thm:implicitMIMO:ch}, we use Cook et al.'s algorithm~\cite{CookEtAl1992} to enumerate the vertices of the integer hull of a polytope.
Part~\ref{thm:implicitMIMO:huge} is a continuation of a long line of work on block structured IPs, in particular $N$-fold IPs (cf. Section~\ref{sec:nfold}).
A breakthrough result is an \FPT algorithm of Hemmecke et al.~\cite{HemmeckeEtAl2013}; the most comprehensive result in terms of the parameterizations is currently due to Koutecký et al.~\cite{KouteckyEtAl2018}, and the fastest current algorithms for $N$-fold IPs are due to Eisenbrand et al.~\cite{EisenbrandEtAl2019} and Jansen, Lassota, and Rohwedder~\cite{JansenLR:2018}.
We use the Steinitz Lemma~\cite{Steinitz1913,SevastjanovBanasczyk1997} which has seen renewed interest (cf.~\cite{EisenbrandEtAl2018,EisenbrandEtAl2019,ChenEtAl2018b,JansenRohwedder2018}) after Eisenbrand and Weismantel used it to obtain improved proximity theorems for ILPs with few rows~\cite{EisenbrandWeismantel2018}.
Finally, part~\ref{thm:implicitMIMO:gr} builds on the Structure Theorem of Goemans and Rothvo{\ss}~\cite{GoemansRothvoss2014}, which shows that the integer points of a polytope can be covered with parallelepipeds, whose vertices then serve as ``important'' configurations to which most weight of a solution can be assigned.
Jansen and Klein~\cite{JansenKlein2017} have shown an alternative theorem (with weaker bounds) stating that already the vertices of the integer hull may be used as ``important'' configurations.

\paragraph*{Parameterized and High Multiplicity Scheduling.}
Fixed-parameter algorithms for scheduling problems have been of interest for over 20 years~
\cite{MnichWiese2015,MnichvanBevern2017,HermelinEtAl2015,ChenMarx2018,FellowsMcCartin2003,BodlaenderFellows1995},
their importance has been highlighted 10 years ago by Demaine et al.~\cite{DemaineEtAl2009}, and a research program of ``15 open problems'' has been recently proposed by Mnich and van Bevern~\cite{MnichvanBevern2017}.
We have already mentioned all relevant specific work.
The high multiplicity encoding has been first proposed by Hochbaum and Shamir in 1991~\cite{HochbaumShamir1991} and has seen much attention since then, cf. e.g.~\cite{GranotEtAl1997,CliffordPosner2001,FilippiRomaninJacur2009,BraunerEtAl2005}.
In particular, it was also used by Goemans and Rothvo{\ss}~\cite{GoemansRothvoss2014} and Onn~\cite{Onn2017} who are primary inspiration for our work.

\medskip

\paragraph*{Paper Organization.}
In Section~\ref{sec:preliminaries} we define the implicit representation of MIMO, give a formal statement of Theorem~\ref{thm:implicitMIMO}, connect it to $N$-fold Integer Programming, and prove parts~\ref{thm:implicitMIMO:micp} and~\ref{thm:implicitMIMO:ch} of Theorem~\ref{thm:implicitMIMO}.
In Section~\ref{sec:huge} we present the main technical contribution of the proof of Theorem~\ref{thm:implicitMIMO}, i.e., we prove Part~\ref{thm:implicitMIMO:huge} of Theorem~\ref{thm:implicitMIMO}.
In Section~\ref{sec:part4} we outline the proof of Part~\ref{thm:implicitMIMO:gr} of Theorem~\ref{thm:implicitMIMO} but postpone the full proof to Section~\ref{sec:technical}, as well as our hardness result (Proposition~\ref{prop:hardness})  and some other proofs; postponed proofs are marked by $(\star)$.
In Section~\ref{sec:applications} we deal with applications and model fundamental scheduling problems as MIMO.
Finally, in Section~\ref{sec:researchdirections} we outline further research directions.

\section{Preliminaries}\label{sec:preliminaries}
For positive integers $m,n$ we set $[m,n] = \{m, m+1, \ldots, n\}$ and $[n] = [1,n]$.
We write vectors in boldface (e.g., $\vex, \vey$) and their entries in normal font (e.g., the $i$-th entry of~$\vex$ is~$x_i$ or $x(i)$).
For $\alpha \in \R$, $\floor{\alpha}$ is the floor of $\alpha$, $\ceil{\alpha}$ is the ceiling of $\alpha$, and we define $\{\alpha\} = \alpha - \floor{\alpha}$, similarly for vectors where these operators are defined component-wise.
\dkcom{actually, I have introduced new operator for this, since I am using it together with set and I find it clearer -- any suggestions?}
\subsection{MIMO: Implicit Sets}
The input of MIMO can be given explicitly only in a few cases.
Thus we are naturally interested in the scenario when each (possibly very large) set $X^i$ is defined succinctly.
The following definition captures the case when each $X^i$ is defined as a projection of integer points of a rational polytope.

\begin{definition}[$P$-representation]
  For $i = 1,\dots,\tau$ let $P^i \subseteq \R^{d+d^i}$ be a polytope and let \mbox{$\pi^i((\vex, \vex')) = \vex\in\R^d$} be a projection discarding the last $d^i$ coordinates.
  We call the collection $P^1, \dots, P^\tau$ a \emph{$P$-representation} of $X^1, \dots, X^\tau$ if $X^i = \pi^i(P^i) \cap \Z^d$, for each $i \in [\tau]$.

  Let each $P^i$ by defined as $P^i = \left\{(\vex, \vex') \mid A^i (\vex, \vex') \leq \veb^i\right\}$ for some $A^i \in \Z^{m^i \times (d+d^i)}$ and $\veb^i \in \Z^{m^i}$.
  The \emph{parameters of a $P$-representation} are the following quantities:
  \begin{tasks}[style=itemize](4)
    \task $M = \max_{i \in [\tau]} m^i$,
    \task $D = \max_{i \in [\tau]} d^i$,
    \task $\Delta = \max_{i \in [\tau]} \|A^i\|_\infty$,
    \task $L = \la \Delta, \veb^1, \dots, \veb^\tau\ra$.
  \end{tasks}
\end{definition}

We consistently use superscripts to refer to objects and quantities related to the types (e.g., $X^i, d^i, f^i, \dots$).
To avoid confusion we always use parentheses when intending to express exponentiation (e.g., $(d^i)^2$).
With each $X^i$ given implicitly, the objective functions $f^i$ also must have implicit representations or given by oracles.
We consider the following objective functions:
\begin{compactitem}
  \item linear: given vectors $\vew^1, \dots, \vew^\tau \in \Z^d$, let $f^i(\vex) = \vew^i \vex$.
  \item convex: each $f^i(\vex)$ is a convex function.
  \item extension-separable convex: each $f^i(\vex) = \min_{\vex': (\vex, \vex') \in P^i \cap \Z^{d+d'}} g^i(\vex, \vex')$ for $g^i$ a separable convex function. (In some of our applications the objective is only expressible as a separable convex function in terms of the auxiliary variables $\vex'$.)
  \item concave: each $f^i(\vex)$ is a concave function.
  \item fixed-charge: each $f^i(\vex) = c^i \in \N$ if $\vex \neq \vezero$ and $f^i(\vex)=0$ otherwise; we call $c^i$ a \emph{penalty}.
\end{compactitem}
For an instance of MIMO given in its $P$-representation, set ${\displaystyle f_{\max} = \max_{i \in [\tau]} \max_{\substack{(\vex, \vex') \in \Z^{d+d^i}:\\ A^i(\vex, \vex') \leq \veb^i}} \left| f^i(\vex) \right|}$.

\begin{reptheorem}{thm:implicitMIMO}[Implicit sets]
  Let $\SSS$ be a MIMO system given in its $P$-representation (if the objective $f$ is convex or concave, we assume it is presented by an evaluation oracle), with the parameters as defined, and let $N = \|\vemu\|_1 = \sum_{i=1}^\tau \mu^i$ and $\hat{L} = L + \la \ven, f_{\max}, N \ra$.
  \begin{enumerate}
    \item \label{thm:implicitMIMO:micp} MIMO with a linear, convex, or fixed-charge objective can be solved in time $(N(d+D))^{\Oh(N(d+D))} \hat{L}^{\Oh(1)}$, thus \FPT parameterized by $N$ and $d$.
    \item \label{thm:implicitMIMO:ch} MIMO with a concave objective can be solved in time $( M N (d+D) \cdot \log \Delta)^{\Oh(N (d+D))} \hat{L}^{\Oh(1)}$, thus \FPT parameterized by $N$, $M$, $d+D$, and with $\Delta$ given in unary.
    \item \label{thm:implicitMIMO:huge} MIMO with a linear or an extension-separable convex objective can be solved in time $(M d \Delta)^{\Oh(M^2 d + d^2 M)}  \hat{L}^{\Oh(1)}$, and thus \FPT parameterized by $M$, $d$, and $\Delta$.
    \item \label{thm:implicitMIMO:gr} MIMO with a linear or fixed-charge objective can be solved in time $(\tau d D M \log \Delta)^{\tau^{(d+D)^{\Oh(1)}}}\hat{L}^{\Oh(1)}$ and thus \FPT parameterized by $\tau$, $M$, $d$, and $D$ if $\Delta$ is given in unary.
  \end{enumerate}
\end{reptheorem}
Parts~\ref{thm:implicitMIMO:huge} and~\ref{thm:implicitMIMO:gr} thus mean that we can solve the MIMO problem either in doubly-exponential \FPT time parameterized by $d$, $\tau$, $m$, and $M$ and all numbers have to be given in unary, or single-exponential \FPT parameterized by $d$, $m$, $M$, and the largest coefficient $\Delta$ (but for polynomial $\tau$).
In parts~\ref{thm:implicitMIMO:ch} and~\ref{thm:implicitMIMO:gr} we are using the fact that $(\log \alpha)^\beta \leq 2^{\beta^2 /2} + \alpha^{o(1)}$~\cite[Hint 3.18]{CyganEtAl2015} to say that having $\log \Delta$ in the base amounts to \FPT algorithms when $\Delta$ is given in unary.

\subsection{Modeling MIMO as \texorpdfstring{$N$-fold}{N-fold} IP}
\label{sec:nfold}
In this subsection we will closely connect MIMO with a special class of integer programs, which is an important building block for proving Parts~\ref{thm:implicitMIMO:micp}--\ref{thm:implicitMIMO:huge}.
The \textsc{Integer Programming} problem is to solve:
\begin{equation}
  \min f(\vex):\, A\vex = \veb, \, \vel \leq \vex \leq \veu, \, \vex \in \Z^n,\label{IP} \tag{IP}
\end{equation}
where $f: \R^n \to \R$, $A \in \Z^{m \times n}$, $\veb \in \Z^m$, and $\vel, \veu \in (\Z \cup \{\pm \infty\})^n$.
We denote ${\displaystyle f_{\max} = \max_{\substack{\vex \in \Z^n:\\ \vel \leq \vex \leq \veu}} |f(\vex)|}$.

A generalized $N$-fold IP matrix is defined as
\begin{align}
  E^{(N)} = \left(
              \begin{array}{cccc}
                E^1_1    & E^2_1    & \cdots & E^N_1    \\
                E^1_2    & 0      & \cdots & 0      \\
                  0    & E^2_2    & \cdots & 0      \\
                \vdots & \vdots & \ddots & \vdots \\
                  0    & 0      & \cdots & E^N_2    \\
              \end{array}
            \right) \enspace .\label{nfold}
\end{align}
Here, $r,s,t,N \in \N$, $E^{(N)}$ is an $(r+Ns)\times Nt$-matrix, $E^i_1 \in \Z^{r \times t}$ and $E^i_2 \in \Z^{s \times t}$, $i \in [N]$, are integer matrices, and $E$ is a $2 \times \tau$ block matrix
\[
E = \left(\begin{matrix}
E_1^1& E^2_1 & \cdots & E^\tau_1 \\
E_2^1& 2^2_2 & \cdots & E^\tau_2
\end{matrix}\right) \enspace .
\]
Problem~\eqref{IP} with $A=E^{(N)}$ is known as \emph{generalized $N$\hy fold integer programming} (generalized $N$-fold IP).
``Regular'' $N$-fold IP is the problem where $E_1^i = E_1^j$ and $E_2^i = E_2^j$ for all $i, j \in [N]$.
Recent work indicates that the majority of techniques applicable to ``regular'' $N$-fold IP also applies to generalized $N$-fold IP~\cite{EisenbrandEtAl2019}.
Because generalized $N$-fold IP allows for easier modeling, we will focus on generalized $N$-fold IP and write ``$N$-fold IP'' for short.

We emphasize that while in previous work $N$\hy fold IP was always considered in the regime with small coefficients, variable $N$, and a separable convex objective, here we broaden our focus and consider also other regimes.
Thus, by \emph{$N$\hy fold IP} we refer to \emph{any} IP with a matrix of the form~\eqref{nfold}.

The structure of $E^{(N)}$ allows us to divide any $Nt$-dimensional object, such as the variables of $\vex$, bounds $\vel, \veu$, or the objective $f$, into $N$ \textit{bricks} of size $t$, e.g. $\vex=(\vex^1, \dots, \vex^N)$.
We use subscripts to index within a brick and superscripts to denote the index of the brick, i.e.,~$x_j^i$ is the $j$-th variable of the $i$-th brick with $j \in [t]$ and $i \in [N]$.
We call a brick \emph{integral} if all of its coordinates are integral, and \emph{fractional} otherwise.

\subparagraph{\texorpdfstring{Huge $N$-fold}{Huge N-fold} IP.}
The \emph{huge $N$-fold IP} problem is an extension of $N$-fold IP to the high-multiplicity scenario, where there are potentially \emph{exponentially} many bricks.
This requires a succinct representation of the input and output.
The input to a huge $N$-fold IP problem with $\tau$ \emph{types of bricks} is defined by matrices $E^i_1 \in \Z^{r \times t}$ and $E^i_2 \in \Z^{s \times t}$, $i \in [\tau]$, vectors $\vel^1, \dots, \vel^\tau$, $\veu^1, \dots, \veu^\tau \in \Z^t$, $\veb^0 \in \Z^r$, $\veb^1, \dots, \veb^\tau \in \Z^s$, functions $f^1, \dots, f^\tau \colon \R^{t} \to \R$ satisfying $\forall i \in [\tau], \, \forall \vex \in \Z^t:\, f^i(\vex) \in \Z$ and given by evaluation oracles, and integers $\mu^1, \dots, \mu^\tau \in \N$ such that $\sum_{i=1}^\tau \mu^i = N$.
We say that a brick is of type $i$ if its lower and upper bounds are $\vel^i$ and $\veu^i$, its right hand side is $\veb^i$, its objective is $f^i$, and the matrices appearing at the corresponding coordinates are $E^i_1$ and $E^i_2$.
The task is to solve~\eqref{IP} with a matrix $E^{(N)}$ which has $\mu^i$ bricks of type $i$ for each $i$.
Onn~\cite{Onn2014} shows that for any solution, there exists a solution which is at least as good and has only few (at most $\tau \cdot 2^t$) distinct bricks.
In Section~\ref{sec:huge} we show new bounds which do not depend exponentially on $t$.
\begin{lemma}[$\star$]
    \label{lem:MIMO_as_nfold}
    Let a MIMO system $\SSS$ be given in its $P$-representation.
    Then in time $(\tau + D + M + L + \la \vemu, \ven \ra)$, one can construct a huge $N$-fold IP which models $\SSS$ and has parameters $r= d$, $s = 2M$, $t = d+D+M$, $\|E\|_\infty = \Delta$, $N = \|\vemu\|_1$, and with $f^i$ being the objective for bricks of type $i$.
\end{lemma}
\begin{proof}[Proof idea] The blocks $A^1, \dots, A^\tau$ become blocks $E_2^1, \dots, E_2^\tau$ by padding with zero columns and rows to ensure that each is an $M \times (d+D+M)$ matrix, with the last $M$ coordinates corresponding to slack variables.
The blocks $E^1_1, \dots, E_1^\tau$ are all defined to be $(I~\vezero)$ where $I$ is the $d \times d$ identity matrix and $\vezero$ is an $d \times (D+M)$ all-zero matrix.
\end{proof}
\begin{proof}[Proof of Theorem~\ref{thm:implicitMIMO}, parts~\ref{thm:implicitMIMO:micp} and~\ref{thm:implicitMIMO:ch}]
  Use Lemma~\ref{lem:MIMO_as_nfold} to obtain an $N$-fold IP instance.
  Part~\ref{thm:implicitMIMO:micp} for convex functions follows by applying an algorithm of Dadush and Vempala~\cite{DadushVempala2013}, which runs in time $p^{\Oh(p)} \hat{L}^{\Oh(1)}$, where $p = N (d+D)$ is the dimension.
  For a fixed-charge objective, we may guess, for each $i \in [\tau]$ where $\vezero \in X^i$, a number $\bar{\mu}^i \leq \mu^i$ such that an optimal solution $\velambda$ has $\lambda^i_{\vezero} = \mu^i - \bar{\mu}^i$.
  With this guess at hand, the objective is fully determined to be $\sum_{i=1}^\tau \bar{\mu}^i c^i$ and it remains to verify whether there exists a corresponding decomposition of $\ven$ by solving MIMO with the vector $\bar{\vemu}$ instead of $\vemu$ and without any objective.
  Finally, pick the best among all guesses whose corresponding MIMO is feasible.
  There are at most $N^\tau \leq N^N$ guesses.

  To prove part~\ref{thm:implicitMIMO:ch}, we use an algorithm by Cook et al.~\cite{CookEtAl1992} to enumerate all vertices of the corresponding polyhedron in time $(\log \Delta \cdot M N D)^{\Oh(ND)} \hat{L}^{\Oh(1)}$.
  It is easy to see that a minimum of a concave function is always attained at a vertex.
  Thus, we evaluate the objective on each vertex and return the best as the solution.
\end{proof}

\section{\texorpdfstring{Huge $N$-fold}{Huge N-fold} IP: Part~\ref{thm:implicitMIMO:huge} of Theorem~\ref{thm:implicitMIMO}} \label{sec:huge}
In this section we will prove the following:
\begin{theorem} \label{thm:hugenfold}
  Huge $N$-fold IP can be solved in time $(\|E\|_\infty rs)^{\Oh(r^2s + rs^2)} (t\tau \la f_{\max}, \vel, \veu, \veb, \vemu \ra)^{\Oh(1)}$ with any separable convex objective.
\end{theorem}
With this theorem at hand, part~\ref{thm:implicitMIMO:huge} of Theorem~\ref{thm:implicitMIMO} is easily proven:
\begin{proof}[Proof of Theorem~\ref{thm:implicitMIMO}, part~\ref{thm:implicitMIMO:huge}]
  Construct a huge $N$-fold IP by Lemma~\ref{lem:MIMO_as_nfold}.
  Now apply Theorem~\ref{thm:hugenfold}.
\end{proof}

\subsection{Graver Bases and the Steinitz Lemma}
Let $\vex, \vey$ be $n$-dimensional vectors.
We call $\vex, \vey$ \emph{sign\hy{}compatible} if they lie in the same orthant, that is, for each $i \in [n]$, $x_i \cdot y_i \geq 0$.
We call $\sum_i \veg^i$ a \emph{sign\hy{}compatible sum} if all $\veg^i$ are pair-wise sign\hy{}compatible.
Moreover, we write $\vey \sqsubseteq \vex$ if $\vex$ and $\vey$ are sign\hy{}compatible and $|y_i| \leq |x_i|$ for each $i \in [n]$.
Clearly, $\sqsubseteq$ imposes a partial order, called ``conformal order'', on $n$-dimensional vectors.
For an integer matrix $A \in \Z^{m \times n}$, its \emph{Graver basis} $\G(A)$ is the set of $\sqsubseteq$-minimal non-zero elements of the \emph{lattice} of $A$, $\ker_{\Z}(A) = \{\vez \in \Z^n \mid A \vez = \mathbf{0}\}$.
A \emph{circuit} of $A$ is an element $\veg \in \ker_{\Z}(A)$ whose support $\suppo(\veg)$ (i.e., the set of its non-zero entries) is minimal under inclusion and whose entries are coprime.
We denote the set of \emph{circuits of $A$} by $\CC(A)$.
It is known that $\CC(A) \subseteq \G(A)$~\cite[Definition 3.1 and remarks]{Onn2010}.
We make use of the following two propositions:
\begin{proposition}[{Positive Sum Property~\cite[Lemma 3.4]{Onn2010}}] \label{prop:possum}
Let $A \in \Z^{m \times n}$ be an integer matrix.
For any integer vector $\vex \in \ker_{\Z}(A)$, there exists an $n' \leq 2n-2$ and a decomposition $\vex = \sum_{j=1}^{n'} \alpha_j \veg_j$ with $\alpha_j \in \N$ for each $j \in [n']$, into $\veg_j \in \G(A)$.
For any fractional vector $\vex \in \ker(A)$ (that is, $A\vex=\vezero$), there exists a decomposition $\vex = \sum_{j=1}^{n} \alpha_j \veg_j$ into $\veg_j \in \CC(A)$, where $\alpha_j \geq 0$ for each $j \in [n]$.
\end{proposition}
\begin{proposition}
	[{Separable convex superadditivity~\cite[Lemma 3.3.1]{DeLoeraEtAl2013}}]
	\label{prop:superadditivity}
	Let $f(\vex) = \sum_{i=1}^n f_i(x_i)$ be separable convex, let $\vex \in \R^n$,
	and let $\veg_1,\dots,\veg_k \in \R^n$ be vectors conformal to $\vex$.
	Then
	\begin{align}
	\label{eq:conv_ineq}
	f \left( \vex + \sum_{j=1}^k \alpha_j \veg_j \right) - f(\vex)
	&\geq \sum_{j=1}^k \alpha_j \left( f(\vex + \veg_j) - f(\vex) \right)
	\end{align}
	for arbitrary integers $\alpha_1,\dots,\alpha_k \in \N$.
\end{proposition}

Our proximity theorem relies on the Steinitz Lemma, which has recently received renewed attention recently~\cite{EisenbrandWeismantel2018,EisenbrandEtAl2018,JansenRohwedder2018}.

\begin{lemma}[Steinitz~\cite{Steinitz1913}, Sevastjanov, Banaszczyk~\cite{SevastjanovBanasczyk1997}]
\label{lem:steinitz}
  Let $\|\cdot\|$ denote any norm, and let $\vex_1, \dots, \vex_n \in \R^d$ be such that $\|\vex_i\| \leq 1$ for $i \in [n]$ and $\sum_{i=1}^n \vex_i = 0$.
  Then there exists a permutation $\pi \in S_n$ such that for all $k = 1,\dots,n$, the prefix sums satisfy $\left\|\sum_{i=1}^k \vex_{\pi(i)}\right\| \leq d$.
\end{lemma}

\subsection{Configurations of \texorpdfstring{Huge $N$-fold IP}{Huge N-fold IP}}
Let a huge $N$-fold IP instance with $\tau$ types be fixed.
Recall that $\mu^i$ denotes the number of bricks of type $i$, and $\vemu = (\mu^1, \dots, \mu^\tau)$.
We define for each $i \in [\tau]$ the set of configurations of type $i$ as
\[
\CC^i = \left\{\vecc \in \Z^t \mid E^i_2 \vecc = \veb^i, \, \vel^i \leq \vecc \leq \veu^i \right\} \enspace .
\]

\medskip
Here we are interested in four instances of convex programming (CP) and convex integer programming (IP) related to huge $N$-fold IP.
First, we have the \emph{Huge IP}
\begin{equation} \label{eq:hugenfold}
\min f(\vex): \, E^{(N)} \vex = \veb, \, \vel \leq \vex \leq \veu, \, \vex \in \Z^{Nt}, \tag{HugeIP}
\end{equation}
and the \emph{Huge CP}, which is a relaxation of~\eqref{eq:hugenfold},
\begin{equation} \label{eq:hugenfold_relax}
\min \hat{f}(\vex): \, E^{(N)} \vex = \veb, \, \vel \leq \vex \leq \veu, \, \vex \in \R^{Nt} \enspace . \tag{HugeCP}
\end{equation}
We shall define the objective function $\hat{f}$ later, for now it suffice to say that for all integral feasible $\vex \in \Z^{Nt}$ we have $f(\vex) = \hat{f}(\vex)$ so that indeed the optimum of~\eqref{eq:hugenfold_relax} lower bounds the optimum of~\eqref{eq:hugenfold}.
Then, there is the \emph{Configuration LP} of~\eqref{eq:hugenfold},
\begin{align}
\min \vev \vey & = \sum_{i=1}^\tau \sum_{\vecc \in \CC^i} f^i(\vecc) \cdot y(i, \vecc) & \label{eq:conflp_start} \\
\sum_{i=1}^{\tau} E^i_1 \sum_{\vecc \in \CC^i} \vecc y(i, \vecc) &= \veb^0 &\notag \\
\sum_{\vecc \in \CC^i} y(i, \vecc) &= \mu^i & \forall i \in [\tau]\notag \\
\vey &\geq \mathbf{0} \enspace . & \label{eq:conflp_end}
\end{align}
Letting $B$ be its constraint matrix and $\ved = (\veb^0, \vemu)$ be the right hand side, we can shorten \eqref{eq:conflp_start}-\eqref{eq:conflp_end} as
\begin{equation} \label{eq:conflp}
\min \vev \vey:\, B \vey = \ved, \, \vey \geq \mathbf{0} \enspace . \tag{ConfLP}
\end{equation}
Finally, by observing that $B\vey=\ved$ implies $y(i,\vecc) \leq \|\vemu\|_\infty$ for all $i \in [\tau], \vecc \in \CC^i$, defining $C = \sum_{i \in [\tau]} |\CC^i|$, leads to the \emph{Configuration ILP},
\begin{equation} \label{eq:confilp}
\min \vev \vey: \,B \vey = \ved, \, \mathbf{0} \leq \vey \leq (\|\vemu\|_\infty, \dots, \|\vemu\|_\infty),\, \vey \in \N^{C} \enspace . \tag{ConfILP}
\end{equation}

A solution $\vex$ of~\eqref{eq:hugenfold_relax} is \emph{configurable} if, for every $i \in [\tau]$, every brick $\vex^j$ of type $i$ is a convex combination of $\CC^i$, i.e., $\vex^j \in \conv(\CC^i)$.
We shall define a mapping from solutions of~\eqref{eq:conflp} to configurable solutions of~\eqref{eq:hugenfold_relax} as follows.
For every solution $\vey$ of~\eqref{eq:conflp} we define a solution $\vex = \varphi(\vey)$ of~\eqref{eq:hugenfold_relax} to have $\floor{y(i, \vecc)}$ bricks of type $i$ with configuration $\vecc$ and, for each $i \in [\tau]$, let $\mathfrak{f}_i = \sum_{\vecc \in \CC^i} \{y(i, \vecc)\}$ and let $\vex$ have $\mathfrak{f}_i$ bricks with value $\hat{\vecc}_i = \frac{1}{\mathfrak{f}_i} \sum_{\vecc \in \CC^i} \{y(i, \vecc)\}\vecc$.
(Because $\sum_{\vecc \in \CC^i} y(i,\vecc) = \mu^i$ and $\sum_{\vecc \in \CC^i} \floor{y(i,\vecc)}$ is clearly integral, $\mathfrak{f}_i = \mu^i - \sum_{\vecc \in \CC^i} \floor{y(i,\vecc)}$ is also integral.)
Note that $\varphi(\vey)$ has at most $|\suppo(\vey)|$ fractional bricks since $\sum_{i=1}^{\tau} \mathfrak{f}_i \leq |\suppo(\vey)|$.
Call a solution $\vex$ of~\eqref{eq:hugenfold_relax} \emph{conf-optimal} if there is an optimal solution $\vey$ of~\eqref{eq:conflp} such that $\vex = \varphi(\vey)$.

We are going to introduce an auxiliary objective function $\hat{f}$, but we first want to discuss our motivation in doing so.
The reader might already see that for any integer solution $\vey \in \Z^{C}$ of~\eqref{eq:confilp}, $\vev \vey = f(\varphi(\vey))$ holds, as we shall prove in Lemma~\ref{lem:conflp_hugecp_optima}.
Our natural hope would be that for a fractional optimum $\vey^*$ of~\eqref{eq:conflp} we would have $\vev \vey^* = f(\varphi(\vey^*))$.
However, by convexity of $f$ and the construction of $\hat{\vecc}_i$ it only follows that $\vev \vey^* \geq f(\varphi(\vey^*))$.
Even worse, there may be two conf-optimal solutions $\vex$ and $\vex'$ with $f(\vex) < f(\vex')$.
To overcome this, we define an auxiliary objective function $\hat{f}$ with the property that for any conf-optimal solution $\vex^*$ of~\eqref{eq:hugenfold_relax} and any optimal solution $\vey^*$ of~\eqref{eq:conflp}, $\vev \vey^* = \hat{f}(\vex^*)$.

Fix a brick $\vex^j$ of type $i$.
We say that a multiset $\Gamma^j \subseteq (\CC^i \times \R_{\geq 0})$ is a \emph{decomposition of $\vex^j$} and write $\vex^j = \sum \Gamma^j$ if $\vex^j = \sum_{(\vecc, \lambda_\vecc) \in \Gamma^j} \lambda_{\vecc} \vecc$ and $\sum_{(\vecc, \lambda_{\vecc}) \in \Gamma^j} \lambda_\vecc = 1$.
We define the objective $\hat{f}(\vex)$ for all configurable solutions as $\hat{f}(\vex) = \sum_{j=1}^N \hat{f}^i(\vex^j)$ where
\begin{equation} \label{eq:def_fhat}
\hat{f}^i(\vex^j) = \min_{\Gamma^j: \sum \Gamma^j = \vex^j} \sum_{(\vecc, \lambda_{\vecc}) \in \Gamma^j} \lambda_{\vecc} \cdot f^i(\vecc) \enspace .
\end{equation}
A decomposition $\Gamma^j$ of $\vex^j$ is \emph{$\hat{f}$-optimal} if it is a minimizer of~\eqref{eq:def_fhat}.
\begin{lemma} \label{lem:fhat_convexity}
Let $\vex$ be a configurable solution of~\eqref{eq:hugenfold_relax}, and $\vex^j$ be a brick of type $i$. Then $f^i(\vex^j) \leq \hat{f}^i(\vex^j)$. If $\vex^j$ is integral, then $f^i(\vex^j) = \hat{f}^i(\vex^j)$.
\end{lemma}
\begin{proof}
By convexity of $f^i$ we have
\begin{align*}
f^i(\vex^j) = f^i\left(\sum_{(\vecc, \lambda_\vecc) \in \Gamma^j} \lambda_{\vecc} \vecc\right) & \leq \sum_{(\vecc, \lambda_\vecc) \in \Gamma^j} \lambda_\vecc f^i(\vecc),
\end{align*}
for any decomposition $\Gamma^j$ of $\vex^j$.
If $\vex^j$ is integral, then $\Gamma^j = \{(\vex^j, 1)\}$ is its optimal decomposition (indeed, it is the only decomposition), concluding the proof.
\end{proof}

Moreover, for each $\vex^j$ there is an $\hat{f}$-optimal decomposition $\Gamma^j$ with $|\Gamma^j|\leq t+1$ since $\hat{f}$-optimal decompositions correspond to optima of a linear program with $t+1$ equality constraints, namely
\begin{equation}
\min \sum_{\vecc \in \CC^i} \lambda_{\vecc} f^i(\vecc) \quad \text{s.t.} \quad \sum_{\vecc \in \CC^i} \lambda_{\vecc} \vecc = \vex^j, \, \|\velambda\|_1 = 1, \, \velambda \geq \vezero \enspace .\label{eq:fhat_lp}
\end{equation}
Let us describe the relationship of the objective values of the various formulations.

\begin{lemma} \label{lem:conflp_fhat}
For any feasible solution $\tilde{\vey}$ of~\eqref{eq:conflp},
\begin{equation} \label{eq:varphi_nonincreasing}
\vev \tilde{\vey} \geq \hat{f}(\varphi(\tilde{\vey})) \enspace .
\end{equation}
\end{lemma}
\begin{proof}
Let $\tilde{\vex} = \varphi(\tilde{\vey})$.
We can decompose $\hat{f}(\varphi(\tilde{\vey})) = U_1 + U_2$, where $U_1$ is the cost of integer bricks of $\varphi(\tilde{\vey})$ and $U_2$ is the cost of its fractional bricks.
It is easy to see that $U_1 = \vev \floor{\tilde{\vey}}$ by the equality of $f^i$ and $\hat{f}^i$, for all $i \in [\tau]$, over integer vectors.
We shall further decompose the value $U_2$ into costs of fractional bricks of each type.
For each $i \in [\tau]$, the cost of each fractional brick of type $i$ is at most $\frac{1}{\mathfrak{f}_i} \sum_{\vecc \in \CC^i} \{\tilde{y}(i, \vecc)\}f^i(\vecc)$ because the decomposition $\left\{\left(\vecc, \frac{1}{\mathfrak{f}_i} y^i_{\vecc}\right) \middle| \vecc \in \CC^i\right\}$ of $\hat{\vecc}_i$ (recall that $\hat{\vecc}_i = \frac{1}{\mathfrak{f}_i} \sum_{\vecc \in \CC^i} \{y(i, \vecc)\}\vecc$) is merely a feasible (not necessarily optimal) solution of~\eqref{eq:fhat_lp}, and summing up this estimate over all $\mathfrak{f}_i$ fractional bricks of type $i$ gives $\mathfrak{f}_i \cdot \frac{1}{\mathfrak{f}_i} \sum_{\vecc \in \CC^i} \{\tilde{y}(i, \vecc)\}f^i(\vecc) = \vev^i \{\vey^i\}$, concluding the proof.
\end{proof}

\begin{lemma} \label{lem:conflp_hugecp_optima}
Let $\hat{\vey}$ be an optimum of~\eqref{eq:confilp}, $\vez^*$ an optimum of~\eqref{eq:hugenfold}, $\vey^*$ an optimum of~\eqref{eq:conflp}, $\tilde{\vex} = \varphi(\vey^*)$, and $\vex^*$ a configurable optimum of~\eqref{eq:hugenfold_relax}.
Then
\[\hat{f}(\vez^*) = f(\vez^*) = f(\varphi(\hat{\vey})) = \vev \hat{\vey} \geq \vev \vey^* = \hat{f}(\tilde{\vex}) = \hat{f}(\vex^*) \enspace .\]
\end{lemma}
\begin{proof}
We have $\hat{f}(\vez^*) = f(\vez^*)$ by equality of $\hat{f}$ and $f$ on integer solutions (Lemma~\ref{lem:fhat_convexity}) and $f(\vez^*) = f(\varphi(\hat{\vey})) = \vev \hat{\vey}$ by the definition of $\varphi$ and the fact that $\hat{\vey}$ is an integer optimum.
Clearly $\vev \hat{\vey} \geq \vev \vey^*$ because~\eqref{eq:conflp} is a relaxation of~\eqref{eq:confilp} and thus the former lower bounds the latter.

Let us define a mapping $\phi$ for any configurable solution $\vex$ of~\eqref{eq:hugenfold_relax}.
Start with $\phi(\vex) = \vey = \vezero$.
For each brick $\vex^j$ of type $i$ let $\Gamma^j$ be its $\hat{f}$-optimal decomposition and update $y^i_\vecc := y^i_\vecc + \lambda_\vecc$ for each $(\vecc, \lambda_\vecc) \in \Gamma^j$.
Now it is easy to see that
\begin{equation} \label{eq:phi_equal}
\vev \phi(\vex) = \hat{f}(\vex) \enspace .
\end{equation}
Our goal is to argue that $\vev \vey^* = \hat{f}(\tilde{\vex}) = \hat{f}(\vex^*)$.
We have $\hat{f}(\tilde{\vex}) = \hat{f}(\varphi(\vey^*)) \leq \vev \vey^*$ by~\eqref{eq:varphi_nonincreasing}, but by optimality of $\vey^*$ and~\eqref{eq:phi_equal} it must be that $\vev \phi(\tilde{\vex}) = \hat{f}(\tilde{\vex}) \geq \vev \vey^*$ and hence $\vev \vey^* = \hat{f}(\tilde{\vex})$.
Similarly,
$$\hat{f}(\vex^*) = \vev \phi(\vex^*) \geq \vev\vey^* \geq \hat{f}(\varphi(\vey^*))$$
with the ``$=$'' by~\eqref{eq:phi_equal}, the first ``$\geq$'' by optimality of $\vey^*$, and the second ``$\geq$'' by~\eqref{eq:varphi_nonincreasing}.
However, since $\hat{f}(\varphi(\vey^*)) \geq \hat{f}(\vex^*)$ by optimality of $\vex^*$, all inequalities are in fact equalities and thus $\vev \vey^* = \hat{f}(\vex^*)$.
\end{proof}

Our goal is to show that the proximity of any conf-optimal solution $\vex^*$ of~\eqref{eq:hugenfold_relax} from an integer optimum $\vez^*$ of~\eqref{eq:hugenfold} depends on the number of fractional bricks.
This number, by definition of~$\varphi$, corresponds to the size of the support of the corresponding solution $\vey$ of~\eqref{eq:conflp}.
The following lemma shows how to produce optima of~\eqref{eq:conflp} with small support.
We emphasize that our proximity theorem does not require that the fractional solution be optimal but rather conf-optimal.

\begin{lemma}
\label{lem:frac}
  An optimal solution $\vey^*$ of~\eqref{eq:conflp} with $|\suppo(\vey^*)| \leq r + \tau$ and a conf-optimal solution $\vex^* = \varphi(\vey^*)$ of~\eqref{eq:hugenfold_relax} with at most $r + \tau$ fractional bricks can be found in time $\|E^{(N)}\|_\infty^{\Oh(s^2)} (r t \tau \la f_{\max}, \vel, \veu, \veb, \vemu \ra)^{\Oh(1)}$.
\end{lemma}
\begin{proof}
  Recall that $\tau$ is the number of brick types in the huge $N$-fold instance.
  Since \eqref{eq:conflp} has exponentially many variables, we take the standard approach and solve the dual LP of~\eqref{eq:conflp} by the ellipsoid method and the equivalence of optimization and separation. Thus in $T = (r t \tau \la f_{\max}, \vel, \veu, \veb, \vemu \ra)^{\Oh(1)}$ calls to a separation oracle  we find an optimal solution while only considering at most $T$ inequalities of the dual (we make this argument specific later in the proof).
  Moreover, we can assume that the discovered optimal solution is a vertex of the dual LP.
  Observe that the dimension of the dual LP is the number of rows of the primal LP, which is $r + \tau$.
  Since each point in $(r+\tau)$-dimensional space is fully determined by $r+\tau$ linearly independent inequalities, there must exist a subset $I$ of $r+\tau$ inequalities among the $T$ inequalities considered by the ellipsoid method which fully determines the dual optimum.

  We can find them as follows.
  Taking the $T$ considered inequalities one by one, if either some inequality of $I$ or the present inequality is dominated\footnote{An inequality $\vea \vex \leq b$ is dominated by $\vecc \vex \leq d$ if for every $\vex$ such that $\vea \vex \leq b$ we also have $\vecc \vex \leq d$.} by an inequality that can be obtained as a non-negative linear combination of the others, discard it; otherwise, include it in $I$ and continue.
  Testing whether an inequality $\ved \vez \leq e'$ is dominated by a non-negative combination of a system of inequalities $D \vez \leq \vece$ can be decided by solving
  \begin{equation}
  \min \vealpha \vece \quad \text{s.t.} \quad \vealpha^{\intercal} D = \ved, \, \vealpha \geq \vezero, \label{eq:redundant_ineq_lp}
  \end{equation}
  and checking whether the optimal value is at most $e'$.
  If it is, then the solution $\vealpha$ encodes a non-negative linear combination of the inequalities $D \vez \leq \vece$ which yields an inequality dominating $\ved \vez \leq e'$, and if it is not, then such a combination does not exists.
  Thus, when a new inequality is considered, we solve~\eqref{eq:redundant_ineq_lp} for at most $r+\tau$ inequalities (the new one and all less than $r+\tau$ already selected ones), and there are $T$ inequalities considered.
  The time needed to solve~\eqref{eq:redundant_ineq_lp} is $\poly(r+\tau, \la \vel, \veu, \veb, f_{\max} \ra)$ because its dimension is at most $r+\tau$ and its encoding length is at most $\la \vel, \veu, \veb, f_{\max} \ra$.
  Altogether, we need time $T \cdot (r+\tau) \cdot \poly(r+\tau, \la \vel, \veu, \veb, f_{\max} \ra) \leq \poly(r t \tau \la f_{\max}, \vel, \veu, \veb, \vemu \ra) =: T'$.

  Let the \emph{restricted~\eqref{eq:conflp}} be the~\eqref{eq:conflp} restricted to the variables corresponding to the inequalities in $I$.
  We claim that an optimal solution to the restricted~\eqref{eq:conflp} is also an optimal solution to~\eqref{eq:conflp}.
  To see that, use LP duality: the optimal objective value of the dual LP restricted to inequalities in $I$ is the same as one of the dual optima, and thus an optimal solution of the restricted~\eqref{eq:conflp} must be an optimal solution of~\eqref{eq:conflp}.
  Finally, we solve the restricted~\eqref{eq:conflp} using any polynomial LP algorithm in time $T'' \leq ((r+\tau) \la f_{\max}, \vel, \veu, \vemu \ra)^{\Oh(1)}$.

  Clearly, the main task is thus solving the Dual LP, which we discuss in the rest of the proof.
  The Dual LP of~\eqref{eq:conflp} in variables $\bm{\alpha} \in \R^r$, $\bm{\beta} \in \R^{\tau}$ is:
  \begin{align}
            \max  & & \veb^0 \bm{\alpha} + \sum_{i=1}^\tau\mu^i \beta^i & \notag \\
    \textrm{s.t.} & & (\bm{\alpha} E^i_1) \vecc - f^i(\vecc)              & \leq -\beta^i & \forall i \in [\tau], \,\forall \vecc \in \CC^i \label{eq:dualcons}
  \end{align}
  To verify feasibility of $(\bm{\alpha}, \bm{\beta})$ for $i \in [\tau]$, we need to maximize the left-hand side of~\eqref{eq:dualcons} over all $\vecc \in \CC^i$ and check if it is at most $-\beta^i$.
  This corresponds to finding integer variables $\vecc$ which for given $(\bm{\alpha}, \bm{\beta})$ solve
  \begin{equation*}
    \min f^i(\vecc)- (\bm{\alpha} E^i_1) \vecc = -\max \,(\bm{\alpha} E^i_1) \vecc - f^i(\vecc) \,:\, E^i_2 \vecc = \veb^i, \, \vel^i \leq \vecc \leq \veu^i, \vecc \in \Z^t \enspace .
  \end{equation*}
  The program above can be solved in time $T''' \leq \|E_2\|_\infty^{s^2} t^3 \la \veb^i, \vel^i, \veu^i \ra$~\cite[Theorem 4]{KouteckyEtAl2018}.
  Grötschel et al.~\cite[Theorem 6.4.9]{GrotschelEtAl1993} show that an optimal solution of LP (even one which is a vertex~\cite[Remark 6.5.2]{GrotschelEtAl1993}) can be found in a number of calls to a separation oracle which is polynomial in the dimension and the encoding length of the inequalities returned by a separation oracle.
  Clearly the inequalities~\eqref{eq:dualcons} have encoding length bounded by $\la f_{\max}, \vel, \veu, \veb, \vemu \ra$ and thus $T = (r t \tau \la f_{\max}, \vel, \veu, \veb, \vemu \ra)^{\Oh(1)}$ calls to a separation oracle are sufficient to find an optimal vertex solution.
  The resulting total time complexity is thus $T \cdot T''' + T'$ to construct the restricted~\eqref{eq:conflp} instance and time $T''$ to solve it, $T \cdot T''' + T' + T''$ total, which is upper bounded by $\|E^{(N)}\|_\infty^{\Oh(s^2)} (r t \tau \la f_{\max}, \vel, \veu, \veb, \vemu \ra)^{\Oh(1)}$, as claimed.

  Let $\vey^*$ be an optimum of~\eqref{eq:conflp} we have thus obtained.
  Since $|I| \leq r+\tau$, the support of $\vey^*$ is of size at most $r+\tau$.
  Now setting $\vex^* = \varphi(\vey^*)$ is enough, since we have already argued that $\vex^*$ has at most $|\suppo(\vey^*)|$ fractional bricks.
\end{proof}

\subsection{Proximity theorem}
Let us give a plan for the next subsection.
We wish to prove that for every conf-optimal solution $\vex^*$ of~\eqref{eq:hugenfold_relax} there is an integer solution $\vez^*$ of~\eqref{eq:hugenfold} nearby.
In the following, let $\vex^*$ be a conf-optimal solution of~\eqref{eq:hugenfold_relax} and $\vez^*$ be an optimal solution of~\eqref{eq:hugenfold} minimizing $\|\vex^* - \vez^*\|_1$.
A technique for proving proximity theorems which was introduced by Eisenbrand and Weismantel~\cite{EisenbrandWeismantel2018} works as follows.
A vector $\veh \in\Z^{Nt}$ is called a \emph{cycle} of $\vex^* - \vez^*$ if $\veh \neq \vezero$, $E^{(N)} \veh  = \vezero$, and $\veh \sqsubseteq \vex^* - \vez^*$.
It is not too difficult to see that if $\vex'$ is an optimal (\emph{not} conf-optimal) solution of~\eqref{eq:hugenfold_relax} with the objective $f$, then there cannot exist a cycle of $\vex' - \vez^*$ (cf. proof of Lemma~\ref{thr:4}).
Based on a certain decomposition of $\vex^* - \vez^*$ into integer and fractional smaller dimensional vectors and by an application of the Steinitz Lemma, the existence of a cycle is proven unless $\|\vex^*-\vez^*\|_1$ is roughly bounded by the number of fractional bricks of $\vex^*$.
However, we cannot apply this technique directly as an optimal solution $\vex'$ of~\eqref{eq:hugenfold_relax} might have many fractional bricks.
On the other hand, an existence of a cycle $\veh$ of $\vex^* - \vez^*$ does not necessarily contradict that $\|\vex^* - \vez^*\|_1$ is minimal, because $\vex^* + \veh$ might not be a conf-optimal solution, which is an essential part of the argument.

All of this leads us to introduce a stronger notion of a cycle.
We say that $\veh \in \Z^{Nt}$ is a \emph{configurable cycle} of $\vex^* - \vez^*$ (with respect to $\vex^*$) if $\veh$ is a cycle of $\vex^* - \vez^*$, for each brick $j \in [N]$ of type $i \in [\tau]$ there exists an $\hat{f}$-optimal decomposition $\Gamma^j$ of $(\vex^*)^j$ such that we may write $\veh^j = \sum_{(\vecc, \lambda_\vecc) \in \Gamma^j} \lambda_\vecc \veh_\vecc$ and for each $(\vecc, \lambda_\vecc) \in \Gamma^j$ we have $\veh_\vecc \sqsubseteq \vecc - (\vez^*)^j$ and $\veh_\vecc \in \Ker_\Z(E^i_2)$.
Soon we will show that if $\|\vex^* - \vez^*\|_1$ is minimal, $\vex^* - \vez^*$ does not have a configurable cycle.
The next task becomes to show how large must $\|\vex^* - \vez^*\|_1$ be in order for a configurable cycle to exist.
Recall that the technique of Eisenbrand and Weismantel~\cite{EisenbrandWeismantel2018} can be used to rule out an existence of a (regular) cycle, not a configurable cycle.
To overcome this, we ``lift'' both $\vex^*$ and $\vez^*$ to a higher dimensional space and show that a cycle in this space corresponds to a configurable cycle in the original space.
Only then are we ready to prove a proximity bound using the aforementioned technique.
\begin{lemma} \label{lem:conf_cycle_conf_sol}
If $\veh$ is a configurable cycle of $\vex^* - \vez^*$, then $\vex^* - \veh$ is configurable.
\end{lemma}
\begin{proof}
Fix $j \in [N]$.
Let $\vep$ be the brick $(\vex^* - \veh)^j$ and let $i \in [\tau]$ be its type.
Now $\vep$ can be written as $\vep = \sum_{(\vecc, \lambda_\vecc) \in \Gamma^j} \lambda_\vecc (\vecc - \veh_\vecc)$.
Furthermore, we have $E^i_2(\vecc - \veh_{\vecc}) = E^i_2 \vecc = \veb^j$, and, by $\veh \sqsubseteq \vex^* - \vez^*$, we also have $\vel \leq \vex^* - \veh \leq \veu$.
\end{proof}
We now need a technical lemma:
\begin{lemma} \label{lem:cycle_inequality}
Let $\vex^*$ be a conf-optimal solution of~\eqref{eq:hugenfold_relax}, let $\vez^*$ be an optimum of~\eqref{eq:hugenfold}, and let $\veh^*$ be a configurable cycle of $\vex^* - \vez^*$.
  Then
  \begin{align}
\hat{f}(\vez^* + \veh^*) + \hat{f}(\vex^* - \veh^*) &\leq \hat{f}(\vez^*) + \hat{f}(\vex^*) \enspace . \label{eq:cycle_convex}
\end{align}
\end{lemma}
\begin{proof}
We begin by a simple observation: let $g: \R \to \R$ be a convex function, $x \in \R$, $z \in \Z$, and $r \in \Z$ be such that $r \sqsubseteq x-z$. By convexity of $g$ we have that
\begin{equation} \label{eq:cycle_convex_univariate}
g(z+r) + g(x-r) \leq g(z) + g(x) \enspace .
\end{equation}
Fix $j \in [N]$ and $\vez = (\vez^*)^j$, $\vex = (\vex^*)^j$, $\veh = (\veh^*)^j$, and let $i$ be the type of brick $j$.
Since $\veh^*$ is a configurable cycle there exists an $\hat{f}$-optimal decomposition $\Gamma$ of $\vex$ such that, for each $(\vecc, \lambda_\vecc) \in \Gamma$, there exists a $\veh_\vecc \sqsubseteq \vecc - \vez$, $\veh_\vecc \in \Ker_\Z(E^i_2)$, and $\veh = \sum_{(\vecc, \lambda_\vecc) \in \Gamma} \lambda_\vecc \veh_\vecc$.
Due to separability of $f$ we may apply~\eqref{eq:cycle_convex_univariate} independently to each coordinate, obtaining for each $\vecc$
\begin{align*}
f^i(\vez + \veh_\vecc) + f^i(\vecc - \veh_\vecc) &\leq f^i(\vez) + f^i(\vecc) \enspace .
\end{align*}
Since all arguments of $f^i$ are integral, we immediately get
\begin{align*}
\hat{f}^i(\vez + \veh_\vecc) + \hat{f}^i(\vecc - \veh_\vecc) &\leq \hat{f}^i(\vez) + \hat{f}^i(\vecc) \enspace .
\end{align*}
Aggregating according to $\Gamma$, we get (recall that we have $\sum_{(\vecc, \lambda_{\vecc}) \in \Gamma} \lambda_{\vecc} = 1$)
\begin{align*}
\sum_{(\vecc, \lambda_\vecc) \in \Gamma} \lambda_\vecc \left(\hat{f}^i(\vez + \veh_\vecc) + \hat{f}^i(\vecc - \veh_\vecc) \right) & \leq \sum_{(\vecc, \lambda_\vecc) \in \Gamma} \lambda_\vecc \left(\hat{f}^i(\vez) + \hat{f}^i(\vecc) \right) = \hat{f}^i(\vez) + \sum_{(\vecc, \lambda_\vecc) \in \Gamma} \lambda_{\vecc}\hat{f}^i(\vecc),
\end{align*}
where by $\hat{f}$-optimality of $\Gamma$ the right hand side is equal to $\hat{f}^i(\vez) + \hat{f}^i(\vex)$.
As for the left hand side, observe that decompositions $\Gamma' = \{(\vez + \veh_\vecc, \lambda_\vecc) \mid (\vecc, \lambda_\vecc) \in \Gamma\}$ and $\Gamma'' = \{(\vecc - \veh_\vecc, \lambda_\vecc) \mid (\vecc, \lambda_\vecc) \in \Gamma\}$ satisfy $\sum \Gamma' = \vez + \veh$ and $\sum \Gamma'' = \vex - \veh$ but are only feasible (not necessarily optimal) solutions of~\eqref{eq:fhat_lp} and thus we have
$$\hat{f}^i(\vez + \veh) + \hat{f}^i(\vex - \veh) \leq \sum_{(\vecc, \lambda_\vecc) \in \Gamma} \lambda_\vecc \left(\hat{f}^i(\vez + \veh_\vecc) + \hat{f}^i(\vecc - \veh_\vecc) \right) \enspace .$$
Combining then yields
$$\hat{f}^i(\vez + \veh) + \hat{f}^i(\vex - \veh) \leq \hat{f}^i(\vez) + \hat{f}^i(\vex),$$
and since we have proven this claim for every brick $j$, aggregation over bricks concludes the proof of the main claim~\eqref{eq:cycle_convex}.
\end{proof}
Let us show that if $\vex^*$ and $\vez^*$ are as stated, then there is no configurable cycle of $\vex^* - \vez^*$.
\begin{lemma}
	\label{thr:4}
	Let $\vex^*$ be a conf-optimal solution of~\eqref{eq:hugenfold_relax} and let $\vez^*$ be an optimal solution of~\eqref{eq:hugenfold} such that $\|\vex^*-\vez^*\|_1$ is minimal.
	Then there is no configurable cycle of $\vex^* - \vez^*$.
\end{lemma}
\begin{proof}[Proof of Lemma~\ref{thr:4}]
For contradiction assume that there exists a configurable cycle $\veh^*$ of $\vex^* - \vez^*$.
By Lemma~\ref{lem:cycle_inequality}, one of two cases must occur:

\noindent\textbf{Case 1: $\hat{f}(\vez^* + \veh^*) \leq \hat{f}(\vez^*)$.}
Then $\vez^* + \veh^*$ is an optimal integer solution (by $\veh \sqsubseteq \vex^* - \vez^*$ we have $\vel \leq \vez^* + \veh \leq \veu$ and by $\veh^* \in \ker_{\Z}\left(E^{(N)}\right)$ we have $E^{(N)} (\vez^* + \veh) = \veb$) which is closer to $\vex^*$, a contradiction to minimality of $\|\vex^* - \vez^*\|_1$.

\noindent\textbf{Case 2: $\hat{f}(\vex^* - \veh^*) < \hat{f}(\vex^*)$.}
Since $\veh^*$ is a configurable cycle, Lemma~\ref{lem:conf_cycle_conf_sol} states that $\vex^* - \veh^*$ is configurable, so we have a contradiction with conf-optimality of $\vex^*$.
\end{proof}
Next, we show that for each brick, there exists an $\hat{f}$-optimal decomposition whose coefficients have small encoding length.
For any matrix $A$, define $g_\infty(A) = \max_{\veg \in \G(A)} \|\veg\|_\infty$.
\begin{lemma} \label{lem:balanced_decomposition}
Each brick of $\vex^*$ of type $i$ has an $\hat{f}$-optimal decomposition $\Gamma$ of size at most $t+1$ and $\max_{(\vecc, \lambda_{\vecc} = p_{\vecc} / q_{\vecc} \in \Gamma)}\{p_\vecc,q_\vecc\} \leq (t+1)! ((2t-2) g_\infty(E^i_2))^{t+1} \leq (t+1)^{(t+1)}(s\|E^i_2\|_\infty +1)^{(s+1)(t+2)}$.
\end{lemma}
\begin{proof}
An $\hat{f}$-optimal decomposition corresponds to a solution of the LP~\eqref{eq:fhat_lp}.
We will argue that there is a solution whose support is composed of columns which do not differ by much, which corresponds to a solution of an LP with small coefficients, and the claimed bound can then be obtained by Cramer's rule.

Specifically, we claim that there exists an $\hat{f}$-optimal decomposition $\Gamma$ which corresponds to an optimal solution $\velambda$ of~\eqref{eq:fhat_lp} such that there exists a point $\vezeta \in \Z^t$ and $\vecc \in \suppo(\velambda) \implies \|\vecc - \vezeta\|_\infty \leq (t-1) g_\infty(E^i_2)$.
For a solution $\velambda$ of~\eqref{eq:fhat_lp}, define $R':= \max_{\vecc, \vecc' \in \suppo(\velambda)} \|\vecc - \vecc'\|_\infty$ to be the diameter of the bounding box of all $\vecc \in \suppo(\velambda)$, define $\vezeta \in \Z^t$ to be an integer center of the bounding box (i.e., $\|\vecc - \vezeta\|_\infty \leq \ceil{\frac{R'}{2}}$), for $\vecc \in \suppo(\lambda)$ define a coordinate $j \in [t]$ to be \emph{tight} if $c_j = \zeta_j - \ceil{\frac{R'}{2}}$ or $c_j = \zeta_j + \ceil{\frac{R'}{2}}$, and define
$S = \sum_{\vecc \in \suppo(\velambda)} \lambda_{\vecc} \sum_{j=1}^t [\text{$j$ is tight in $\vecc$}]$ (where ``$[X]$'' is an indicator of the statement $X$) to be the weighted number of tight coordinates.
For contradiction assume that $\velambda$ is an optimal solution of~\eqref{eq:fhat_lp} which minimizes $R'$ and $S$ and $R' > (2t-2)g_\infty(E^i_2)$.
Assuming $\Gamma$ is a decomposition of a brick of type $i$, we have $\vecc, \vecc' \in \CC^i = \{\tilde{\vecc} \in \Z^t \mid E^i_2 \tilde{\vecc} = \veb^i, \, \vel^i \leq \tilde{\vecc} \leq \veu^i\}$ and thus $\vecc - \vecc' \in \Ker_{\Z}(E^i_2)$.
By Proposition~\ref{prop:possum} we may write $\vecc - \vecc' = \sum_{j=1}^{2t-2} \lambda_j \veg_j$ with $\veg_j \in \G(E^i_2)$ and $\veg_j \sqsubseteq \vecc - \vecc'$ for all $j \in [2t-2]$.
Note that because $\|\vecc - \vecc'\|_\infty > R := (2t-2)g_\infty(E^i_2)$, we have that there exists $j \in [2t-2]$ such that $\lambda_j > 1$.
Hence $\veg := \sum_{j=1}^{2t-2} \lfloor\frac{\lambda_j}{2}\rfloor \veg_j$ satisfies $\veg \neq \vezero$.
Let $\bar{\vecc} := \vecc - \veg$, and  $\bar{\vecc}' := \vecc' + \veg$.

First, because $\bar{\vecc} - \bar{\vecc}' = (\vecc - \vecc') + 2\veg = \sum_{j=1}^{2t-2} (\lambda_j - 2\lfloor\frac{\lambda_j}{2}\rfloor) \veg_i$, we may bound $\|\bar{\vecc} - \bar{\vecc}'\|_\infty \leq (2t-2) g_\infty(E^i_2) = R$.
Second, by the conformality of the decomposition, $\bar{\vecc}, \bar{\vecc}' \in \CC^i$.
Third, by separable convex superadditivity (Proposition~\ref{prop:superadditivity}), we have that $f(\vecc) + f(\vecc') \geq f(\bar{\vecc}) + f(\bar{\vecc}')$.
Fourth, there exist a coordinate $j \in [t]$ such that $|c_j - c'_j|=R'$ but, since $\|\bar{\vecc} - \bar{\vecc}'\|_\infty \leq R$, $|\bar{c}_j - \bar{c}'_j| \leq R < R'$ and thus $j$ is no longer a tight coordinate for either $\bar{\vecc}$ or $\bar{\vecc}'$ (or both).
Without loss of generality, let $\lambda_{\vecc} \leq \lambda_{\vecc'}$.
Now let $\velambda' := \velambda$ and set $\lambda'_{\bar{\vecc}}, \lambda'_{\bar{\vecc}'} := \lambda_{\vecc}$, $\lambda'_{\vecc} := 0$, $\lambda'_{\vecc'} := \lambda_{\vecc'} - \lambda_{\vecc}$.
By our arguments above, $\velambda'$ is another optimal solution of~\eqref{eq:fhat_lp} but the weighted number of tight coordinates has decreased by the fourth point, a contradiction.

Thus, there exists a point $\vezeta \in \Z^t$ and an optimal solution $\velambda$ of~\eqref{eq:fhat_lp} such that $\forall \vecc \in \suppo(\velambda)$, $\|\vecc - \vezeta\|_\infty \leq R/2 = (t-1) g_\infty(E^i_2)$.
Obtain a reduced LP from~\eqref{eq:fhat_lp} by deleting all columns $\vecc$ with $\|\vecc - \vezeta\|_\infty > R/2$ and denote the remaining set of columns $\bar{\CC}^i$:
\begin{equation}
\min \sum_{\vecc \in \bar{\CC}^i} \lambda_{\vecc} f^i(\vecc) \quad \text{s.t.} \quad \sum_{\vecc \in \bar{\CC}^i} \lambda_{\vecc} \vecc = \vex^j, \, \|\velambda\|_1 = 1, \, \velambda \geq \vezero \enspace .\label{eq:fhat_lp2}
\end{equation}
This LP is equivalent to one obtained by subtracting $\vezeta$ from all columns and the right hand side:
\begin{equation}
\min \sum_{\vecc \in \bar{\CC}^i} \lambda_{\vecc} f^i(\vecc) \quad \text{s.t.} \quad \sum_{\vecc \in \bar{\CC}^i} \lambda_{\vecc} (\vecc-\vezeta) = (\vex^j-\vezeta), \, \|\velambda\|_1 = 1, \, \velambda \geq \vezero \enspace .\label{eq:fhat_lp3}
\end{equation}
Now, this LP has $t+1$ rows and its columns have the largest coefficient bounded by $R/2$ in absolute value.
A basic solution $\velambda$ has $|\suppo(\velambda)| \leq t+1$ and, by Cramer's rule, the denominator of each $\lambda_\vecc$ is bounded by $(t+1)!$ times the largest coefficient to the power of $t+1$, thus bounded by $(t+1)! R^{t+1} \leq (t+1)! ((2t-2) g_\infty(E^i_2))^{t+1} \leq (t+1)^{(t+1)}(s\|E^i_2\|_\infty +1)^{s(t+2)}$, where we use $g_\infty(E^i_2) \leq \|E^i_2\|_\infty (2s\|E^i_2\|_\infty +1)^s$~\cite[Lemma 2]{EisenbrandEtAl2018}.
\end{proof}
Next, we will need the notion of an Egyptian fraction.
For a rational number $p/q$, $p,q \in \N$, its \emph{Egyptian fraction} is a finite sum of distinct unit fractions such that
\[
\frac{p}{q} = \frac{1}{q_1} + \frac{1}{q_2} + \cdots + \frac{1}{q_k},
\]
for $q_1, \dots, q_k \in \N$ distinct.
Call the number of terms $k$ the \emph{length} of the Egyptian fraction.
Vose~\cite{Vose1985} has proven that any $p/q$ has an Egyptian fraction of length $\Oh(\sqrt{\log q})$.
Since our algorithm requires an exact bound, we present the following weaker yet exact result:
\begin{lemma}[Egyptian Fractions] \label{lem:egypt}
Let $p, q \in \N$, $1 \leq p \leq q$. Then $p/q$ has an Egyptian fraction of length at most $2(\log_2 q)+1$.
\end{lemma}
\begin{proof}
Let $a=2^k$ be largest such that $a < q$, so $k=\ceil*{(\log_2 q)-1} < \log_2 q$.
Write $ap = bq+r$, $0 \leq r < q$.
Note that $p < q \implies b < a$ and $q \leq 2a \implies r < 2a$.
Now let $[b] = (b_{k-1}, \dots, b_1, b_0)$ be the binary representation of $b < a$ so $b=\sum_{i=0}^{k-1} 2^i b_i$ and $[r] = (r_{k-1}, \dots, r_1, r_0)$ be that of $r < 2a$ so $r=\sum_{i=0}^{k} r_i 2^i$.
Then we have
\[
\frac{p}{q} = \frac{ap}{aq} = \frac{bq + r}{aq} = \frac{b}{a} + \frac{1}{q}\frac{r}{a}
= \sum_{i=0}^{k-1} \frac{b_i}{2^{k-i}} + \sum_{i=0}^{k} \frac{r_i}{q \cdot 2^{k-i}},
\]
where $b_i, r_i \in \{0,1\}$, so a sum of at most $2k+1 \leq 2 (\log_2 q)+1$ terms with all denominators $d_i \leq q 2^k = qa \leq q^2$.
Moreover, all denominators in the first sum are distinct and at most $2^k$, and all in the second sum are distinct and at least $q > 2^k$, hence all distinct, so this is an Egyptian fraction of $p/q$ of length $2(\log_2 q)+1$ and denominators of $\Oh(q^2)$.
\end{proof}

Recall that our goal is to obtain a configurable cycle.
However, for that we also need a special form of a decomposition.
Say that $\Gamma$ is \emph{scalable decomposition} of a brick $(\vex^*)^j$ of type $i$ if it is its $\hat{f}$-optimal decomposition, and for each $(\vecc_\gamma, \lambda_\gamma) \in \Gamma$, $\lambda_\gamma$ is of the form $1/q_{\gamma}$ for some $q_{\gamma} \in \N$.
We say that $|\Gamma|$ is the \emph{size} of the decomposition.
We note that in what follows we do not need an algorithm computing a scalable decomposition, only the following existence statement.
\begin{lemma} \label{lem:scalable_decomposition}
Each brick of $\vex^*$ has a scalable decomposition of size at most $26t^3 \log(t\|E_2\|_\infty)$.
\end{lemma}
\begin{proof}
Fix $j \in [N]$.
Let $\vex = (\vex^*)^j$ be a brick of $\vex^*$ of type $i$.
By Lemma~\ref{lem:balanced_decomposition}, there exists an $\hat{f}$-optimal decomposition of $\vex$ of size $t+1$ where each coefficient $\lambda_\vecc=p_\vecc/q_\vecc$ satisfies $p_\vecc,q_\vecc \leq (t+1)^{(t+1)}(s\|E^i_2\|_\infty +1)^{(s+1)(t+2)}$.
For each $\vecc$ in the decomposition now express $\lambda_\vecc$ as an Egyptian fraction:
\[\lambda_\vecc = \frac{p_\vecc}{q_\vecc} = \frac{1}{a_1} + \frac{1}{a_2} + \cdots + \frac{1}{a_{\mathfrak{e}}} \enspace .\]
By Lemma~\ref{lem:egypt}, $p/q$ has an Egyptian fraction of length
\[\mathfrak{e} \leq 2(\log_2 q_{\vecc})+1 = 2\left(\log \left((t+1)^{(t+1)}(s\|E^i_2\|_\infty +1)^{(s+1)(t+2)}\right)\right)+1 \leq 25st \log (st\|E^i_2\|_\infty) \enspace .\]
Thus the resulting decomposition is of size at most $(t+1) 25st \log (st\|E^i_2\|_\infty) \leq 26t^3 \log(t\|E^i_2\|_\infty)$ (by $s \leq t$ this justifies the deletion of $s$ in the $\log()$ so the last bound holds) and is scalable, since each coefficient is of the form $1/q_{\gamma}$ for some $q_{\gamma} \in \N$.
\end{proof}

We will now show that we are guaranteed a configurable cycle of $\vex^* - \vez^*$ if there exists an analogue of a regular cycle of a certain ``lifting'' of $\vex^*$ and $\vez^*$.

Fix for each brick of $\vex^*$ a scalable decomposition $\Gamma^j$.
Let $\uparrow\vex^*$ be the \emph{rise of $\vex^*$} defined as a vector obtained from $\vex^*$ by keeping every integer brick $(\vex^*)^j$, and replacing every fractional brick $(\vex^*)^j$ with $|\Gamma^j|$ terms $\lambda_\gamma \vecc_\gamma$, one for each $(\vecc_\gamma, \lambda_\gamma) \in \Gamma^j$.
Observe that each brick of $\uparrow \vex^*$ is of the form $\lambda_{\vecc} \vecc$ for some configuration $\vecc$ and some coefficient $0 \leq \lambda_{\vecc} \leq 1$.
Thus for a brick $\lambda_{\vecc} \vecc$ we say that $\vecc$ is its configuration, $\lambda_{\vecc}$ is its coefficient, and its type is identical to the type of brick it originated from; in particular, bricks which originated from an integer brick $\vep = (\vex^*)^j$ are of the form $\lambda_{\vep} \vep$ with $\lambda_{\vep} = 1$.
Let $N'$ be the number of bricks of $\uparrow \vex^*$ and define a mapping $\nu: [N'] \to [N]$ such that if a brick $j \in [N']$ of $\uparrow \vex^*$ was defined from brick $\ell \in [N]$ of $\vex^*$, then $\nu(j) = \ell$.
The natural inverse $\nu^{-1}$ is defined such that, for $\ell \in [N]$, $\nu^{-1}(\ell)$ is the set of bricks of $\uparrow \vex^*$ which originated from $(\vex^*)^\ell$.

\begin{lemma} \label{lem:size_rise}
The vector $\uparrow \vex^*$ has at most $(r+\tau) \cdot 26t^3 \log(t\|E^1_2,\dots,E^\tau_2\|_\infty)$ fractional bricks.
\end{lemma}
\begin{proof}
By Lemma~\ref{lem:frac} there is a conf-optimal $\vex^*$ with at most $r+\tau$ fractional bricks.
By Lemma~\ref{lem:scalable_decomposition} for each fractional brick of $\vex^*$ of type $i$ there is a scalable decomposition of size at most $26t^3 \log(t\|E^i_2\|_\infty) \leq 26t^3 \log(t\|E^1_2,\dots,E^\tau_2\|_\infty)$.
Thus $\uparrow \vex^*$ has at most $26t^3 \log(t\|E^1_2,\dots,E^\tau_2\|_\infty)$ fractional bricks for each fractional brick of $\vex^*$, of which there are at most $r+\tau$, totaling $(r+\tau)\cdot 26t^3 \log(t\|E^1_2,\dots,E^\tau_2\|_\infty)$ fractional bricks.
\end{proof}

Then, denote by $\uparrow\vez^* \in \R^{N't}$ the \emph{rise of $\vez^*$} (with respect to $\vex^*$) defined as follows.
Let $j \in [N']$, $\ell = \nu(j)$, and $\lambda$ be the coefficient of the $j$-th brick of $\uparrow \vex^*$.
Then the $j$-th brick of $\uparrow \vez^*$ is $(\uparrow \vez^*)^j := \lambda (\vez^*)^{\ell}$.
Observe that $\|\uparrow \vex^* - \uparrow \vez^*\|_1 \geq \|\vex^* - \vez^*\|_1$ by triangle inequality\mkcom{elaborate: decompose for one brick and say the rest is by aggregation}.

For any vector $\vex \in \R^{N't}$, define the \emph{fall of $\vex$} as a vector $\downarrow \vex \in \R^{Nt}$ such that for $\ell \in [N]$, $(\downarrow \vex)^\ell = \sum_{j \in \nu^{-1}(\ell)} \vex^j$.
We see that $\downarrow (\uparrow \vex^*) = \vex^*$ and $\downarrow (\uparrow \vez^*) = \vez^*$.
Say that $\ver$ is a \emph{cycle of $\uparrow \vex^* - \uparrow \vez^*$} if $\ver \sqsubseteq \uparrow \vex^* - \uparrow \vez^*$ and $\ver \in \Ker_{\Z}(E^{(N')})$.

\begin{lemma}
If $\ver$ is a cycle of $\uparrow \vex^* - \uparrow \vez^*$, then $\downarrow \ver$ is a configurable cycle of $\vex^* - \vez^*$.
\end{lemma}
\begin{proof}
To show that $\downarrow \ver$ is a configurable cycle we need to show that $\downarrow \ver \in \Ker_{\Z}(E^{(N)})$ and, for each brick $\vex$ of $\vex^*$, there is an $\hat{f}$-optimal decomposition of $\vex$ such that $\veh = (\downarrow \ver)^j$ decomposes accordingly.
For the first part, $\downarrow \ver$ is integral because it is obtained by summing bricks of $\ver$, which is integral.
Denote by $i(j)$ the type of a brick $j$.
By the fact that $\ver \in \Ker_{\Z}(E^{(N')})$ and the definition of $\downarrow \ver$, we have $\vezero =  \sum_{j=1}^{N'} E^{i(j)}_1 \ver^j = \sum_{j=1}^{N} E^{i(j)}_1 (\downarrow \ver)^j$, and, for each $\ell \in [N]$, $\vezero = \sum_{j \in \nu^{-1}(\ell)} E^{i(j)}_2 \ver^j = E^{i(\ell)}_2(\downarrow \ver)^\ell$, thus $\downarrow \ver \in \Ker_{\Z}(E^{(N)})$.

To see the second part, fix a brick $j \in [N]$ of type $i$ and let $\vex = (\vex^*)^j$, $\vez = (\vez^*)^j$ and $\veh = (\downarrow \ver)^j$.
We need to show that $\veh = \sum_{\gamma \in \nu^{-1}(j)} \veh_\gamma$ can be written as $\sum_{\vecc \in \CC^i} \lambda_\vecc \veh_\vecc$ with $\veh_\vecc \sqsubseteq \vecc - \vez$ and $\veh_\vecc \in \Ker_\Z(E^i_2)$.
By definition of $\uparrow \vex$ and $\ver$, there is a scalable decomposition $\Gamma$ of $\vex$ such that for each $\gamma \in \nu^{-1}(j)$, $\veh_\gamma \sqsubseteq \lambda_\gamma (\vecc_\gamma - \vez)$ and $\veh_\gamma \in \Ker_\Z(E^i_2)$.
Thus we may write $\veh = \sum_{\gamma \in \nu^{-1}(j)} \lambda_\gamma \cdot (\lambda^{-1}_\gamma \veh_\gamma)$ with $\lambda^{-1}_\gamma \veh_\gamma \sqsubseteq \vecc_\gamma - \vez$ and $\lambda^{-1}_\gamma \veh_\gamma$ integral by the fact that $\lambda_\gamma = 1/q_{\gamma}$ with $q_{\gamma} \in \N$, concluding the proof.
\end{proof}
We are finally ready to use the Steinitz Lemma to derive a bound on $\|\vex^* - \vez^*\|_1$.
\begin{theorem}
\label{thm:proximity}
  Let $\vex^*$ be a conf-optimal solution of~\eqref{eq:hugenfold_relax} with at most $r+\tau$ fractional bricks.
  Then there exists an optimal solution $\vez^*$ of~\eqref{eq:hugenfold} such that
  \begin{equation*}
    \|\vez^*-\vex^*\|_1 \leq \left((r+\tau) 26t^4 \log(t\|E^1_2,\dots,E^\tau_2\|_\infty)\right)  (2r)^{r+1} (\|E\|_\infty s)^{3rs} \enspace .
  \end{equation*}
\end{theorem}
\begin{proof}
  Denote by $\bar{E}_1$ the first $r$ rows of the matrix $E^{(N)}$.
  Let $\vez^*$ be an optimal integer solution such that $\|\vez^* - \vex^*\|_1$ is minimal, let $\uparrow \vex^*$ be the rise of $\vex^*$ with at most $(r+\tau) \cdot 26t^3 \log(t\|E^1_2,\dots,E^\tau_2\|_\infty)$ fractional bricks (cf. Lemma~\ref{lem:size_rise}), let $\uparrow \vez^*$ be a rise of $\vez^*$, and let $\veq = \uparrow \vex^* - \uparrow \vez^*$.

  We want to get into the setting of the Steinitz Lemma, that is, to obtain a sequence of vectors with small $\ell_1$-norm and summing up to zero.
  To this end, we shall decompose $\bar{E}_1\veq$ in the following way; we stress that we have $\bar{E}_1\veq = \vezero$.
  For every integral brick $\veq^i$ of type $\ell \in [\tau]$ we have its decomposition $\veq^i = \sum_j \veg^i_j$ into elements of $\G(E^\ell_2)$ by the Positive Sum Property (Proposition~\ref{prop:possum}); for each $\veg^i_j$ append $E^\ell_1\veg^i_j$ into the sequence.
  For every fractional brick $\veq^i$ of type $\ell \in [\tau]$ we have its decomposition $\veq^i = \sum_{j=1}^{t} \alpha_j \veg^i_j$, $\alpha_j \geq 0$ for each~$j$, into elements of $\CC(E^\ell_2)$; for each $\veg^i_j$ append $\floor{\alpha_j}$ copies of $E^\ell_1 \veg^i_j$ into the sequence, and finally append $E^\ell_1 \{\alpha_j\} \veg^i_j$.
  Observe that since $\uparrow \vex^*$ has at most $(r+\tau) \cdot 26t^3 \log(t\|E^1_2,\dots,E^\tau_2\|_\infty)$ fractional bricks, so does $\veq$, and thus we have appended  $\mathfrak{f} \leq t \cdot (r+\tau)26t^3 \log(t\|E^1_2,\dots,E^\tau_2\|_\infty) \leq (r+\tau) 26t^4 \log(t\|E^1_2,\dots,E^\tau_2\|_\infty)$ fractional vectors into the sequence.
  Now we have a sequence
  \begin{equation}
  \label{eq:13}
    \veo_1,\dots,\veo_m,\vep_{m+1},\dots,\vep_{m+\mathfrak{f}}
  \end{equation}
  with $m$ integer vectors $\veo_1, \dots, \veo_m$ and $\mathfrak{f}$ fractional vectors $\vep_{m+1}, \dots, \vep_{m+\mathfrak{f}}$.
  Moreover, since, for each $i \in [\tau]$, $\CC(E^i_2) \subseteq \G(E^i_2)$ and $\max_{\veg \in \G(E^i_2)} \|\veg\|_\infty \leq (2s\|E^i_2\|_\infty +1)^{s}$~\cite[Lemma 2]{EisenbrandEtAl2018}, each vector has $\ell_\infty$-norm of $\|E^1_1,\dots,E^\tau_1\|_\infty \cdot (2s\|E^1_2,\dots,E^\tau_2\|_\infty +1)^{s} \leq (2s\|E\|_\infty +1)^{s+1} $ and they sum up to $\vezero$. Observe that  $(m+\mathfrak{f})\cdot \max_{\ell \in [\tau]}g_\infty(E^\ell_2) \geq \|\veq\|_1 = \|\uparrow \vex^* - \uparrow \vez^*\|_1 \geq \|\vex^* - \vez^*\|_1$.
   We now focus on bounding $m+\mathfrak{f}$.
  The Steinitz Lemma (Lemma~\ref{lem:steinitz}) implies that there exists a permutation $\pi$ such that the sequence~\eqref{eq:13} can be re-arranged as
  \begin{equation}
  \label{eq:16}
    \vev_{1},\dots,\vev_{m+\mathfrak{f}},
  \end{equation}
  where $\vev_i$ is $\veo_{\pi^{-1}(i)}$ if $i \in [1,m]$ and $\vep_{\pi^{-1}(i)}$ if $i \in [m+1, m+\mathfrak{f}]$, respectively, and for each $1 \le k \le m+\mathfrak{f}$ the prefix sum $\vet_k = \sum_{i=1}^k \vev_{i}$ satisfies
  \begin{equation*}
    \|\vet_k\|_\infty \le r (2s\|E\|_\infty +1)^{s+1} \enspace .
  \end{equation*}
  We will now argue that there cannot be indices $1 \le k_1 < \cdots < k_{\mathfrak{f}+2} \le \mathfrak{f}+m$ with
  \begin{equation}
  \label{eq:19}
    \vet_{k_1} = \cdots = \vet_{k_{\mathfrak{f}+2}},
  \end{equation}
  which implies that $\mathfrak{f}+m$ is bounded by $\mathfrak{f}+1$ times $r$ times the number of integer points of norm at most $r (2s\|E\|_\infty +1)^{s+1}$ and therefore, denoting $g_\infty(E_2) = \max_{i \in [\tau]} g_\infty(E^i_2)$,
  \begin{align*}
    \|\vex^* - \vez^*\|_1 &\leq \|\uparrow \vex^* - \uparrow \vez^*\|_1 \leq (\mathfrak{f}+1) \cdot r \left(2r (2s\|E\|_\infty +1)^{s+1} +1\right)^r \cdot g_\infty(E_2) \\
    &\leq \left((r+\tau) 26t^4 \log(t\|E^1_2,\dots,E^\tau_2\|_\infty)\right)  \cdot r \left(2r (2s\|E\|_\infty +1)^{s+1} +1\right)^r\cdot (2s\|E\|_\infty + 1)^{s+1} \\
    &\leq \left((r+\tau) 26t^4 \log(t\|E^1_2,\dots,E^\tau_2\|_\infty)\right)  (2r)^{r+1} (\|E\|_\infty s)^{3rs} \enspace .
  \end{align*}
  Assume for contradiction that there exist $\mathfrak{f}+2$ indices $1 \le k_1 < \cdots < k_{\mathfrak{f}+2} \le \mathfrak{f}+m$ satisfying~\eqref{eq:19}.
  By the pigeonhole principle there must exists an index $k_\ell$ such that all the vectors
  \begin{displaymath}
    \vev_{k_{\ell}+1},\dots,\vev_{k_{\ell+1}}
  \end{displaymath}
  from the rearrangement~\eqref{eq:16} correspond to integer vectors $\veo_{\pi^{-1}(p)}$ for $p \in [k_{\ell}+1, k_{\ell+1}]$.
  We will show that this corresponds to a cycle $\veh$ of $\uparrow \vex^* - \uparrow \vez^*$ which by the minimality of $\|\vex^* - \vez^*\|_1$ and Lemma~\ref{thr:4} is impossible.
  To obtain the cycle, for each $p \in [k_{\ell}+1, k_{\ell+1}]$, let $i(p)$, $j(p)$, and $\ell(p)$ be such that $\veo_{\pi^{-1}(p)} = E^{\ell(p)}_1 \veg_{j(p)}^{i(p)}$.
  Initialize $\veh := \vezero \in \Z^{N't}$ and, for each $p \in [k_{\ell}+1, k_{\ell+1}]$, let $\veh^{i(p)} := \veh^{i(p)} + g_{j(p)}^{i(p)}$.
  Now we check that $\veh$ is, in fact, a cycle.
  First, to see that $E^{(N')} \veh = \vezero$, we have $E^\ell_2 \veh^i = \vezero$ for every brick $i \in [N']$ of type $\ell$ by the fact that $\veh^i$ is a sum of $\veg_j^i \in \G(E^\ell_2) \subseteq \Ker_{\Z}(E^\ell_2)$, and we have $\bar{E}_1 \veh = \vezero$ by the fact that $\vet_{k_\ell} = \vet_{k_{\ell+1}}$ and thus $\sum_{p \in [m+\mathfrak{f}]} E^{\ell(p)}_1 \veg_{j(p)}^{i(p)} =\vezero$.
  Second, $\veh \sqsubseteq \veq$ because, for every brick $i \in [N']$, $\veh^i$ is a sign-compatible sum of elements $\veg^i_j \sqsubseteq \veq^i$.
\end{proof}

\subsection{Algorithm}
\begin{proof}[{Proof of Theorem~\ref{thm:hugenfold}}]
We first give a description of the algorithm which solves huge $N$-fold IP, then show its correctness, and finally give a time complexity analysis.
\paragraph{Description of the Algorithm.}
  First, obtain an optimal solution $\vey$ of~\eqref{eq:conflp} and from it a conf-optimal solution $\vex^* = \varphi(\vey)$ with at most $r+\tau$ fractional bricks by Lemma~\ref{lem:frac}.
  Applying Theorem~\ref{thm:proximity} to $\vex^*$ guarantees the existence of an integer optimum $\vez^*$ satisfying
  \begin{equation} \label{eq:algo_prox}
  \|\vex^* - \vez^* \|_1 \leq P := \left((r+\tau) 26t^4 \log(t\|E^1_2,\dots,E^\tau_2\|_\infty)\right)  (2r)^{r+1} (\|E\|_\infty s)^{3rs}  \enspace .
  \end{equation}
  This implies that $\vez^*$ differs from $\vex^*$ in at most $P$ bricks.
  The idea of the algorithm is to ``fix'' the value of the solution on ``almost all'' bricks and compute the rest using an auxiliary $\bar{N}$-fold IP problem with a polynomial $\bar{N}$.

  Formally, our goal is to compute an optimal solution $\vez$ of~\eqref{eq:hugenfold} represented succinctly by multiplicities of configurations, or in other words, as a solution $\vezeta$ of~\eqref{eq:confilp}.
  Denote by $\vey_{-P}$ the vector whose coordinates are defined by setting, for every type $i \in [\tau]$ and every configuration $\vecc \in \CC^i$, $\vey_{-P}(i,{\vecc}) = \max\{0, \floor{y(i, \vecc)} - P\}$
  This leaves us with $\|\vey\|_1 - \|\vey_{-P}\|_1 \leq |\suppo(\vey)| P \leq (r+\tau) P =: \bar{P}$ bricks to determine.
  Let $\bar{\vezeta} = \vey - \vey_{-P}$, $\bar{\vex} = \varphi(\bar{\vezeta})$, and $\bar{N} = \|\bar{\vezeta}\|_1$.
  Construct an auxiliary $\bar{N}$-fold IP instance with the same blocks $E^i_1, E^i_2$, $i \in [\tau]$, by, for each brick $\bar{\vex}^j$ of type $i$, setting
  \begin{tasks}[style=itemize](4)
    \task $\bar{f}^j = f^i$,
    \task $\bar{\veb}^j = \veb^i$,
    \task $\bar{\vel}^j = \vel^i$,
    \task $\bar{\veu}^j = \veu^i$.
  \end{tasks}
  We say that such a brick was \emph{derived from type $i$}.
  Let $\bar{\veb}^0 = \veb^0 - \sum_{i=1}^{\tau} \sum_{\vecc \in \CC^i} \zeta(i, \vecc) E^i_1 \vecc$.

  After obtaining an optimal solution $\bar{\vez}$ of this instance we update $\vezeta$ as follows.
  For each brick $\bar{\vez}^j$ derived from type $i$, increment $\zeta(i,\bar{\vez}^j)$ by one.

\paragraph{Correctness.}
By~\eqref{eq:algo_prox} it is correct to assume that there exists a solution $\vezeta$ of~\eqref{eq:confilp} which has $\zeta(i,\vecc) \geq \max\{0, \floor{y(i, \vecc)} - P\}$ for each $i \in [\tau]$ and $\vecc \in \CC^i$.
Thus we may do a variable transformation of~\eqref{eq:confilp} $\vezeta = \bar{\vezeta} + \vey_{-P}$, obtaining an auxiliary~\eqref{eq:confilp} instance
\[
\min \vev (\bar{\vezeta} + \vey_{-P}) \,:\, B(\bar{\vezeta} + \vey_{-P}) = \ved,\, \vezero \leq \bar{\vezeta} \enspace .
\]
The auxiliary huge $\bar{N}$-fold instance is simply the instance corresponding to the above, and the final construction of $\vezeta$ corresponds to the described variable transformation.

\paragraph{Complexity.}
Since $\|\bar{\vezeta}\|_1 \leq \bar{P}$, we can obtain an optimal solution $\bar{\vez}$ of the auxiliary instance in time $(\|E\|_\infty r s)^{\Oh(r^2s + rs^2)} (t \bar{P}) \log (t \bar{P}) \la f_{\max}, \bar{\veb}, \bar{\vel}, \bar{\veu} \ra$~\cite[Corollary 91]{EisenbrandEtAl2019}.
Let us now compute the time needed altogether.
To solve~\eqref{eq:conflp}, we need time $$\|E\|_\infty^{\Oh(s^2)}(rt\tau \la f_{\max}, \vel, \veu, \veb, \vemu \ra)^{\Oh(1)} \enspace .$$
To solve the auxiliary instance above, we need time
\begin{equation*}
(\|E\|_\infty r s)^{\Oh(r^2s + rs^2)} (t \bar{P}) \log (t \bar{P}) \la f_{\max}, \bar{\veb}, \bar{\vel}, \bar{\veu} \ra, \quad \text{where,}
\end{equation*}
\begin{equation*}
\bar{P} = (r+\tau)P = (r+\tau)\left((r+\tau) 26t^4 \log(t\|E^1_2,\dots,E^\tau_2\|_\infty)\right)  (2r)^{r+1} (\|E\|_\infty s)^{3rs}  \enspace .
\end{equation*}
Hence we can solve huge $N$-fold IP in time at most
\[ (\|E\|_\infty rs)^{\Oh(r^2s + rs^2)} (t\tau \la f_{\max}, \vel, \veu, \veb, \vemu \ra)^{\Oh(1)} \enspace .\qedhere\]
\end{proof}

\section{Part~\ref{thm:implicitMIMO:gr} of Theorem~\ref{thm:implicitMIMO}} \label{sec:part4}
\begin{proof}[Proof idea for Part~\ref{thm:implicitMIMO:gr} of Theorem~\ref{thm:implicitMIMO}]
  Our proof builds on a Structure Theorem of Goemans and Rothvo{\ss}~(Proposition~\ref{prop:structurethm}) and the idea of the proof of their main theorem~\cite[Theorem 2.2]{GoemansRothvoss2014}.
  The Structure Theorem applies to the single-type setting and says (translated into the setting of MIMO) that for any solution $\velambda$ corresponding to a decomposition of $\ven$, there exists a solution $\hat{\velambda}$ whose support mostly lies within a precomputable and not-too-large set $Y$ of ``important'' configurations.

  We first extend the Structure Theorem into the multitype setting (Lemma~\ref{lem:structure}), and then use it as follows.
  For each type $i$, we compute the set of ``important'' configurations $Y^i$, and then guess from it a small subset of configurations which will appear in the solution.
  Using this, we construct an ILP in small dimension, solve it using Kannan's algorithm, and derive from it an optimal solution~$\velambda$.
  We take special care to enforce the multiplicity constraint (i.e., $\|\velambda^i\|_1 = \mu^i$, for each $i \in [\tau]$) and argue how to encode a linear and a fixed-charge objective.
\end{proof}

\noindent\textbf{Remark:} Goemans and Rothvo{\ss} prove a similar statement~\cite[Corollary 5.1]{GoemansRothvoss2014} to Part~\ref{thm:implicitMIMO:gr} of Theorem~\ref{thm:implicitMIMO}, where the input $\ven$ and the coefficients $\vew$ have to be given in unary if one desires an \FPT algorithm, whereas in our case they can be given in binary.
The difference is that they invoke the Structure Theorem on a polytope $P$ which is a disjunctive formulation of the union of polyhedra $P^1 \cup \cdots \cup P^\tau$.
This disjunctive construction however introduces a large coefficient, increasing $\Delta$.
Similarly, a linear objective could be handled in their setting by introducing an extra variable $x_{d+1}$ and setting $x_{d+1} = \vew \vex$, but this constraint would again increase $\Delta$.
We circumvent both of these limitations by using the Structure Theorem directly.

\clearpage
\section{Applications: Scheduling, Bin Packing, and Surfing} \label{sec:applications}
In this section we present an extended exposition of using MIMO as a modeling tool.
The majority of our focus is on the setting of high multiplicity non-preemptive scheduling (Sections~\ref{sec:nonpreempt} and~\ref{sec:polyobj}), where we begin with structural observations and gradually extend them to increasingly complex scenarios and objective functions.
Each subsection culminates with a ``modeling lemma'' which links the parameters of a MIMO instance we have constructed with the parameters of the problem instance it encodes.
A straightforward application of Theorem~\ref{thm:implicitMIMO} then gives an ``effective theorem'', stating the thus obtained \FPT algorithms.

\paragraph{Used Techniques.}
In order to provide MIMO models for scheduling problems with release times and due dates we study the structural properties of such schedules.
A notion of a schedule cycle has been introduced by Goemans and Rothvo\ss{}~\cite{GoemansRothvoss2014}.
Since parameterized scheduling algorithms have not been their focus, their structural observation about scheduling cycles is relatively basic.
We prove stronger structural results which allow us for example to reduce the largest coefficient in our models.
Moreover, the fact that part~\ref{thm:implicitMIMO:huge} of Theorem~\ref{thm:implicitMIMO} applies to MIMO models with certain convex objectives allows us to express more complicated scheduling objectives such as $\sum w_j C_j$.

\subsection{Makespan Minimization and Related Objectives} \label{sec:nonpreempt}
We begin by observing the structure of a schedule on a single machine (Lemma~\ref{lem:cyclesScheduleCMax}).
This allows us to restrict our attention to so-called regular schedules which can be decomposed into a small number of schedule cycles.
We then study the set $\CC$ of all potential cycles of a schedule and provide some basic observations about them.
We guess the correct value $\bar{C}_{\max}$ of the objective and ``trim'' the instance accordingly (i.e., no due date is after $\bar{C}_{\max}$).
Using the above we provide a MIMO model describing an assignment of jobs to cycles.
To connect the feasible solutions of this MIMO model to feasible schedules, we provide an algorithm which, given a solution to our model, yields an admissible schedule with the same makespan.
Finally, we discuss a different representation of the scheduling instance in which machines have speeds and we alter the previous model to capture this feature.
This is done via ``time scaling'' which only affects the right hand sides of our model, meaning the coefficients are not increased even though the model now potentially encodes jobs with large job  size (corresponding to jobs on slow machines).
This, in turn, allows us to show tractability of a wider range of instances.

We first describe the idea for the problem $R|r^i_j,d^i_j|C_{\max}$ without speeds.
Let us define the problem:

\prob{\textsc{Makespan Minimization on Unrelated Machines} ($R | r^i_j,d^i_j | C_{\max}$)}
{
There are $\kappa$ kinds of machines and $d$ types of jobs.
The number of machines of kind $i \in [\kappa]$ is $\mu^i$ and the number of jobs of type $j \in [d]$ is $n_j$, with $\vemu = (\mu^1, \dots, \mu^\kappa)$ and $\ven = (n_1, \dots, n_d)$, hence there are $m=\|\vemu\|_1$ machines and $n = \|\ven\|_1$ jobs.
Each job type is specified by three vectors giving its size, release time, and due date on each machine kind, i.e., for each $j \in [d]$ given are vectors $\vep_j = \left( p^1_j, \ldots, p^\kappa_j \right) \in \left(\N \cup \{\infty\}\right)^\kappa$, $\ver_j = \left( r^1_j, \ldots, r^\kappa_j \right) \in \N^\kappa$, and $\ved_j = \left( d^1_j, \ldots, d^\kappa_j \right) \in \N^\kappa$.
}
{A non-preemptive schedule of all of the jobs on the specified $m$ machines (if one exists) minimizing the time when the last job finishes (i.e., the makespan) such that processing a job of type $j$ on a machine of kind $i$ does not start prior to $r^i_j$ and finishes no later than $d^i_j$.}

Recall that $0 \in \N$, so, for example, a release time $0$ is allowed.
We use $J$ to denote an individual job and use $r^i(J)$ to denote the release time of $J$ on a machine of kind $i$, that is, $r^i(J) = r^i_j$ if $J$ is of type $j$, and define $d^i(J)$ similarly.
The set of all jobs is denoted $\mathcal{J}$.
A \emph{schedule of $\mathcal{J}$} is a mapping $\sigma$ that to each job $J \in \mathcal{J}$ assigns a machine and a time interval of size $p^i(J)$ if the machine is of kind $i \in [\kappa]$, which satisfies the following conditions.
Let $\lambda(J)$ be the left end point (i.e., the start) of an interval assigned by $\sigma$ to $J$ and let $\rho(J)$ be its right end point.
Then $\sigma$ must satisfy
\begin{itemize}
  \item $r^i(J) \le \lambda(J) < \rho(J) \le d^i(J)$ and $\rho(J) = \lambda(J) + p^i(J)$ for each job $J \in \mathcal{J}$ if $J$ is scheduled by $\sigma$ to a machine of kind $i$, and
  \item for every machine, the interiors of intervals corresponding to jobs assigned to it do not overlap.
\end{itemize}

The number $d^i_j$ is called a due date or a deadline in the literature.
The distinction usually is that in the former case, a job can be scheduled after the due date, but this incurs a penalty, while in the latter case a job can never be scheduled after a deadline.
Because we deal with both scenarios but the distinction is clear from which scheduling objective is optimized, we choose to always use the term due date.

We assume $p^i_j < \infty$ for all $i \in [\kappa]$ and all $j \in [d]$ by the following argument.
Suppose there is a job type $j$ with $p^i_j = \infty$ for some machine kind $i$.
This means no job of type $j$ can be scheduled to run on a machine of kind $i$.
We alter the given instance by setting $p^i_j = 1$ and $r^i_j = d^i_j = r^i_{j'}$, where $j' \in [d]$ is a job type with $p^i_{j'} < \infty$ (note that such a job type exists, since otherwise we may omit the machine kind $i$ completely as no job can be scheduled on any of these machines).
Thus, from now on we suppose $p^i_j \in \N$, in particular, $p_{\max} = \max_{i \in [\kappa]} \max_{j \in [d]} p^i_j$ is finite and well defined.

A \emph{cycle $C$} is a sequence of jobs for which there exists a permutation of job types $\pi\colon [d] \to [d]$ such that in $C$ there are first jobs of type $\pi(1)$, then jobs of type~$\pi(2)$, and so forth, up to type~$\pi(d)$ (with some of these subsequences of jobs of type $\pi(j)$ possibly empty).
As long as we discuss the $C_{\max}$ objective, the permutation $\pi$ plays no role and we may assume that $\pi(j)=j$ for each $j \in [d]$.
We will later see what role the permutation $\pi$ plays in other objectives.
Furthermore, we require the jobs assigned to one cycle to be executed one after another, in particular, there is no idle time on the machine during the execution of (the job set of) a cycle.
We stress that the number of jobs of any type in a cycle can be~0, in fact, a cycle does not have to contain any jobs at all.
The significance of a cycle is that on a machine of kind $i$ only a cycle entirely contained in $(r^i_j, d^i_j)$ may contain a job of type $j$.
A \emph{cycle decomposition $\DD$ of a schedule $\sigma$} is a partition of the jobs of $\sigma$ into cycles.
This decomposition is typically not unique and our main structural result guarantees the existence of a cycle decomposition with several useful properties.
Goemans and Rothvo{\ss} showed that any schedule admits a cycle decomposition with at most $4d$ cycles (observe that a cycle decomposition with at most $n$ cycles is trivial by having a cycle for each job).

From now on we fix a machine kind $i \in [\kappa]$.
We will now reason about a schedule on a single machine of kind $i$ with the goal of describing the set of configurations of jobs on this machine using linear constraints, hence giving a description of a polytope $P^i$, which will then be used to construct a MIMO instance.
We define the set $T = \left\{ r^i_j, d^i_j \mid j \in [d] \right\}$ of \emph{critical times}\footnote{Because $T$ depends on the machine kind $i$, it would be more precise to call it $T^i$, and similarly for other objects we shall define. We omit the superscript for brevity.}; note that $|T| \le 2d$.
Let $C$ be a cycle in any cycle decomposition of any schedule of jobs on this machine.
By $\lambda(C)$ we denote the left end of $C$, that is, the time when the first job of $C$ starts being processed and by $\rho(C)$ we denote the completion time of the last job in $C$; we say $\rho(C)$ is the completion time of $C$.
Note that $\rho(C) - \lambda(C) = \sum_{J \in C} p^i(J)$.
A cycle $C$ is \emph{internal} if there is no critical time in the interval $\left( \lambda(C), \rho(C) \right)$ (i.e., $T \cap \left(\lambda(C), \rho(C)\right) = \emptyset$) and is \emph{external} otherwise.
For a single machine of kind $i \in [\kappa]$, a vector $\vex \in \N^d$ defines a scheduling instance with $x_j$ jobs of type $j$, for $j \in [d]$.
For short, we say that a schedule $\sigma$ of this instance is a schedule of $\vex$.
We index cycles according to their starting times, e.g., the ``first'' cycle is the one containing the earliest scheduled job, and saying a cycle is odd or even refers to it having an odd or even index, respectively.
A \emph{gap} in a schedule $\sigma$ is a maximal time interval $I$ in $\sigma$ such that no job starts nor ends in $I$.

We index the critical times non-decreasingly as $T = \left\{ t_1, \ldots, t_{|T|} \right\}$ with $t_k < t_{k+1}$ for all $1 \le k \le |T|-1$.

\begin{definition}[$S$-regular decomposition, $S$-regular schedule]
Let $\sigma$ be a schedule, $\DD$ be its decomposition, and $S \subseteq \R_{\geq 0}$.
We say that $\DD$ is a \emph{regular decomposition} if it contains at most $4d-3$ cycles, the interval $[t_\ell, t_{\ell+1}]$ contains at most one internal cycle for each $\ell \in [|T|-1]$, and every external cycle contains at most one job.
We say $\DD$ is \emph{$S$-regular} if it is regular and $\forall C \in \DD$, $\lambda(C) \in S$.
We say that $\sigma$ is \emph{regular} or \emph{$S$-regular schedule} if it has a regular or $S$-regular decomposition, respectively.
\end{definition}

\begin{lemma}\label{lem:cyclesScheduleCMax}
	Let $\vex \in \N^d$ and $\sigma'$ be a schedule of $\vex$.
	Then there exists an $\N$-regular schedule $\sigma$ of $\vex$ with not larger completion time as $\sigma'$.
\end{lemma}
\begin{proof}
  Let $\DD'$ be any cycle decomposition of $\sigma'$ (e.g., each job in a separate cycle).
  We will transform $\sigma'$ and $\DD'$ into $\sigma$ and its $\N$-regular decomposition $\DD$ in several steps.
  First, we describe how to ensure that
  \begin{itemize}
    \item for every two consecutive critical times $t_k, t_{k+1}$ there is at least one (possibly empty) cycle in $(t_k, t_{k+1})$, or
    \item there exists an external cycle containing both $t_k$ and $t_{k+1}$.
  \end{itemize}
  Clearly, if $t_k$ is contained in a different cycle than $t_{k+1}$, then we may insert an empty internal cycle at the completion time of the cycle containing $t_k$.
  Furthermore, by splitting cycles we may assume that external cycles are either empty or contain exactly one job.

  Observe that now there are at most $|T| - 2$ external cycles, because an external cycle must contain a critical time other than $t_1$ or $t_{|T|}$ in its interior, and there are $|T|-2$ of these.

  Regarding internal cycles, note that any two consecutive cycles $C^1, C^2$ in a time interval $(t_k, t_{k+1})$ for some $k \in [|T|-1]$ can be merged by permuting the jobs in $C^1 \cup C^2$.
  This merging is possible as all jobs in $C^1 \cup C^2$ were released at time at most $t_k$ and have due dates of at least $t_{k+1}$, and the total size of jobs does not increase as a result of this permuting.
  There are $|T|-1$ intervals $(t_k, t_{k+1})$ with $k \in [|T|-1]$, hence at most $|T|-1$ internal cycles, hence $(|T|-2) + (|T|-1)$ cycles in total, and by $|T| \leq 2d$ the bound of $4d-3$ follows.

  Finally, we ensure that for each cycle $C \in \DD'$, $\lambda(C) \in \N$, i.e., $\N$-regularity.
  Say that a $t \in T$ is \emph{permissible} for $C$ if no job in $C$ has a release time larger than $t$.
  Now simply repeatedly pick $C \in \DD'$ with smallest $\lambda(C)$ such that $\lambda(C) \not\in \N$ and shift it to the left such that $C$ starts at the larger of either the closest smaller $\rho(C')$ for $C' \in \DD'$, or the smallest permissible $t \in T$.
  This shifting must be possible because no job runs in the time between $\lambda(C)$ and its new starting time, and since both $\rho(C')$ and all $t \in T$ are integral, the new $\lambda(C)$ must now also be integral.
  Since the smallest $\lambda(C) \not\in \N$ increases in each iteration, we must terminate in at most $|\DD'|$ steps, the new schedule satisfies the required property, and the completion time has not increased.
\end{proof}

\paragraph{Cycle Structure.}
Let us now make a few more observations about regular cycle decompositions, and introduce some helpful notation.
If $|t_{k+1} - t_k| \ge p_{\max}$ holds for two consecutive critical times, then there is an internal cycle between them, because an external cycle can contain at most one job.
Furthermore, each critical time is either contained in the interior of an external cycle or it is the left end point of one cycle and the right end point of another cycle (either internal or external with both options possible).
Of course, the machine can in general have some idle time in the schedule before/after a critical time.
We always think of the schedule in a left-to-right manner in such a way that $t_1$ is the leftmost point in the schedule and $t_{|T|}$ is the rightmost point in the schedule.
For an overview of the structure of a regular cycle decomposition of a schedule cf. Figure~\ref{fig:criticalTimesAndCycles}.

Lemma~\ref{lem:cyclesScheduleCMax} allows us to restrict our attention to regular schedules.
Let us define a set $\CC$ of \emph{potential cycles}, which capture all possible ways how cycles may contain or intersect the critical times $T$.
We caution that cycles and potential cycles are quite different objects: a cycle is a schedule of a job set into a time interval, whereas a potential cycle is merely a time interval.
Hence $\CC$ is defined independently of any particular schedule $\sigma$.
Crucially, in any schedule $\sigma$ each cycle corresponds to some potential cycle (we say that it is a \emph{realization} of this potential cycle), and in regular schedules (but not in general) each potential cycle has at most one realization in $\sigma$.
In order not to introduce extra notation, we denote the potential cycles and their realizations in a particular schedule identically.
The set $\CC$ is defined as $\CC = \CC^{\text{int}} \cup \CC^{\text{ext}}$, with $\CC^{\text{int}}$ potential internal cycles, and $\CC^{\text{ext}}$ potential external cycles, which are themselves defined as follows.
We set $\mathcal{C}^{\text{int}} = \left\{ C^{\text{int}}_1, \ldots, C^{\text{int}}_{|T| - 1} \right\}$ with one $C \in \CC^{\text{int}}$ for every potential internal cycle, i.e., for every interval $(t_k, t_{k+1})$, $k \in [|T|-1]$.
We set $\CC^{\text{ext}} = \left\{ C^{\text{ext}}_{k,\ell} \mid k,\ell \in \{2,3,\ldots, |T| - 1\},\, k \le \ell \right\}$, where $C^{\text{ext}}_{k,\ell}$ is a potential external cycle containing all the critical times $t_k, \ldots, t_{\ell}$ in the interior of its interval.

To elucidate the meaning of potential cycles, consider the two potential external cycles $C_{2,3}^{\text{ext}}$ and $C_{3,4}^{\text{ext}}$.
Clearly, in any schedule $\sigma$ both cycles cannot be realized, since both would have to contain $t_3$.
However, the idea of the set $\CC$ is that in any regular schedule, any external cycle will be a realization of one from $\CC^{\text{ext}}$, and similarly for internal cycles.

For each $C \in \CC$ we define $\leftCritical(C)$ to be the index of the largest critical time that is smaller or equal to the start of execution of the first job in $C$ and we define $\rightCritical(C)$ to be the index of the smallest critical time that is larger or equal to the completion time of the last job in $C$.
For example we have $t_k = t_{\leftCritical(C^{\text{int}}_{k})}$, $t_k = t_{\rightCritical(C^{\text{int}}_{k-1})}$, and $k = \rightCritical(C^{\text{ext}}_{k+1,\ell})$ (for any $\ell \ge k + 1$).
It is worth noting that since in regular cycle decompositions external cycles are allowed to contain at most one job, we have $t_{\rightCritical(C)+1} - t_{\leftCritical(C)-1} < p_{\max}$ for each $C \in \mathcal{C}^{\text{ext}}$.
The discussion above thus shows:
\begin{lemma}\label{lem:PCmaxCanonicalCycleDecomposition}
	Let $\vex \in \N^d$ and $\sigma'$ be a schedule of $\vex$.
	Then there is an $\N$-regular schedule $\sigma$ with the same makespan and an $\N$-regular cycle decomposition $\DD$ of $\sigma$ such that each of its cycles is a realization of some $C \in \mathcal{C}$, and each $C \in \CC$ has at most one realization in $\sigma$.
\qed
\end{lemma}
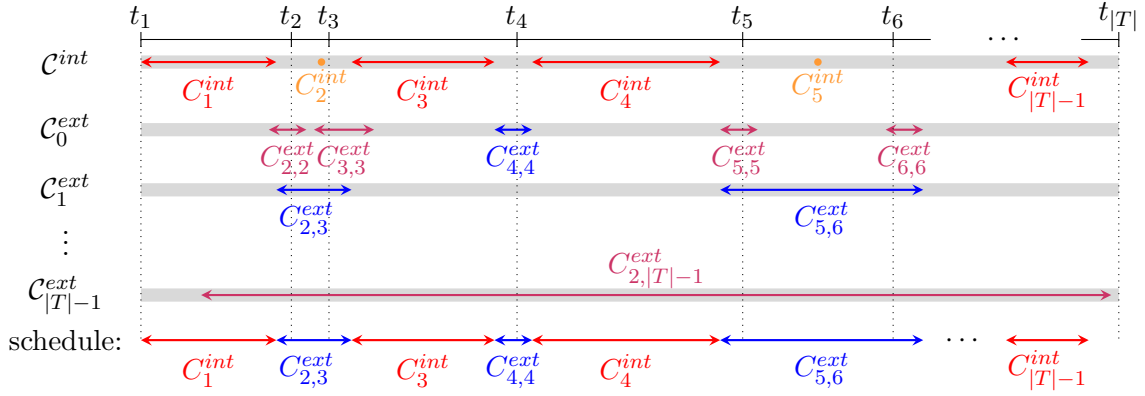
\begin{figure}[bt]
  \begin{center}
  \usetikzlibrary{calc}

\begin{tikzpicture}
\tikzstyle{scheduleCycle}=[thick,inner sep=1pt,>=stealth,fill,circle]
\tikzstyle{internalCycle}=[scheduleCycle, red]
\tikzstyle{internalNonCycle}=[scheduleCycle, orange!80]
\tikzstyle{externalCycle}=[scheduleCycle,blue]
\tikzstyle{externalNonCycle}=[scheduleCycle,purple!80]
\tikzstyle{cLabel}=[midway,yshift=-.4cm]
\tikzstyle{scheduleLine}=[gray!30,line width=5pt]

\coordinate (t1) at (0,2);
\coordinate (t2) at (2,2);
\coordinate (t3) at (2.5,2);
\coordinate (t4) at (5,2);
\coordinate (t5) at (8,2);
\coordinate (t6) at (10,2);
\coordinate (t7) at (13,2);

\foreach \i in {1,...,7} {
  \draw ($(t\i) - (0,.4)$) -- ($(t\i) - (0,.2)$);
}
\foreach \i in {1,...,6} {
  \node (T\i) at (t\i) {$t_{\i}$};
}
\node (T7) at (t7) {$t_{|T|}$};

\draw ($(t1) - (0,.3)$) -- ($(t6) - (0,.3) + (.5,0)$);
\draw ($(t7) - (0,.3)$) -- ($(t7) - (0,.3) - (.5,0)$);
\node at ($(t6)!.5!(t7) - (0,.3)$) {$\cdots$};

\begin{scope}[yshift=-2.3cm]
  \node at (-1,0) {schedule:};
  \draw[internalCycle,<->] (0,0) -- (1.8,0) node [cLabel] {$C^{int}_1$};
  \draw[internalCycle,<->] (2.8,0) -- (4.7,0) node [cLabel] {$C^{int}_3$};
  \draw[internalCycle,<->] (5.2,0) -- (7.7,0) node [cLabel] {$C^{int}_4$};
  \draw[internalCycle,<->] (11.5,0) -- (12.6,0) node [cLabel] {$C^{int}_{|T|-1}$};
  \draw[externalCycle,<->] (4.7,0) -- (5.2,0) node [cLabel] {$C^{ext}_{4,4}$};
  \draw[externalCycle,<->] (1.8,0) -- (2.8,0) node [cLabel,xshift=-.1cm] {$C^{ext}_{2,3}$};
  \draw[externalCycle,<->] (7.7,0) -- (10.4,0) node [cLabel] {$C^{ext}_{5,6}$};
  \node at ($(10.4,0)!.5!(11.5,0)$) {$\cdots$};
\end{scope}

\begin{scope}[yshift=1.4cm]
  \node at (-1,0) {$\mathcal{C}^{int}$};
  \draw[scheduleLine] (0,0) to (13,0);

  \draw[internalCycle,<->] (0,0) -- (1.8,0) node [cLabel] {$C^{int}_1$};
  \node[internalNonCycle,label={[yshift=.1cm,orange!80]270:{$C^{int}_2$}}] at (2.4,0) {};
  \draw[internalCycle,<->] (2.8,0) -- (4.7,0) node [cLabel] {$C^{int}_3$};
  \draw[internalCycle,<->] (5.2,0) -- (7.7,0) node [cLabel] {$C^{int}_4$};
  \node[internalNonCycle,label={[yshift=.1cm,orange!80]270:{$C^{int}_5$}}] at (9,0) {};
  \draw[internalCycle,<->] (11.5,0) -- (12.6,0) node [cLabel] {$C^{int}_{|T|-1}$};
\end{scope}

\begin{scope}[yshift=.5cm]
  \node at (-1,0) {$\mathcal{C}^{ext}_0$};
  \draw[scheduleLine] (0,0) to (13,0);

  \draw[externalNonCycle,<->] (1.7,0) -- (2.2,0) node [cLabel] {$C^{ext}_{2,2}$};
  \draw[externalNonCycle,<->] (2.3,0) -- (3.1,0) node [cLabel] {$C^{ext}_{3,3}$};
  \draw[externalCycle,<->] (4.7,0) -- (5.2,0) node [cLabel] {$C^{ext}_{4,4}$};
  \draw[externalNonCycle,<->] (7.7,0) -- (8.2,0) node [cLabel] {$C^{ext}_{5,5}$};
  \draw[externalNonCycle,<->] (9.9,0) -- (10.4,0) node [cLabel] {$C^{ext}_{6,6}$};
\end{scope}

\begin{scope}[yshift=-.3cm]
  \node (CextOne) at (-1,0) {$\mathcal{C}^{ext}_1$};
  \draw[scheduleLine] (0,0) to (13,0);

  \draw[externalCycle,<->] (1.8,0) -- (2.8,0) node [cLabel,xshift=-.1cm] {$C^{ext}_{2,3}$};
  \draw[externalCycle,<->] (7.7,0) -- (10.4,0) node [cLabel] {$C^{ext}_{5,6}$};
\end{scope}

\begin{scope}[yshift=-1.7cm]
  \node (CextT) at (-1,0) {$\mathcal{C}^{ext}_{|T|-1}$};
  \draw[scheduleLine] (0,0) to (13,0);

  \draw[externalNonCycle,<->] (.8,0) -- (12.9,0) node [cLabel,above] {$C^{ext}_{2,|T|-1}$};
\end{scope}

\foreach \x in {1,...,7} {
  \draw[dotted] ($(t\x) - (0,.3)$) -- ($(t\x) - (0,4.3)$);
}

\node at ($(CextOne)!.5!(CextT) + (0,.1)$) {$\vdots$};

\end{tikzpicture}
  \end{center}
  \caption{\label{fig:criticalTimesAndCycles}%
  Visualization of interleaving internal and external potential cycles in a regular cycle decomposition.
  Red and orange potential cycles are internal ($\mathcal{C}^{\text{int}}$), blue and purple potential cycles are external ($\mathcal{C}^{\text{ext}}$).
  Red and blue cycles have realizations in the example schedule while orange and purple have no realization in this schedule (see the schedule on the bottom).
  For clarity we partition the set $\mathcal{C}^{\text{ext}}$ into layers $\mathcal{C}^{\text{ext}}_0, \mathcal{C}^{\text{ext}}_1, \ldots$, according to the difference between their indices, i.e., $C^{\text{ext}}_{k,\ell} \in \mathcal{C}^{\text{ext}}_{\ell-k}$.
  }
\end{figure}

\paragraph{Modeling Idea.}
Lemma~\ref{lem:PCmaxCanonicalCycleDecomposition} says that every regular schedule has a decomposition into cycles which are realizations of potential cycles $\CC$.
Observe that $|\CC| = O(d^2)$, that is, its size is bounded solely in terms of $d$, the number of job types.
Hence, we want to find an assignment of jobs to potential cycles in $\CC$ in such a way that these cycles can be realized to form a schedule.
There are two constraints that need to be enforced:
\begin{enumerate}
  \item a job $J$ can only be assigned to a cycle $C$ which is contained in $(r^i(J), d^i(J))$,
  \item the combined size of jobs assigned to each cycle is sufficiently small so that the cycles can be arranged to form a schedule.
  \dkcom{well, strictly speaking for external cycles we also need that a job is sufficiently long}
\end{enumerate}
As for the first constraint, we can directly use the $\leftCritical(C)$ and $\rightCritical(C)$ operators to enforce it.
The second task is more involved, since it demands that the cycles obey their limitations (e.g., the total size of jobs assigned to the cycle $C^{\text{int}}_k$ is at most $t_{k+1} - t_{k}$) or for example if a job is assigned to the cycle $C^{\text{ext}}_{2,4}$, then all of the cycles $C^{\text{int}}_2,C^{\text{int}}_3,C^{\text{ext}}_{2,2},C^{\text{ext}}_{3,3},C^{\text{ext}}_{4,4},C^{\text{ext}}_{2,3}$, and $C^{\text{ext}}_{3,4}$ must be empty (i.e., no job can be assigned to any of these cycles).

\paragraph{When can a job be in a cycle?}
We introduce a binary indicator constant $\chi^i_{j,C} \in \{ 0,1 \}$ expressing whether jobs of type $j$ can be assigned to a potential cycle $C \in \mathcal{C}$:
\begin{equation}\label{eq:chiDefinition}
  \chi^i_{j,C} =
  \begin{cases}
    1 & \mbox{if } C = C^{\text{int}}_k \mbox{, and } r^i_j \le t_k, t_{k+1} \le d^i_j,    \\
    1 & \mbox{if } C = C^{\text{ext}}_{\ell,k} \mbox{, and } r^i_j \le t_{\ell-1}, t_{k+1} \le d^i_j, t_{k+1} - t_{\ell-1} \ge p^i_j, t_{k}-t_\ell \le p^i_j - 2  \\
    0 & \mbox{otherwise.}
  \end{cases}
\end{equation}
Intuitively, a job of type $j$ is allowed to be assigned to a cycle $C$ if $C$ is contained in $(r^i_j, d^i_j)$, and if it is an external cycle, then additionally we require that  $|t_{\rightCritical(C)+1} - t_{\leftCritical(C)-1}| \geq p^i_j$ so that the (single) job fits into $C$ but $|t_{\rightCritical(C)} - t_{\leftCritical(C)}| \le p^i_j - 2$ so $t_{\rightCritical(C)}, t_{\leftCritical(C)}$ could be interior points of the interval of processing a job of type $j$ on this machine.

We say that two potential cycles $C,C' \in \mathcal{C}$ are \emph{incompatible} and write $C \not\approx C'$ if there is no schedule realizing both $C$ and $C'$.
This is for example the case of $C^{\text{int}}_3$ and $C^{\text{ext}}_{2,4}$ in the example above.
Next, for a cycle $C$ and a critical time $t_k$ we say $C$ \emph{ends before} $t_k$ (write $C \lhd t_k$) if $t_{\rightCritical(C)} \le t_k$.
Furthermore, we write $t_k \lhd C$ if $C$ \emph{starts after} $t_k$, that is, if $t_k \le t_{\leftCritical(C)}$.
Thus, we have for example $C^{\text{int}}_1 \lhd t_3$ but not $C^{\text{ext}}_{2,2} \lhd t_2$ and $t_2 \lhd C^{\text{ext}}_{6,6}$ but not $t_3 \lhd C^{\text{ext}}_{3,3}$.

\subsubsection{The MIMO model}
We guess the value $\bar{C}_{\max}$ of the optimal makespan with $\min_{i \in [\kappa]} \max_{j \in [d]} r^i_j < \bar{C}_{\max} \le \max_{i \in [\kappa]} \max_{j \in [d]} d^i_j$, and we construct a MIMO instance which is feasible if a schedule with makespan at most $\bar{C}_{\max}$ exists.
The optimal makespan $C_{\max}$ is then found by binary search at an additional polynomial factor in the time complexity.
Alter the input instance as follows.
For each machine kind $i \in [\kappa]$ and job type $j \in [d]$, set $d^i_j = \min(d^i_j, \bar{C}_{\max})$.
Note that if there is a schedule after this change, then its makespan is at most $\bar{C}_{\max}$.
Recall that a MIMO instance is defined via $\tau$ polytopes, objective functions, and multiplicities, and here the number of polytopes $\tau$ is the number of machine kinds $\kappa$.
The machine multiplicities $\mu^i$ translate directly to the MIMO multiplicities $\mu^i$.
For the makespan scheduling objective, the MIMO instance will not have any objective (i.e., it will be a feasibility instance).
Thus, it remains to give a description of the polytope $P^i$ for each $i \in [\kappa]$.
Recall that $P^i$ has $d+d^i$ dimensions, and in our description the first $d$ variables will always correspond to a configuration vector, i.e., the vector of multiplicities of jobs of each type scheduled on a single machine.
The remaining variables are auxiliary and encode an $\N$-regular cycle decomposition.
Goemans and Rothvo{\ss}~\cite{GoemansRothvoss2014} note that these auxiliary variables are in fact necessary, because the set of configurations is not convex.
Next, we list the (integer) variables we use to describe the set of configurations of jobs on a single machine of kind $i$:
\begin{itemize}
  \item $x^i_j$ for each job type $j \in [d]$ denotes the number of jobs of type $j$ in a configuration,
  \item $y^i_{j,C}$ for each potential cycle $C \in \CC$ denotes the number of jobs of type $j$ in a realization of $C$ on this machine,
  \item $z^i_C$ for each potential external cycle $C \in \CC^{\text{ext}}$ is a binary indicator of whether $C$ has a non-empty realization, i.e., $z^i_C = 1$ if and only if exactly one job is assigned to $C$.
\end{itemize}

Clearly the $x^i_j$ variables are obtained by aggregation of the $y^i_{j,C}$ variables over all $C \in \CC$.
The following constraints use the constants $\chi^i_{j,C}$ to enforce the intended meaning of definition~\eqref{eq:chiDefinition}:
\begin{align}
      x^i_j & = \sum_{C \in \mathcal{C}} y^i_{j,C}           & \forall j \in [d] \label{eq:Cmax:sumOfExes}\\
      z^i_C &= \sum_{j \in [d]} y^i_{j,C}              & \forall C \in \mathcal{C}^{\text{ext}} \label{eq:Cmax:cycle_bounds} \\
      0 \leq y^i_{j,C} &\leq \chi^i_{j,C} \cdot n_j          & \forall j \in [d] \,, \forall C \in \mathcal{C} \label{eq:Cmax:chiMIMObounds} \\
      0 \leq z^i_C &\leq 1                                 & \forall C \in \mathcal{C}^{\text{ext}} \label{eq:Cmax:ziC}\\
      y^i_{j,C},z^i_C,x^i_j & \in \mathbb{N}               & \forall j \in [d] \,, \forall C \in \mathcal{C}
\end{align}
Note that constraints \eqref{eq:Cmax:cycle_bounds} enforce that a (binary) variable $z^i_C$ is set to $1$ if and only if at least one job is assigned to the cycle $C$ and, moreover, at most one job can be assigned to $C$ by~\eqref{eq:Cmax:ziC}.
We now use another set of constraints to forbid assigning jobs to cycles contained in an external potential cycle with a non-empty realization.
We note that the largest coefficient of the following constraint is $p_{\max}$, which is one of the important improvements of our approach over the one of Goemans and Rothvo{\ss}:
\begin{equation}\label{eq:Cmax:incompatibleCyclesDisable}
  \sum_{j=1}^d \sum_{C' \not\approx C} y^i_{j,C'} \leq p_{\max} \cdot (1-z^i_C)       \qquad\qquad\qquad\qquad\qquad \forall C \in \mathcal{C}^{\text{ext}}
\end{equation}
It is worth noting that, since each external cycle $C$ can contain at most one job, no cycle $C'$ which is incompatible with $C$ can contain jobs whose total size amounts to more than $p_{\max}$; and thus, the constraints \eqref{eq:Cmax:incompatibleCyclesDisable} do not restrict such a cycle $C'$ if $z^i_C = 0$.
Finally, we have to ensure that it is possible to arrange the cycles (which are now all compatible) into a schedule.
To this end we add the following constraints.
\begin{equation}\label{eq:Cmax:cycleVolumeBounds}
  \sum_{j \in [d]} \sum_{t_\ell \lhd C \lhd t_k} p^i_j \cdot y^i_{j,C} \leq t_k - t_\ell \qquad\qquad\qquad\qquad\qquad \forall k,\ell \in [|T|] \,, \ell < k
\end{equation}
We stress here that in this condition in the sum in the left hand-side we sum over all cycles with \(t_\ell \lhd C \lhd t_k\), i.e., both internal and external cycles.

\paragraph{From Vectors to Schedules.}
Now we are going to describe how to interpret a vector $\left( \vex, \vey, \vez \right)$ satisfying constraints \eqref{eq:Cmax:sumOfExes}--\eqref{eq:Cmax:cycleVolumeBounds}) as a schedule.
Fix a vector $(\vex, \vey, \vez)$ satisfying \eqref{eq:Cmax:sumOfExes}--\eqref{eq:Cmax:cycleVolumeBounds}.
Recall that a schedule for a particular fixed machine is a mapping of jobs to non-overlapping time intervals (except for their endpoints).
We begin by defining a total (linear) order $\prec$ on $\mathcal{C}$ as follows:
\[
  C^{\text{int}}_1 \prec C^{\text{ext}}_{2,2} \prec C^{\text{ext}}_{2,3} \prec \cdots \prec C^{\text{ext}}_{2,|T|-1} \prec C^{\text{int}}_2 \prec \cdots \prec C^{\text{int}}_{|T| - 1},
\]
that is, $C^{\text{int}}_k \prec C^{\text{int}}_\ell$ if $k < \ell$, $C^{\text{int}}_\ell \prec C^{\text{ext}}_{\hat{\ell},k}$ if $\ell < \hat{\ell}$, and $C^{\text{ext}}_{\ell,k} \prec C^{\text{ext}}_{\ell,\hat{k}} \prec C^{\text{int}}_{\bar{k}}$ if $k < \hat{k}$ and $\bar{k} \geq \ell$.
Now, given a vector $\vey$ we define a mapping~$\sigma(\vey)$ by incrementally processing the potential cycles $\CC$ in the order $\prec$ as shown in Algorithm~\ref{alg:sigmaFromY}.
(Note that although Algorithm~\ref{alg:sigmaFromY} runs in time polynomial in $n$ and not $\log n$, this is not an issue since its purpose is to define the mapping~$\sigma$, which is only used to prove the correctness of the constructed model.
Moreover, it is not difficult to see how to modify Algorithm~\ref{alg:sigmaFromY} to run in time polynomial in $\log n$ and return a compact encoding of the mapping~$\sigma$.)
\begin{algorithm}[tb]
  \SetKwProg{Def}{def}{:}{}
  \SetKwFunction{sigmaFunction}{$\sigma$}
  \SetKwFunction{cycleHandler}{handleCycle}
  \SetKwFunction{scheduleEnd}{endOf}
  \DontPrintSemicolon

  \Def{\sigmaFunction{$i, \vey$}}{
    $\sigma \leftarrow \emptyset$ \;
    \ForEach{$C \in \mathcal{C}$ in order $\prec$ \textnormal{\textbf{if}} $\vey_C \neq \mathbf{0}$}{
      \cycleHandler{$i, \sigma, C, \vey_C$} \;
    }
    \Return $\sigma$ \;
  }

  \BlankLine

  \Def{\cycleHandler{$i, \sigma, C, \vey_C$}}{
    \For{$j = 1$ \KwTo $d$}{
      \For{$\ell = 1$ \KwTo $y_{j,C}$}{
        \nlset{Start}\label{alg:sigmaFromY:Start} $t \leftarrow \max($\scheduleEnd{$\sigma$}$, \leftCritical(C))$ \; 
        $\sigma \leftarrow \sigma \cup \left\{ \left( j, \interval{t}{t + p^i_j} \right) \right\}$ \;
      }
    }
  }

  \caption{\label{alg:sigmaFromY}
  Computing a schedule from a vector satisfying \eqref{eq:Cmax:sumOfExes}--\eqref{eq:Cmax:cycleVolumeBounds}.
  The function \texttt{endOf} returns the completion time of the schedule it is given.
  }
\end{algorithm}
We are going to prove later that if a schedule for $\vey$ exists, then $\sigma(\vey)$ is a schedule, however, it is worth noting that in such case $\sigma(\vey)$ is not necessarily a unique schedule corresponding to the cycle decomposition encoded by $\vey$.
Notice that (due to the line \ref{alg:sigmaFromY:Start} of Algorithm~\ref{alg:sigmaFromY}) the produced schedule is ``left aligned''.
Left aligned schedules are very natural and, even though this is not the case for the objective $C_{\max}$, may have better objective values (e.g., for $\sum C_j$).
This is formalized in the following lemma.
We assume each machine is idle in the time interval $(-\infty,0)$.
Note that Lemma~\ref{lem:leftAlignedScheduleCMax} does not yet prove that $\sigma(\vey)$ is a schedule, but will be used for that purpose later.

\begin{lemma}\label{lem:leftAlignedScheduleCMax}
  Fix a machine of kind $i$ and let $\sigma(\vey)$ be as defined by Algorithm~\ref{alg:sigmaFromY}.
  Then, for each job~$J$ assigned by $\sigma(\vey)$ there exists a critical time $t(J) \in T$ such that the machine is
  \begin{itemize}
    \item busy from $t(J)$ to the time when $J$ begins to be processed, and
    \item idle for at least one time unit right before $t(J)$, that is, no job is processed during time $(t(J) - 1, t(J))$.
  \end{itemize}
\end{lemma}
\begin{proof}
  We prove this by induction on the number of jobs in a prefix of $\sigma(\vey)$, and denote this number~$k$.
  Clearly, if $k = 1$, then both our assumptions hold, since the first job in $\sigma(\vey)$ is scheduled to start at some critical time $t_\ell$ due to the line \ref{alg:sigmaFromY:Start} and indeed the machine is idle before $t_\ell$.

  Assuming both conditions hold for $(k-1) \ge 1$ we want to show that both conditions hold for~$k$.
  Let $J$ be the $k$-th job in $\sigma(\vey)$.
  If $J$ is scheduled to be processed on the machine right after its predecessor $J'$, which is the $(k-1)$-st job, we are done, since we can set $t(J) = t(J')$ and the rest follows from the induction hypothesis.
  Otherwise, we claim $J$ is scheduled to start being processed at some critical time $t_{\ell} \in T$.
  This again follows from the line \ref{alg:sigmaFromY:Start}, since when assigning $J$ to $\sigma(\vey)$ we clearly have \scheduleEnd{$\sigma$}$ < t_{\leftCritical(C)}$, where $C$ is the cycle to which $J$ belongs (as its first job).
  Now we are done since the machine must be idle right before $t_\ell$ (for at least one time unit).
\end{proof}

We stress here that in $\sigma(\vey)$ we have $t(J) \in T$ for every job $J$.
Consequently, any gap in $\sigma(\vey)$ can only be of the form $(t,t')$ with $t' \in T \cup \{\infty\}$, that is, $\sigma(\vey)$ is $\N$-regular.

We now use Lemma~\ref{lem:leftAlignedScheduleCMax} to prove an equivalence between feasibility of the constraints \eqref{eq:Cmax:sumOfExes}--\eqref{eq:Cmax:cycleVolumeBounds} and the existence of an ($\N$-regular) schedule for $\vex$.
Note a certain ambiguity: it is possible that for two distinct vectors $\vey_1, \vey_2$ satisfying constraints \eqref{eq:Cmax:sumOfExes}--\eqref{eq:Cmax:cycleVolumeBounds} we have $\sigma(\vey_1) = \sigma(\vey_2)$.
For example, take a scheduling instance with critical times $0,2,4$ and with a single job $J$ of size $2$ which may be assigned into both time slots, i.e., $r(J) = 0$ and $d(J) = 4$.
Then a vector $\vey_1$ encoding that $J$ is assigned to $C_1^{\text{int}}$ and a vector $\vey_2$ encoding that $J$ is assigned to $C^{\text{ext}}_{2,2}$ are such that $\sigma = \sigma(\vey_1) = \sigma(\vey_2)$ with $\sigma$ scheduling $J$ to the interval $[0,2]$.
This causes no problems but is helpful to keep in mind as we approach the next proof.

\begin{lemma}\label{lem:modelPlacesCyclesCorrectly}
Let $i \in [\kappa]$ be a machine kind and let $\vex \in \N^d$.
There exists a schedule of $\vex$ with makespan at most $\bar{C}_{\max}$ if and only if there exists $\vey,\vez$ such that $(\vex, \vey, \vez)$ satisfies \eqref{eq:Cmax:sumOfExes}--\eqref{eq:Cmax:cycleVolumeBounds}, and, moreover, $\sigma(\vey)$ is an $\N$-regular schedule with makespan at most $\bar{C}_{\max}$.
\end{lemma}
\begin{proof}
  For the forward direction suppose there exists a schedule of $\vex$ with makespan $\bar{C}_{\max}$.
  Then by Lemma~\ref{lem:cyclesScheduleCMax} it is possible to turn a schedule for $\vex$ into an $\N$-regular one together with an $\N$-regular cycle decomposition with the same or smaller makespan, and an $\N$-regular cycle decomposition has an encoding by variables $y^i_{j,C}$ and $z_C^i$ satisfying the constraints \eqref{eq:Cmax:sumOfExes}--\eqref{eq:Cmax:cycleVolumeBounds}.
  \dkcom{we can maybe give precise assignment to variables and verify the model---this lemma is being heavily reused}
  \mkcom{we could, but this seems clear enough ==> low priority IMHO.}

  For the backward direction we assume that $\vex,\vey,\vez$ satisfy all constraints \eqref{eq:Cmax:sumOfExes}--\eqref{eq:Cmax:cycleVolumeBounds}.
  Let $\sigma = \sigma(\vey)$.
  It follows from the definition of $\chi_{j,C}^i$, the feasibility of $(\vex, \vey, \vez)$, and Algorithm~\ref{alg:sigmaFromY} that a job $J$ cannot be scheduled before $r^i(J)$ in $\sigma$.
  We have to show that no job is scheduled after its due date; the fact that no job is scheduled after $\bar{C}_{\max}$ follows since we have altered the due dates to be at most $\bar{C}_{\max}$.
  Assume for a contradiction there is a job $J$ which finishes in $\sigma$ after its due date $d^i(J)$.
  Let $C(J)$ be the cycle which was processed by the outer loop of Algorithm~\ref{alg:sigmaFromY} when $J$ was added to $\sigma$, and let $t(J)$ be defined as in Lemma~\ref{lem:leftAlignedScheduleCMax}.
  Then by Lemma~\ref{lem:leftAlignedScheduleCMax} some constraint \eqref{eq:Cmax:cycleVolumeBounds} was violated, namely the constraint saying that the total size of jobs between $t(J)$ and $t_{\rightCritical(C(J))}$ must be at most $t_{\rightCritical(C(J))} - t(J)$, i.e.,
  \[
    \sum_{j \in [d]} \sum_{t(J) \lhd C \lhd t_{\rightCritical(C(J))}} p^i_j \cdot y^i_{j,C} \leq t_{\rightCritical(C(J))} - t(J) \enspace .
  \]
  Recall that $t(J) \in T$ for any job $J$ scheduled in $\sigma$.
  Thus $\sigma$ is an $\N$-regular schedule of $\vex$ with makespan $\bar{C}_{\max}$.
\end{proof}

\paragraph{Finishing the MIMO Model.}
We have described, for each $i \in [\kappa]$, the constraints defining a polytope $P^i$.
The target vector of the MIMO model is the vector of multiplicities of jobs, i.e., $\ven = \left(n_1, \ldots, n_d\right)$.
The last task in the design of a MIMO model is to define the projections $\pi^i$.
In our case $\pi^i$ projects out all variables except $\vex$.
Let us now determine the parameters of the model; recall $p_{\max}$ is the maximum finite job size, that is, $p_{\max} = \max_{j=1}^d \max_{i = 1}^{\kappa} p^i_j$:
\begin{lemma}[MIMO model for $R|r^i_j,d^i_j|C_{\max}$]
\label{lem:MIMOPCmax}
  Let $\mathcal{I}$ be an instance of $R|r^i_j,d^i_j|C_{\max}$ with $m$ machines of $\kappa$ kinds and $d$ job types with maximum job size $p_{\max}$.
  There is a MIMO model $\mathcal{S}$ for $\mathcal{I}$ with the following values of MIMO parameters:
  \begin{tasks}[style=itemize](3)
    \task $\mathcal{S}(\Delta) = p_{\max}$,
    \task $\mathcal{S}(M) = \Oh((d)^2)$,
    \task $\mathcal{S}(d) = d$,
    \task $\mathcal{S}(d^i) = \Oh((d)^2)$,
    \task $\mathcal{S}(N) = m$, and
    \task $\mathcal{S}(\tau) = \kappa$.
  \end{tasks}
\end{lemma}
\begin{proof}
  For each $i \in [\kappa]$, the constraints \eqref{eq:Cmax:sumOfExes}--\eqref{eq:Cmax:cycleVolumeBounds} define a polytope $P^i$.
  The projections $\pi^i$ drop all variables except for $\vex$, and the target vector is the vector of job multiplicities $\ven$.
  The correctness of the model follows directly from Lemma~\ref{lem:modelPlacesCyclesCorrectly}.

  The coefficients in constraints \eqref{eq:Cmax:sumOfExes}--\eqref{eq:Cmax:cycleVolumeBounds} are either 1, $p^i_j$, or $p_{\max}$, and thus the largest coefficient in the thus obtained $P$-representation is $\mathcal{S}(\Delta) = p_{\max}$.
  The number of variables $y^i_{j,C}$ is $|\mathcal{C}|$ and $z^i_{C}$ is $|\mathcal{C}^{\text{ext}}|$; we have $|\mathcal{C}| = \mathcal{O}((d)^2)$.
  Thus, we can estimate the number of projected out variables by $\mathcal{S}(d^i) = \mathcal O(d^2)$ and the number of remaining variables $\mathcal{S}(d) = d$.
  Since there are $\Oh((d)^2)$ constraints and $\Oh((d)^2)$ box constraints, we have $\mathcal{S}(M) = \mathcal O((d)^2)$.
  Clearly, $\mathcal{S}(N)$ is the number of machines $m=\|\vemu\|_1$.
  Finally, we have $\mathcal{S}(\tau) = \kappa$, since we have $\kappa$ different types of polytopes corresponding to $\kappa$ kinds of machines.
\end{proof}

By combining Lemma~\ref{lem:MIMOPCmax} with Theorem~\ref{thm:implicitMIMO} we arrive to the following theorem.

\begin{theorem}
\label{thm:FPT:MIMOPCmax}
  Problem $R|r^i_j,d^i_j|C_{\max}$ with $m$ machines of $\kappa$ kinds and $d$ job types with maximum job size $p_{\max}$ have fixed-parameter tractable algorithms
  \begin{itemize}
    \item single-exponential in $m + d$,
    \item single-exponential in $p_{\max} + d$, or
    \item double-exponential in $d + \kappa$ if $p_{\max}$ is given in unary.
  \end{itemize}
\end{theorem}
\begin{proof}
  Note that in total we construct only polynomially many MIMO models, since there are only polynomially many possibilities (that need to be checked via the binary search procedure) for the guess of the value $\bar{C}_{\max}$.
  We apply part 1, 3, and 4 of Theorem~\ref{thm:implicitMIMO} to obtain the claimed result:
  \begin{itemize}
  \item For application of part 1 we must be able to construct a MIMO model $\mathcal{S}$ and bound $\mathcal{S}(N)$ and $\mathcal{S}(D) = \max_{i \in [d]} \mathcal{S}(d^i)$.
  By Lemma~\ref{lem:MIMOPCmax} we have $\mathcal{S}(N) = m$ and $\mathcal{S}(d^i) = \mathcal O(d^2)$.
  \item For application of part 3 we must be able to construct a MIMO model $\mathcal{S}$ and bound $\mathcal{S}(M), \mathcal{S}(d)$, and $\mathcal{S}(\Delta)$.
  By Lemma~\ref{lem:MIMOPCmax} we have $\mathcal{S}(M) = O(d^2), \mathcal{S}(d) = d$, and $\mathcal{S}(\Delta) = p_{\max}$.
  \item For application of part 4 we must be able to construct a MIMO model $\mathcal{S}$, bound $\mathcal{S}(\tau)$, $\mathcal{S}(M)$ and $\mathcal{S}(D)$ if $\mathcal{S}(\Delta)$ is given in unary.
  By Lemma~\ref{lem:MIMOPCmax} we have $\mathcal{S}(D) = \SSS(M) = \mathcal O(d^2), \mathcal{S}(\tau) = \kappa$, and $\mathcal{S}(\Delta) = p_{\max}$.
  \qedhere
  \end{itemize}
\end{proof}

\subsubsection{Introducing Speeds into the Model}
We begin by extending the definition of the compact encoding of the $R | r^i_j,d^i_j | C_{\max}$ problem studied so far by introducing speeds.
The motivation is a scenario when the job size vectors of several machine kinds are obtained by scaling (up or down) a single ``unit speed'' job size vector, such as (but not only) in the case of uniformly related machines.
Clearly a more compact way of encoding such an instance is by giving the ``unit speed'' job size vector and for each machine kind related to this vector the corresponding scalar.
Note that this does not generalize the problem, but since introducing speeds into the MIMO model will not increase the coefficients, it will allow us to provide an efficient algorithm for a wider range of instances such as problems with non-constant speed ratios.

\prob{\textsc{Makespan Minimization on Unrelated Machines with Speeds} ($R | r^i_j,d^i_j | C_{\max}$)}{
There are $\kappa$ kinds of machines, each kind with $\bar{\tau}$ different speeds encoded by a speed vector $\ves^i=(s^i_1, \dots, s^i_{\bar{\tau}})$, hence $\tau = \bar{\tau} \cdot \kappa$ different machine types altogether, and $d$ types of jobs.
The number of machines of kind $i \in [\kappa]$ and speed $s^i_q$, $q \in [\bar{\tau}]$, is $\mu^i_q$ and the number of jobs of type $j \in [d]$ is $n_j$, and $\vemu = (\mu^1_1, \dots, \mu^1_{\bar{\tau}}, \dots, \mu_{\bar{\tau}}^\kappa)$ and $\ven = (n_1, \dots, n_d)$, hence there are $m=\|\vemu\|_1$ machines and $n = \|\ven\|_1$ jobs.
Each job type is specified by three vectors giving its size, release time, and due date on each machine kind, i.e., for each $j \in [d]$ given are vectors $\vep_j = \left( p^1_j, \ldots, p^\kappa_j \right) \in \left(\N \cup \{\infty\}\right)^\kappa$, $\ver_j = \left( r^1_j, \ldots, r^\kappa_j \right) \in \N^\kappa$, and $\ved_j = \left( d^1_j, \ldots, d^\kappa_j \right) \in \N^\kappa$.
}
{A schedule of all of the jobs on the specified $m$ machines (if one exists) minimizing the time when the last job finishes (i.e., the makespan).}

We call this representation of an instance its \emph{$\ves$-representation}.
The definition of a schedule is nearly the same as before except that a job $J$ needs to be assigned a processing time interval of size $p^i_j / s^i_q$ if it is assigned to a machine of kind $i$ with speed $s^i_q$.

Let us first observe why the most obvious place for introducing machine speeds to the model above does not lead to a valid MIMO formulation.
Undoubtedly, one is tempted to divide the job sizes with the speed $s$ and keep the rest of the model unchanged, however, when doing so we may introduce coefficients which are non-integer and quite large.
To circumvent this we intuitively leave the size intact and ``make the clock tick $s$ times slower''.
Again we are facing the integrality issue but this time in the right-hand sides (not in the constraints) and, as we are going to see, this makes the issue solvable by rounding the resulting right hand sides.

Again, we assume to have guessed a makespan $\bar{C}_{\max}$, modified the due dates accordingly, and our goal is to state linear constraints defining a polytope $P^{i,q}$ of configurations of jobs on a machine of kind $i \in [\kappa]$ with a speed index $q \in [\tilde{\tau}]$.
Hence, from now on we fix a machine type, i.e., fix a machine kind $i \in [\kappa]$ and a speed index $q \in [\tilde{\tau}]$.
Now we replace constraints \eqref{eq:Cmax:cycleVolumeBounds} with
\begin{equation}\label{eq:CmaxSpeed:cycleVolumeBounds}
  \sum_{j \in [d]} \sum_{t_\ell \lhd C \lhd t_k} p^i_j \cdot y^i_{j,C} \leq  \left\lfloor s^i_q \cdot (t_k - t_\ell) \right\rfloor \qquad\qquad\qquad\qquad\qquad \forall k,\ell \in [|T|] \,, \ell < k
\end{equation}
We stress that the floor of the right hand-side does not change the set of encoded configurations, since the sum on the left hand side is integral.
Furthermore, in order to reuse the model given in the previous section, we have to alter the definition of $\chi^i_{j,C}$ (i.e., to change~\ref{eq:chiDefinition}) in a way which captures the machine speed $s^i_q$.
Note that the speed only affects the processing time of a job type $j$ and thus it is enough to redefine $\chi^i_{j,C}$ for external cycles only
\[
  \chi^i_{j,C} = 1 \text{ if }
  r^i_j \le t_{\ell - 1},
  t_{k+1} \le d^i_j,
  t_{k+1} - t_{\ell-1} \ge p^i_j / s^i_q,
  \text{ and }
  t_{k} - t_{\ell} < p^i_j / s^i_q \,,
\]
where $C = C^{\operatorname{ext}}_{\ell,k}$.
Lemma~\ref{lem:cyclesScheduleCMax} stated that, in the case without speeds, it was sufficient to restrict our attention to $\N$-regular schedules.
It is straightforward to see here that because on a machine with speed $s$ all processing times are multiples of $\frac{1}{s}$, it is sufficient to restrict our attention to $\frac{\N}{s}$-regular schedules.
The proof of the following lemma goes along the same lines as the proof of Lemma~\ref{lem:modelPlacesCyclesCorrectly}.
Here, we only add the argument showing that~\eqref{eq:Cmax:incompatibleCyclesDisable} remains valid.
\begin{lemma}\label{lem:QCmaxIntroducingSpeeds}
  Let $i \in [\kappa]$, $q \in [\tilde{\tau}]$, and let $\vex \in \N^d$.
  Then there exists an $\frac{\N}{s^i_q}$-regular schedule of $\vex$ with makespan $\bar{C}_{\max}$ if and only if there exists $\vey,\vez$ such that $(\vex, \vey, \vez)$ satisfy constraints \eqref{eq:Cmax:sumOfExes}--\eqref{eq:Cmax:incompatibleCyclesDisable} and \eqref{eq:CmaxSpeed:cycleVolumeBounds}, and, moreover, $\sigma(\vey)$ given by Algorithm~\ref{alg:C_maxSpeed:sigmaFromY} is an $\frac{\N}{s^i_q}$-regular schedule with makespan $\bar{C}_{\max}$.
\end{lemma}
\begin{proof}
  First observe that $\sigma(\vey)$ given by Algorithm~\ref{alg:C_maxSpeed:sigmaFromY} is $\frac{\N}{s^i_q}$-regular, since again we have $t(J) \in T$ for every job $J$ assigned in $\sigma(\vey)$.
  Fix a potential external cycle $\hat{C} = C^{\operatorname{ext}}_{\ell,k}$.
  Recall that we have $t_k - t_{\ell} \le p_{\max} / s^i_q$, since otherwise $\hat{C}$ is not a potential cycle as no job can be assigned to it.
  Since the case $z_{\hat{C}} = 1$ is clearly valid, we have to argue about the other case, that is, the case when some jobs are assigned to cycles fully contained in $\hat{C}$.
  Since all cycles incompatible with $\hat{C}$, $\mathcal{C}_{\not\approx \hat{C}}$, are contained in the interval $\interval{t_\ell}{t_k}$, we conclude that
  \[
  \sum_{j \in [d]} \sum_{C \in \mathcal{C}_{\not\approx \hat{C}}} \frac{p^i_j}{s^i_q} \cdot y_{j,C} \le t_k - t_\ell \,.
  \]
  Clearly, we get $\sum_{j \in [d]} \sum_{C \in \mathcal{C}_{\not\approx \hat{C}}} \frac{p^i_j}{s^i_q} \cdot y_{j,C} \le p_{\max} / s^i_q$ which is equivalent to $\sum_{j \in [d]} \sum_{C \in \mathcal{C}_{\not\approx \hat{C}}} p^i_j \cdot y_{j,C} \le p_{\max}$.
  The lemma follows, since we have $1 \le p^i_j$ for all $j \in [d]$, that is, $\sum_{j \in [d]} \sum_{C \in \mathcal{C}_{\not\approx \hat{C}}} y_{j,C} \le \sum_{j \in [d]} \sum_{C \in \mathcal{C}_{\not\approx \hat{C}}} p^i_j \cdot y_{j,C}$
\end{proof}

\begin{algorithm}[tb]
  \SetKwProg{Def}{def}{:}{}
  \SetKwFunction{sigmaFunction}{$\sigma$}
  \SetKwFunction{cycleHandler}{handleCycle}
  \SetKwFunction{scheduleEnd}{endOf}
  \DontPrintSemicolon

  \Def{\cycleHandler{$i, \sigma, C, \vey_C$}}{
    \For{$j = 1$ \KwTo $d$}{
      \For{$\ell = 1$ \KwTo $y_{j,C}$}{
        \nlset{Start}\label{alg:C_maxSpeed:sigmaFromY:Start} $t \leftarrow \max($\scheduleEnd{$\sigma$}$, \leftCritical(C))$ \;
        $\sigma \leftarrow \sigma \cup \left\{ \left( j, \interval{t}{t + \frac{p^i_j}{s}} \right) \right\}$ \;
      }
    }
  }

  \caption{\label{alg:C_maxSpeed:sigmaFromY}
  Computing a schedule from a vector satisfying \eqref{eq:Cmax:sumOfExes}--\eqref{eq:Cmax:incompatibleCyclesDisable} and \eqref{eq:CmaxSpeed:cycleVolumeBounds}.
  Here, $s$ is the speed of the machine under consideration.
  }
\end{algorithm}

Let us now determine the parameters of the resulting MIMO model.
Recall that $p_{\max}$ is the maximum finite job size in a machine \emph{kind}, and the processing time of a job on a given machine \emph{type} is obtained by scaling the size with respect to the machine kind (i.e., speed); cf. Section~\ref{sec:intro:apps}.
For this reason we stress that $p_{\max} = \max_{i = 1}^\kappa\max_{j = 1}^d p^i_j$, that is, $p_{\max}$ is defined with respect to job sizes $p^i_j$ not taking speeds into account.
\begin{lemma}[MIMO model for $R|r^i_j,d^i_j|C_{\max}$]
\label{lem:MIMORCmax}
  Let $\mathcal{I}$ be an instance of $R|r^i_j,d^i_j|C_{\max}$ in its $\ves$-representation with $m$ machines of $\tau$ machine types, and $d$ job types with maximum job size $p_{\max}$.
  There is a MIMO model $\mathcal{S}$ for $\mathcal{I}$ with the following values of MIMO parameters:
  \begin{tasks}[style=itemize](3)
    \task $\mathcal{S}(\Delta) = p_{\max}$,
    \task $\mathcal{S}(M) = \mathcal O((d)^2)$,
    \task $\mathcal{S}(d) = d$,
    \task $\mathcal{S}(d^i) = \mathcal O((d)^2)$,
    \task $\mathcal{S}(N) = m$, and
    \task $\mathcal{S}(\tau) = \tau = \kappa \cdot \bar{\tau}$.
  \end{tasks}
\end{lemma}
\begin{proof}
  There are only two changes in the model when compared with the model discussed previously, namely, we have introduced machine types and speeds.
  First, by Lemma~\ref{lem:QCmaxIntroducingSpeeds}, we only changed right-hand sides in the inequality description of $P^i$'s and thus this does not affect the MIMO parameters at all.
  We stress here that the left hand sides of the constraints in our model use the given job sizes for the unit speed machine of each kind; consequently the largest coefficient in the model is still governed by $p_{\max}$.
  Second, the number of polytope types, $\mathcal{S}(\tau)$, is $\tau = \kappa \cdot \bar{\tau}$.
  All other bounds on parameters of $\mathcal{S}$ follow from Lemma~\ref{lem:MIMOPCmax}.
\end{proof}

\begin{theorem}
  Problem $R|r^i_j,d^i_j|C_{\max}$ in its $\ves$-representation with $m$ machines of $\tau$ types and $d$ job types of maximum job size $p_{\max}$ admits fixed-parameter algorithms for
  \begin{itemize}
    \item parameter $m + d$,
    \item parameter $p_{\max} + d$, and
    \item parameter $d + \tau$ when $p_{\max}$ is given in unary.
  \end{itemize}
\end{theorem}
\begin{proof}
	The parameters stay the same as in Theorem~\ref{thm:FPT:MIMOPCmax}.
\end{proof}

\subsubsection{Objectives Similar to \texorpdfstring{$C_{\max}$}{Cmax}}
\paragraph{Other Maximal Objectives.}
We now use the simple observation that both the $L_{\max}$ and $F_{\max}$ objectives are equivalent to $C_{\max}$ (with strict due dates).
To see this, we utilize binary search for the objective value $\varphi$ and alter the due dates as follows.
For every $j \in [d]$ and every $i \in [\kappa]$, for objective $L_{\max}$ set the new due date $(d')_j^i = d^i_j + \varphi$, and for objective $F_{\max}$ set it to $(d')^i_j = r^i_j + \varphi$.

\paragraph{Weighted Throughput.}
Another objective with a very similar model to $C_{\max}$ is the maximum throughput or minimum weighted penalty, which is to minimize $\sum w_jU_j$.
To see this let us add an auxiliary penalty machine that is intended to collect all late jobs.
This is the only machine for which we introduce an objective and which is of a separate type.
Note that we may assume all late jobs are scheduled after the last critical time on any (normal) machine, and we model this by scheduling them on the newly introduced penalty  machine.
On this machine all jobs have unit size, are released at time $0$, and their due date is $\sum_{j \in [d]} n_j$.
The objective of this machine is to minimize the sum of weights of jobs scheduled on this auxiliary machine, i.e., if $w_j$ is the weight of a job of type $j \in [d]$, the objective is $\min \sum_{j=1}^d w_j x_j$.

\paragraph{Maximizing Minimum Load ($C_{\min}$)}
Let us discuss how to handle the $C_{\min}$ objective which asks to maximize the minimum load of a machine.
The idea is quite similar to $C_{\max}$: guess the optimal value $\bar{C}_{\min}$ and adds a constraint
\begin{equation}
\sum_{j \in [d]} \sum_{C \in \CC} p_j \cdot y_{j,C} \geq \left\lceil \bar{C}_{min} \cdot s^{i,q} \right\rceil \,,
\end{equation}
which clearly ensures that the load of every machine is at least $\bar{C}_{\min}$.
\mkinline{Do we need more stuff about the correctness of the model?}

\medskip

Note that all the alterations above of the previously given model for $R|r^i_j,d^i_j|C_{\max}$ do not affect any of the parameters of the MIMO model.
Thus, the following theorem directly follows from previous argumentation.
\begin{theorem}
\label{thm:FPT:MIMOROtherObjectives}
  Problems $R|r^i_j,d^i_j|L_{\max}$, $R|r^i_j,d^i_j|F_{\max}$, $R|r^i_j,d^i_j|\sum w_jU_j$, and $R|r^i_j,d^i_j|C_{\min}$ in their $\ves$-representation with $m$ machines of $\tau$ types and $d$ job types admit fixed-parameter algorithms for
  \begin{itemize}
    \item parameter $m + d$,
    \item parameter $p_{\max} + d$, and
    \item parameter $d + \tau$ if $p_{\max}$ is given in unary. \qed
  \end{itemize}
\end{theorem}

\subsubsection{$\ell_p$-norm minimization}
Another objective which can be handled by an almost identical model is the $\ell_p$ norm of a load vector of a schedule for integer values of~$p$, which is defined as follows.
In a given schedule $\sigma$, the \emph{load $L$} of a machine is the total time it spends processing jobs, hence if $\vex \in \N^d$ is the compact encoding of the set of jobs scheduled to run on a machine of kind $i$ and with speed $s$, its load is $\vep^i \vex / s$.
The objective $\ell_p$ is to minimize the $\ell_p$ norm of the vector $\veL = (\LL_1, \dots, \LL_m)$ of loads which has one coordinate for each machine equal to the load of this machine.
Note that, for $p \in \N_{\geq 1}$, the norm function itself is a convex but \emph{not} separable convex, since it is defined as $\sqrt[p]{\sum_{i=1}^m (\LL_i)^p}$.
However, it is easy to see that the optima of this function are the same as the optima of the function $\sum_{i=1}^m (\LL_i)^p$ (for $\LL_i \ge 0$), which \emph{is} separable convex.
The function might be fractional (for fractional speeds $s$) and since MIMO requires an objective function which is integral on integral points, this issue must be handled.
Since a simple scaling argument suffices, we defer this discussion to Section~\ref{sec:fractionality}.

Hence, taking the constraints \eqref{eq:Cmax:sumOfExes}--\eqref{eq:Cmax:incompatibleCyclesDisable} and \eqref{eq:CmaxSpeed:cycleVolumeBounds} and defining, for a machine of kind~$i \in [\kappa]$ and speed~$s$, the objective function to be
\[
f^{i,s}(\vex) = \vep^i \vex / s,
\]
we have obtained a MIMO model with a separable convex objective, implying the following modeling lemma and effective theorem:
\begin{lemma}[MIMO model for $R|r^i_j,d^i_j|\ell_p$]
	\label{lem:MIMORlp}
	Let $\mathcal{I}$ be an instance of $R|r^i_j,d^i_j|\ell_p$ in its $\ves$-representation with $m$ machines of $\tau$ machine types, and $d$ job types with maximum job size $p_{\max}$.
	There is a MIMO model $\mathcal{S}$ for $\mathcal{I}$ with extension-separable convex objective functions and with the following values of MIMO parameters:
	\begin{tasks}[style=itemize](3)
		\task $\mathcal{S}(\Delta) = p_{\max}$,
		\task $\mathcal{S}(M) = \mathcal O((d)^2)$,
		\task $\mathcal{S}(d) = d$,
		\task $\mathcal{S}(d^i) = \mathcal O((d)^2)$,
		\task $\mathcal{S}(N) = m$, and
		\task $\mathcal{S}(\tau) = \tau = \kappa \cdot \bar{\tau}$. \qed
	\end{tasks}
\end{lemma}
\begin{theorem}
	The problem $R|r^i_j,d^i_j|\ell_p$ in its $\ves$-representation with $m$ machines of $\tau$ types and $d$ job types admits a fixed-parameter algorithms for
	\begin{itemize}
		\item parameter $m + d$, and
		\item parameter $p_{\max} + d$. \qed
	\end{itemize}
\end{theorem}

\begin{table}[bt]
  \begin{tabular}{| c | p{.85\textwidth} |}
    \hline
    Notation & Meaning \\
    \hline\hline
    $\mathcal{C}$                                               & set of all potential cycles \\
    $\mathcal{C}^{\operatorname{int/ext}}$                      & set of all potential internal / external cycles \\
    $C^{\operatorname{int}}_k$                                  & a potential internal cycle in the interval $\interval{t_k}{t_{k+1}}$ \\
    $C^{\operatorname{ext}}_{k,\ell}$                           & a potential external cycle containing in its interior critical times $t_k, \ldots, t_{\ell}$ \\
    $\mathcal{C}^{\operatorname{ext}}_{\ge k,\ell}$             & set of all potential external cycles containing in their interior critical times $t_{k'}, \ldots, t_{\ell}$ for some $k'$ with $k \le k' \le \ell$ \\
    $\mathcal{C}^{\operatorname{ext}}_{\le k,\ell}$             & set of all potential external cycles containing in their interior critical times $t_{k}, \ldots, t_{\ell}$ (if the inequality is strict, then also $t_{k-1}$) \\
    $\mathcal{C}^{\operatorname{ext}}_{k,\le \ell}$             & set of all potential external cycles containing in their interior critical times $t_{k}, \ldots, t_{\ell'}$ for some $\ell'$ with $k \le \ell' \le \ell$ \\
    $\mathcal{C}^{\operatorname{ext}}_{k,*}$                   & $\mathcal{C}^{\operatorname{ext}}_{k,*}  =  \mathcal{C}^{\operatorname{ext}}_{k,\le |T|-1}$ \\
    $\mathcal{C}^{\operatorname{ext}}_{*,k}$                   & $\mathcal{C}^{\operatorname{ext}}_{*,k}  =  \mathcal{C}^{\operatorname{ext}}_{\le k,k}$ \\
    \hline
  \end{tabular}
  \caption{\label{tab:potentialCycleExtendedNotation}
    Potential cycle notation review.
  }
\end{table}

\subsection{Polynomial Objectives: Total Weighted Completion Time, Flow Time, and Tardiness} \label{sec:polyobj}
When it comes to objectives such as $\sum w_j C_j$, we revisit properties of schedule cycles.
In particular, we give a refined proof of the cycle decomposition lemma (Lemma~\ref{lem:cyclesScheduleCMax}) suitable for polynomial objectives (recall that we call the objectives $\sum w_j C_j$, $\sum w_j F_j$, and $\sum w_j T_j$ polynomial objectives).
We extend the ideas of~\cite{KnopKoutecky2017} who showed that, in a setting without release times and due dates, the $\sum w_j C_j$ objective is expressible as a separable convex function in certain auxiliary variables.
This construction gains intuition by visualizing the objective using 2D Gantt charts.
In order to keep track of how exactly the critical times split the computation (size) of a job scheduled to an external cycle we introduce new auxiliary variables.
Finally, we re-introduce the speeds back to the model, however, this time it is not as simple as in the case of makespan minimization.
This follows from the fact that we need to know \emph{precisely} how a critical time splits the size of a job in an external cycle.
On the other hand, we prove that there is only a limited number (a function of $d$) of options for the ``speed introduced shift''.
Consequently, it is possible to adapt the model (by adding more auxiliary variables) so as to express the value of an objective such as $\sum w_j C_j$ as a separable convex objective function.
We treat the remaining polynomial objectives similarly.

We begin with minimization of sum of weighted completion times.
Formally, we focus on the following problem:
\prob{\textsc{Minimizing Sum of Weighted Completion Times on Unrelated Machines} ($R | r^i_j,d^i_j | \sum w_j C_j$)}
{
	There are $\kappa$ kinds of machines and $d$ types of jobs.
	The number of machines of kind $i \in [\kappa]$ is $\mu^i$ and the number of jobs of type $j \in [d]$ is $n_j$, with $\vemu = (\mu^1, \dots, \mu^\kappa)$ and $\ven = (n_1, \dots, n_d)$, hence there are $m=\|\vemu\|_1$ machines and $n = \|\ven\|_1$ jobs.
	Each job type is specified by its weight $w_j$ and three vectors giving its size, release time, and due date on each machine kind, i.e., for each $j \in [d]$ given are vectors $\vep_j = \left( p^1_j, \ldots, p^\kappa_j \right) \in \left(\N \cup \{\infty\}\right)^\kappa$, $\ver_j = \left( r^1_j, \ldots, r^\kappa_j \right) \in \N^\kappa$, and $\ved_j = \left( d^1_j, \ldots, d^\kappa_j \right) \in \N^\kappa$.
}
{A non-preemptive schedule of all of the jobs on the specified $m$ machines (if one exists) minimizing $\sum_{J\in \mathcal{J}} w(J) C_J$, where $C_J$ is the time when job $J$ finishes and $w(J)$ is its weight.}

The main structural insight behind the previous model has been Lemma~\ref{lem:cyclesScheduleCMax} guaranteeing the existence of an optimal $S$-regular schedule (for appropriate $S$).
We will prove a different version of the aforementioned lemma which takes into account the specifics of the more complicated objectives which are now our focus.
Let us point out one feature of cycles in a schedule.
Recall that the proof of Lemma~\ref{lem:cyclesScheduleCMax} is completely independent of the permutation $\pi$ of job types in a cycle.
However, when we consider polynomial objectives, the ordering of jobs in a cycle plays a role and may depend on (the ``position'' of) the corresponding potential cycle.

The polynomial objectives (other than $\ell_p$-norm, which we have dealt with already) are well known to be in a close relation to the Smith's rule~\cite{Smith1956}, as we shall discuss now.
Take for example $\mathcal{R} = \sum w_jC_j$: there is a weak order $\preceq_{\mathcal{R}}$, which is given by the ratio $w_j / p_j$, such that if jobs of an (internal) cycle are not ordered according to this ratio, then rearranging them yields a schedule of smaller value.
This brings us to define, for each objective $\RR\in \mathfrak{C}_{\text{poly}} \setminus \{\ell_p\}$, each machine kind $i \in [\kappa]$, and each integer $\ell \in [|T^i| - 1]$, an ordering $\preceq_{\RR, \ell}^i$ of the job types, which will satisfy the property that ordering all jobs strictly contained in $[t_\ell, t_{\ell+1}]$ by $\preceq_{\RR, \ell}^i$ results in a schedule which is at least as good as the original one.
Specifically, for $\RR = \sum w_j C_j$ and $\RR = \sum w_j F_j$, we define $\preceq_{\RR, \ell}^i$ to be the ordering of $[d]$ by the ratios $w_j/p^i_j$ non-increasingly.
For $\RR = \sum w_j T_j$, define $w^i_{\ell,j}$ to be $w_j$ if $t_\ell \geq d_j^i$ and $0$ otherwise, and define $\preceq_{\RR, \ell}^i$ to be the ordering of $[d]$ by the ratios $w^i_{\ell,j}/p^i_j$ non-increasingly.
Note that for the objective $\sum w_jC_j$ we have $\preceq^i_{\RR,1} = \cdots = \preceq^i_{\RR,|T^i| - 1}$, since the Smith ratio does not depend on $r^i_j$ or $d^i_j$ (some of the jobs types cannot be scheduled due to $\chi^i_{j,C}$ restrictions~\eqref{eq:chiDefinition}, though).
On the other hand, when it comes to $\sum w_jT_j$, the orders may differ, since the integer $\ell \in [|T|]$ determines whether a job is late (in the interval $\interval{t_{\ell}}{t_{\ell + 1}}$), which determines whether it contributes to the objective.
We omit the superscript $i$ if the machine kind is clear from context.
Next, we show that these definitions satisfy the property claimed before, which is formally phrased as follows:
\begin{lemma} \label{lem:ordered}
Fix a machine kind $i \in [\kappa]$.
Let $\RR \in \mathfrak{C}_{\text{poly}} \setminus \{\ell_p\}$, let $\sigma'$ be a schedule of jobs on a single machine of kind $i$, let $\ell \in [|T^i|-1]$, and let $\sigma$ be the schedule obtained from $\sigma'$ by rearranging all jobs which are strictly contained in $[t_\ell, t_{\ell+1}]$ by the ordering $\preceq_{\RR, \ell}^i$.
Then the objective function value of $\sigma$ is at most that of $\sigma'$.
\end{lemma}
\begin{proof}
Let $\JJ$ be the set of all jobs, let $\JJ_{\ell}$ be the set of jobs strictly contained in $[t_\ell, t_{\ell+1}]$, and let $\bar{\JJ}_\ell = \JJ\setminus \JJ_\ell$.
First note that the definition of each $\RR \in \mathfrak{C}_{\text{poly}} \setminus \{\ell_p\}$ is such that the objective function value of a schedule can be decomposed into the contributions of $\JJ_\ell$ and $\bar{\JJ}_\ell$.
Hence, if we only change the order of $\JJ_\ell$ but not the order of $\bar{\JJ}_\ell$, and if the contribution of the jobs in $\JJ_\ell$ does not increase after this change, we are done.
Second, the fact that reordering $\JJ_\ell$ according to $\preceq_{\RR, \ell}^i$ is optimal is exactly the statement of Smith's rule~\cite{Smith1956}, concluding the proof.
\end{proof}

In the following we shall focus on $\sum w_j C_j$, but note that $\sum w_jC_j$ and $\sum w_jF_j$ are equivalent if the release times of the jobs are independent of the assignment.
In our case, these two objective have a difference that is a linear function of the assignment variables $\vex$, and thus the model we present for $\sum w_j C_j$ exhibits a model for $\sum w_j F_j$ as well.
For more details cf.~\cite[Chapter~2]{pinedo2012}.

\paragraph{2D Gantt Charts.}
It is convenient to visualize ordered objectives in a so called 2D Gantt chart; see Figure~\ref{fig:twoCyclesMerging}.
Take for example $\sum w_jC_j$ and let us create a chart in two dimensions $x,y$ for a set of jobs $\mathcal{J}$ and an admissible schedule $\sigma$ of $\mathcal{J}$ (i.e., we focus on one machine).
Each job $J$ is represented by a rectangle of height $w(J)$ and width $p(J)$.
The rectangle for a particular job $\hat{J}$ is drawn at a height equal to the sum of weights of all the jobs scheduled after $\hat{J}$ in $\sigma$, that is, its bottom corners have $y = \sum_{J \in \mathcal{J} : C_J > C_{\hat{J}}} w(J)$, where $C_J$ is a completion time of a job $J$.
The left corners of the rectangle for $\hat{J}$ are at $x$-position equal to the sum of sizes of all jobs scheduled before $\hat{J}$ in $\sigma$, that is, $x = \sum_{J \in \mathcal{J} : C_J < C_{\hat{J}}} p(J)$.
Finally, the value $\sum w_j C_j$ is exactly the area of the rectangles and between these rectangles and the two axes $x,y$.
The value of all of the ordered objectives is expressible in a similar way.
For more discussion on 2D Gantt charts and $\sum w_jC_j$ see e.g.~\cite{KnopKoutecky2017,GoemansWilliamson2000}.
For more discussion on 2D Gantt charts and $\sum w_jT_j$ see Section~\ref{sec:tardiness}.

\begin{lemma}
\label{lem:cyclesScheduleOrderedObjective}
  Fix a machine of kind $i$ and $\RR \in \mathfrak{C}_{\text{poly}} \setminus \{\ell_p\}$, and let\/ $\vex \in \N^d$.
  If there is a schedule $\sigma'$ of\/ $\vex$ with value $\varphi$ under $\RR$, then there is an $\mathbb{N}$-regular schedule $\sigma$ of $\vex$ with value at most~$\varphi$ under $\mathcal{R}$ such that each cycle of $\sigma$ is a realization of some potential cycle $C \in \CC$, and the jobs in each internal cycle $C = C^{\operatorname{int}}_k$ of $\sigma$ are ordered by $\preceq^i_{\mathcal{R},k}$.
\end{lemma}
\begin{proof}
  The proof of this lemma goes along the lines of the proofs of Lemmas~\ref{lem:cyclesScheduleCMax} and~\ref{lem:PCmaxCanonicalCycleDecomposition}, with the difference that we have to pay attention to the orders (i.e., job type permutations) of individual internal cycles.
  Let $\DD'$ be some cycle decomposition of $\sigma'$.
  We perform the following changes on $\sigma'$ and $\DD'$:
  \begin{enumerate}
  	\item \label{lem:cyclesScheduleOrderedObjective:1} split all external cycles so that afterwards each external cycle contains at most one job (this only affects $\DD'$),
  	\item \label{lem:cyclesScheduleOrderedObjective:2} for each $\ell \in [|T|-1]$, reorder all jobs strictly contained in $[t_\ell, t_{\ell+1}]$ by $\preceq_{\RR, \ell}$,
  	\item \label{lem:cyclesScheduleOrderedObjective:3} left-align the schedule in exactly the same way as in the end of Lemma~\ref{lem:cyclesScheduleCMax}, which makes the schedule $\N$-regular.
  \end{enumerate}
  Call the resulting schedule $\sigma$ and its cycle decomposition $\DD$.
  It is clear that $\sigma$ and $\DD$ are $\mathbb{N}$-regular, so it remains to argue that our construction (points~\ref{lem:cyclesScheduleOrderedObjective:1}--\ref{lem:cyclesScheduleOrderedObjective:3}) does not increase the objective.
  Point~\ref{lem:cyclesScheduleOrderedObjective:1} only changes the cycle decomposition, so it does not affect the objective.
  Point~\ref{lem:cyclesScheduleOrderedObjective:2} is a safe operation by Lemma~\ref{lem:ordered}.
  For an example of this step with $\mathcal{R} = \sum w_jC_j$ and two cycles see Figure~\ref{fig:twoCyclesMerging}.
  Regarding point~\ref{lem:cyclesScheduleOrderedObjective:3} note that left-aligning only possibly decreases the completion times of all jobs, so it does not increase the objective.
\end{proof}

  \begin{figure}[bt]
    \begin{center}
        \usetikzlibrary{fit}

\begin{tikzpicture}[scale=.6]
\pgfmathsetmacro\unitProcessingTime{1cm*.6}
\pgfmathsetmacro\unitWeight{.5cm*.6}

\tikzstyle{job}=[rectangle,draw,ultra thick,fill=orange!30]
\tikzstyle{job3-1}=[job,minimum width=\unitProcessingTime,minimum height=3*\unitWeight]
\tikzstyle{job2-1}=[job,minimum width=\unitProcessingTime,minimum height=2*\unitWeight]
\tikzstyle{job1-1}=[job,minimum width=\unitProcessingTime,minimum height=\unitWeight]
\tikzstyle{job2-2}=[job,minimum width=2*\unitProcessingTime,minimum height=2*\unitWeight]
\tikzstyle{job2-3}=[job,minimum width=3*\unitProcessingTime,minimum height=2*\unitWeight]
\tikzstyle{job1-4}=[job,minimum width=4*\unitProcessingTime,minimum height=\unitWeight]

\tikzstyle{cycle}=[draw,inner sep=0,dotted]
\tikzstyle{trendLine}=[blue,thick]

\fill[black!10] (-8.5,8.75) -- (-.5,8.75) -- (-.5,3.75) -- (-8.5,3.75) -- cycle;
\begin{scope}[local bounding box=c2]
\fill[black!10] (-.5,8) -- (0,8) -- (1,6.5) --  (2,5.25) -- (2.5, 4.75) -- (2.5, 3.75) -- (-.5,3.75) -- cycle;
\node[job3-1] at (0,8) {};
\node[job3-1] at (1,6.5) {};
\node[job2-1] at (2,5.25) {};
\node[job2-3] at (4,4.25) {};
\draw[trendLine] (-.5,8.75) -- (1.5,5.75) -- (2.5, 4.75) -- (5.5, 3.75);
\end{scope}
\node[cycle,fit=(c2),label={[yshift=-.8cm,xshift=-.65cm]45:$C^2$}] {};

\begin{scope}[local bounding box=c1,xshift=-8cm,yshift=6.5cm]
\fill[black!10] (-.5,4.5) -- (0,4.5) -- (1,4) -- (2.5, 3.25) -- (3.5,2.75)  -- (3.5,2.25) -- (-.5,2.25) -- cycle;
\node[job1-1] at (0,4.5) {};
\node[job1-1] at (1,4) {};
\node[job2-2] at (2.5,3.25) {};
\node[job1-4] at (5.5,2.5) {};
\draw[trendLine] (-.5,4.75) -- (0,4.5) -- (1,4) -- (2.5, 3.25) -- (3.5,2.75)  -- (7.5,2.25);
\end{scope}
\node[cycle,fit=(c1),label={[yshift=-.8cm,xshift=1cm]45:$C^1$}] {};

\fill[black] (-13,3.8) -- (-13,11.2) -- (-8.5,11.2) -- (-8.5,3.8) -- cycle;

\draw[dashed] (-.5,11.5) -- (-.5,3.5);

\begin{scope}[yshift=-8.5cm]
\fill[black] (-13,3.8) -- (-13,11.2) -- (-8.5,11.2) -- (-8.5,3.8) -- cycle;
\end{scope}

\begin{scope}[local bounding box=c,xshift=-8cm,yshift=-6cm]
\fill[black!10] (-.5,7.25) -- (.5,7.25) -- (1.5,5.75) -- (2.5,4.75) -- (3.5,4.25) -- (4.5,3.75) -- (6.5, 2.75) -- (9.5, 1.75) -- (9.5,1.25) -- (-.5,1.25) -- cycle;
\node[job3-1] at (0,8) {};
\node[job3-1] at (1,6.5) {};
\node[job2-1] at (2,5.25) {};
\node[job1-1] at (3,4.5) {};
\node[job1-1] at (4,4) {};
\node[job2-2] at (5.5,3.25) {};
\node[job2-3] at (8,2.25) {};
\node[job1-4] at (11.5,1.5) {};
\draw[trendLine] (-.5,8.75) -- (.5,7.25) -- (1.5,5.75) -- (2.5,4.75) -- (3.5,4.25) -- (4.5,3.75) -- (6.5, 2.75) -- (9.5, 1.75) -- (13.5,1.25);
\end{scope}
\node[cycle,fit=(c),label={[yshift=-.8cm,xshift=.3cm]45:$C^1 \cup C^2$}] {};

\draw[thick] (-10.5,11.5) -- (-10.5,-5);
\draw[thick] (6.5,11.5) -- (6.5,-5);
\node at (6.5,-5.2) {$t_{\ell+1}$};
\node at (-10.5,-5.2) {$t_{\ell}$};

\node at (-2,-5.2) {processing time};

\end{tikzpicture}
    \end{center}
    \caption{\label{fig:twoCyclesMerging}%
        Two adjacent internal cycles $C^1,C^2$ between consecutive critical times (upper part) and a cycle obtained by merging them (bottom part).
        Cycle $C^1$ consists of two jobs with size 1 and weight 1, a job with size 2 and weight 2, and a job with size 3 and weight 1.
        Cycle $C^2$ consists of two jobs with size 1 and weight 3, a job with size 1 and weight 2, and a job with size 3 and weight 2.
        Jobs in all the displayed cycles are ordered according to Smith's rule (as witnessed by the blue line whose slope is non-decreasing).
        The objective $\sum w_jC_j$ corresponds to the sum of the areas of the rectangles (orange), below them (gray), and to the left of them (black).
    }
  \end{figure}
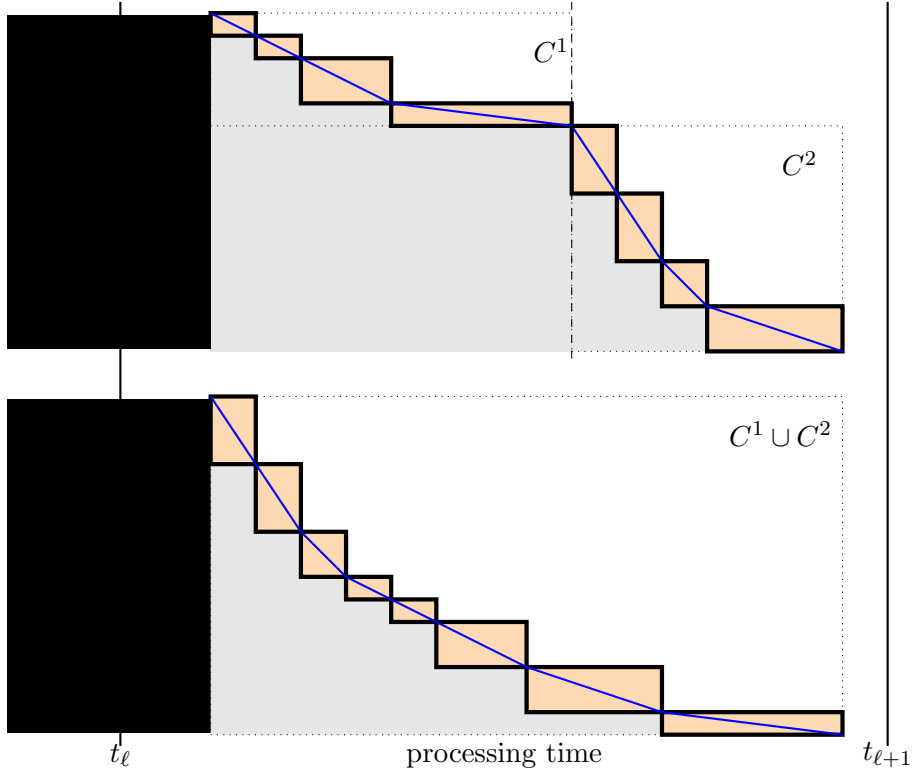

\subsubsection{Sum of Weighted Completion Times}
We use the structure provided by Lemma~\ref{lem:cyclesScheduleOrderedObjective} to express $P^i$ and $f^i$, the polyhedron of all possible configurations for a single machine of kind $i$ and the corresponding objective function, respectively, by a model similar to the one for makespan minimization.
As already shown by Knop and Koutecký~\cite[Corollary~1]{KnopKoutecky2017}, for $\mathcal{R} = \sum w_jC_j$ and without release times and due dates it is possible to express the value of a schedule $\sigma$ under $\mathcal{R}$ as a separable convex function in certain auxiliary variables.
Here we add these variables in order to express the value of the scheduling objective $\sum w_jC_j$ as an extension-separable convex objective over $P^i$.
We stress here that the earlier work of Knop and Koutecký~\cite{KnopKoutecky2017} only deals with the case when there is a single cycle, which makes the structure of the objective and thus the modeling drastically simpler, as we are about to see.
Next, we shall add new variables to our model and then we introduce some further conditions to the model in order to bind them with the variables already present in the model given in the previous part~\eqref{eq:Cmax:sumOfExes}--\eqref{eq:Cmax:cycleVolumeBounds}.

\paragraph{New Variables.}
We introduce new variables into our model in order to express, for each non-empty external cycle in a schedule, how many time units of the total size of the job contained in this cycle are processed before the rightmost critical time contained in its interior and how many after.
Let $C$ be an external cycle and consider the model given by~\eqref{eq:Cmax:sumOfExes}--\eqref{eq:Cmax:cycleVolumeBounds}.
Recall we have a variable $y^i_{j,C}$ which is set to $1$ if a job of type $j$ is scheduled in $C$ (consequently, $C$ is nonempty) and a variable $z^i_C$ which is set to $1$ if any job is scheduled in $C$.
We introduce $2p_{\max}$ binary variables for every potential external cycle $C \in \mathcal{C}^{\text{ext}}$ into our model: $y^i_{C, R,p}$ and $y^i_{C, L, p}$, for each $p \in [p_{\max}]$.
The intended meaning (see Lemma~\ref{lem:sumw_jC_j:criticalTimeSplitExternalCycles} below) is that $y^i_{C, L, p} = 1$ if and only if there is a job scheduled to a realization of an external cycle $C$ such that exactly $p$ units of its size are already processed at time $t_k$, the right critical time of $C$ (i.e., $C = C^{\operatorname{ext}}_{\ell,k}$ for some $\ell \in [k]$).
Now, we enforce the intended meaning of the new variables in the following way:
\begin{align}
  \left( \sum_{p \in [p_{\max}]} p \cdot y^i_{C,L,p} \right) + \left( \sum_{p \in [p_{\max}]} p \cdot y^i_{C,R,p} \right) &= \sum_{j \in [d]} p^i_j y^i_{j,C}    \qquad\qquad &\forall C \in \mathcal{C}^{\text{ext}} \label{eq:sumw_jC_j:first} \\
  \sum_{p \in [p_{\max}]} y^i_{C,L,p} &= z^i_C  & \forall C \in \mathcal{C}^{\text{ext}}  \label{eq:sumw_jC_j:second} \\
  \sum_{p \in [p_{\max}]} y^i_{C,R,p} &= z^i_C  & \forall C \in \mathcal{C}^{\text{ext}}  \label{eq:sumw_jC_j:third} \\
  0 \le y^i_{C,L,p}, y^i_{C,R,p} &\le 1   & \forall C \in \mathcal{C}^{\text{ext}}, \forall p \in [p_{\max}] \label{eq:sumw_jC_j:fourth} \,.
\end{align}
Observe that if now $z^i_C = 1$, then exactly one variable $y^i_{C,L,p}$ is set to $1$ for some $p \in [p_{\max}]$ and the same holds for some $y^i_{C,R,p'}$, where again $p' \in [p_{\max}]$.
Furthermore, in this case we have by~\eqref{eq:sumw_jC_j:first} that
\[
p \cdot y^i_{C,L,p} + p' \cdot y^i_{C,R,p'} = p^i_j\,, \quad\text{where } p^i_j \text{ is the size of the job scheduled to } C.
\]
Finally, we alter the set of constraints~\eqref{eq:Cmax:cycleVolumeBounds}.
\begin{equation}
  \sum_{t_\ell \lhd C \lhd t_k} p^i_j \cdot y^i_{j,C}
  +
  \sum_{C \in \mathcal{C}^{\operatorname{ext}}_{< \ell,k}} p \cdot y^i_{C,L,p}
  \le
  t_k - t_\ell
  \qquad \forall k \in [|T^i|] \,, 1 \le \ell < k \label{eq:sumw_jC_j:cycleVolumeBounds}
\end{equation}

\begin{algorithm}[tb]
	\SetKwProg{Def}{def}{:}{}
	\SetKwFunction{cycleHandler}{handleCycle}
	\SetKwFunction{scheduleEnd}{endOf}
	\DontPrintSemicolon
	\Def{\cycleHandler{$i, \sigma, C, \vey_C$}}{
		\If{$C \in \mathcal{C}^{\operatorname{int}}$}{
			\nlset{Order}\label{alg:sigmaFromYWithOrder:leftCritical}\ForEach{$j \in [d]$ in order $\preceq^i_{\mathcal{R},\leftCritical(C)}$}{
				\For{$\ell = 1$ \KwTo $y_{j,C}$}{
					\nlset{Next}\label{alg:sigmaFromYWithOrder:tofJ} $t \leftarrow \max($\scheduleEnd{$\sigma$}$, \leftCritical(C))$ \;
					$\sigma \leftarrow \sigma \cup \left\{\left( j, [t, t + p^i_j] \right)\right\}$ \;
				}
			}
		}
		\Else{
			Let $p,p'$ be such that $y_{C,L,p} = 1$ and $y_{C,R,p'} = 1$ \;
			Let $k$ be such that $C \in \mathcal{C}^{\operatorname{ext}}_{*,k}$ \;
			\nlset{ext}\label{alg:sigmaFromYWithOrder:externalCritical}$\sigma(\vey) \leftarrow \sigma(\vey) \cup \left\{\left( j, [t_k - p, t_k + p'] \right)\right\}$ \;
		}
	}

	\caption{\label{alg:sigmaFromYWithOrder}
		We only redefine the \texttt{handle\_cycle} function, the rest of the algorithm is identical to Algorithm~\ref{alg:sigmaFromY}.
		The algorithm computes a left aligned schedule from a vector $\vey$ by placing jobs in internal cycles in the orders $\preceq^i_{\mathcal{R},k}$.
	}
\end{algorithm}

The following technical lemma shows that the variables we have introduced to the model in this section have their intended meaning in $\sigma(\vey, \mathcal{R})$ as produced by Algorithm~\ref{alg:sigmaFromYWithOrder}.
When $\RR$ is fixed we omit it from ``$\sigma(\vey, \RR)$'' and continue to write $\sigma(\vey)$.
Also recall that when a machine of kind $i \in [\kappa]$ is fixed, we often omit the index or parameter $i$, such as when we write $\sigma(\vey)$ instead of $\sigma(i,\vey^i)$.

Consider Algorithm~\ref{alg:sigmaFromYWithOrder}.
Observe that, since we have special branch for external cycles, it is not clear (as it was {e.g.} for Algorithm~\ref{alg:sigmaFromY}) that jobs are not overlapping.
Neither we can be sure that the resulting schedule is regular.
On the other hand, since we deal with the jobs assigned to an internal cycle in nearly the same way as we did in Algorithm~\ref{alg:sigmaFromY}, it still makes sense to define for such a job $J$ the critical time $t(J)$, as this is a usefull notion in proofs.
Indeed we would like to have this for all jobs the vector $\vex$ assigns to a particular machine; this can be achieved by a slight modification of the definition of the critical time $t(J)$.
Let $J$ be a job assigned to a (potential) cycle $C$
\[
  t(J) =
  \begin{cases}
    \operatorname{left}(C)  & \text{if $J$ starts in a critical time in $\sigma(\vey)$ (note that it is then $\operatorname{left}(C)$ if $J$ is assigned to the cycle $C$)} \\
    t(\hat{J})              & \text{for a predecessor $\hat{J}$ of $J$ in } \sigma(\vey), otherwise
  \end{cases} \,.
\]

\begin{lemma}\label{lem:sumw_jC_j:criticalTimeSplitExternalCycles}
  Fix a machine of kind $i \in [\kappa]$ and let $\vex \in \mathbb{N}^d$.
  There exists an $\mathbb{N}$-regular schedule of $\vex$ if and only if there exists $\vey,\vez$ such that all of the constraints \eqref{eq:Cmax:sumOfExes}--\eqref{eq:Cmax:incompatibleCyclesDisable} and \eqref{eq:sumw_jC_j:first}--\eqref{eq:sumw_jC_j:cycleVolumeBounds} are satisfied.
  Moreover, $\sigma(\vey)$ is one such $\mathbb{N}$-regular schedule.
  Furthermore, $y^i_{C,L,p} = 1$ and $y^i_{C,R,p'} = 1$ for a potential external cycle $C = C^{\operatorname{ext}}_{\ell,k}$ for $\ell,k \in \left\{ 2, \ldots, T^i \right\}$ with $\ell \le k$
  if and only if
  there is a job $J$ with size $p^i_J = p+p'$ assigned to a realization of $C$ in $\sigma(\vey)$ such that at time $t_k$ exactly $p$ units of the total size of the job $J$ on a machine of kind $i$ are processed.
\end{lemma}
\begin{proof}
  We begin with showing that for an $\N$-regular schedule of $\vex$ there exists $\vey,\vez$ satisfying the model.
  We know that an $\N$-regular schedule comes with its cycle decomposition and that this allows us to assign ``the old variables'' (i.e., those presented already in our model for $C_{\max}$).
  Thus, using the same arguments as in Lemma~\ref{lem:leftAlignedScheduleCMax} we can assume that \eqref{eq:Cmax:sumOfExes}--\eqref{eq:Cmax:incompatibleCyclesDisable} are satisfied.
  In order to assign new variables and satisfy the conditions \eqref{eq:sumw_jC_j:first}--\eqref{eq:sumw_jC_j:cycleVolumeBounds} we will use further properties of regular schedules.
  Since we have introduced new variables only for external cycles, we will now focus on those.
  Observe that if an external cycle is empty, we have to set all of the variables associated with it to $0$ and this way we satisfy \eqref{eq:sumw_jC_j:first}--\eqref{eq:sumw_jC_j:fourth}.
  Let $C$ be a nonempty external cycle of the assumed regular schedule and let $J$ be the job of type $j$ assigned to it.
  Suppose $C$ is a realization of the potential cycle $C^{\operatorname{ext}}_{\ell,k}$.
  We claim that there exist $p,p' \in \mathbb{N}$ with $p+p'=p^i(J)$ such that the completion time of $J$ in the regular schedule for $\vex$ is exactly $t_k+p'$.
  This follows directly from the regularity of the schedule and integrality of all of the job sizes.
  Now, we set
  \[
    y^i_{C,L,p} = 1 \qquad\qquad\text{and}\qquad\qquad y^i_{C,R,p'} = 1
  \]
  and we set all other newly introduced variables for $C$ to $0$.
  It is straightforward to check that such an assignment satisfies \eqref{eq:sumw_jC_j:first}--\eqref{eq:sumw_jC_j:fourth}.
  Thus, it remains to verify the conditions \eqref{eq:sumw_jC_j:cycleVolumeBounds}.
  However, this is not hard, since we have assumed a (regular) schedule for $\vex$ and \eqref{eq:sumw_jC_j:cycleVolumeBounds} only assures that it is possible to fit all cycles that are between a pair of critical times into the time window given by the two critical times leaving the correct time windows for the jobs assigned to external cycles.
  This finishes the proof of the first part.

  Now, we check the validity of the presented model.
  To that end, we check that the schedule $\sigma(\vey)$ produced by Algorithm~\ref{alg:sigmaFromYWithOrder} (note that the difference between this and Algorithm~\ref{alg:sigmaFromY} is rather subtle -- this time we have to schedule jobs assigned to internal cycle in a specific order and schedule jobs to external cycles according to the newly introduced variables) is admissible.
  Observe that, since the order does not affect external cycles (as there is only one job assigned to these), we only have to check validity of the line labeled~\ref{alg:sigmaFromYWithOrder:leftCritical} for internal cycles.
  If $C = C^{\operatorname{int}}_k$ is a cycle (with $\vey_C \neq \bm{0}$), the line labeled~\ref{alg:sigmaFromYWithOrder:leftCritical} uses the correct order, since indeed we have $k = \leftCritical(C)$.
  Note that all implications of constraints \eqref{eq:Cmax:sumOfExes}--\eqref{eq:Cmax:incompatibleCyclesDisable} are the same as in Lemma~\ref{lem:modelPlacesCyclesCorrectly}, that is,
  \begin{itemize}
    \item $z_C = 1$ for an external cycle $C$ if and only if there exists $j \in [d]$ such that $y_{j,C} = 1$,
    \item if $C, C'$ are noncompatible cycles, then at least one of them is empty, and
    \item all jobs specified in $\vex$ are assigned to some (potential) cycle.
  \end{itemize}
  Now, observe that, due to \eqref{eq:sumw_jC_j:second} and \eqref{eq:sumw_jC_j:third}, if $z_C = 1$ for an external cycle, then there exists $p,p' \in [p_{\max}]$ such that
  \begin{align*}
    y_{C,L,p}  &= 1 &\text{and}   &&y_{C,L,\hat{p}} &= 0   &&\forall \hat{p} \in [p_{\max}] \setminus \{p\} \\
    y_{C,R,p'} &= 1 &\text{and}   &&y_{C,R,\hat{p}} &= 0   &&\forall \hat{p} \in [p_{\max}] \setminus \left\{p'\right\} \,.
  \end{align*}
  Furthermore, it follows from \eqref{eq:sumw_jC_j:first} that we have $p + p' = p^i_j$ if $y_{j,C} = 1$, that is, if a job of type $j$ is meant to be scheduled to an external cycle $C = C^{\operatorname{ext}}_{\ell,k}$, then $p,p'$ are the time this job is processed prior to and after $t_k$, respectively.

  Finally, we use \eqref{eq:sumw_jC_j:cycleVolumeBounds} to prove that $\sigma(\vey)$ is admissible.
  Recall that we have that $\sigma(\vey)$ cannot schedule a job to start before its release date.
  Suppose a job $J$ is assigned by $\sigma(\vey)$ to be in a realization of an internal potential cycle $C^{\operatorname{int}}_{\hat{k}}$.
  Then, \eqref{eq:sumw_jC_j:cycleVolumeBounds} applied for $t_\ell = t(J)$ and $t_k = t_{\hat{k}+1}$ (and $t_\ell = t_{\hat{k}}$ and $t_k = t_{\hat{k} + 1}$) implies that it is possible to fit $J$ in the intended interval.
  Suppose that the job $J$ is assigned by $\sigma(\vey)$ to be in a realization of an external potential cycle $C^{\operatorname{ext}}_{\ell',k}$ and let $p,p'$ be as above.
  Then, \eqref{eq:sumw_jC_j:cycleVolumeBounds} applied for $t_\ell = t(J)$ implies that it is possible to process at least $p$ time units of $J$ in the interval prior to $t_k$.
  Clearly, $\sigma(\vey)$ is an $\mathbb{N}$-regular schedule.
  We conclude that $\sigma(\vey)$ is admissible and has the desired properties, so the lemma follows.
\end{proof}

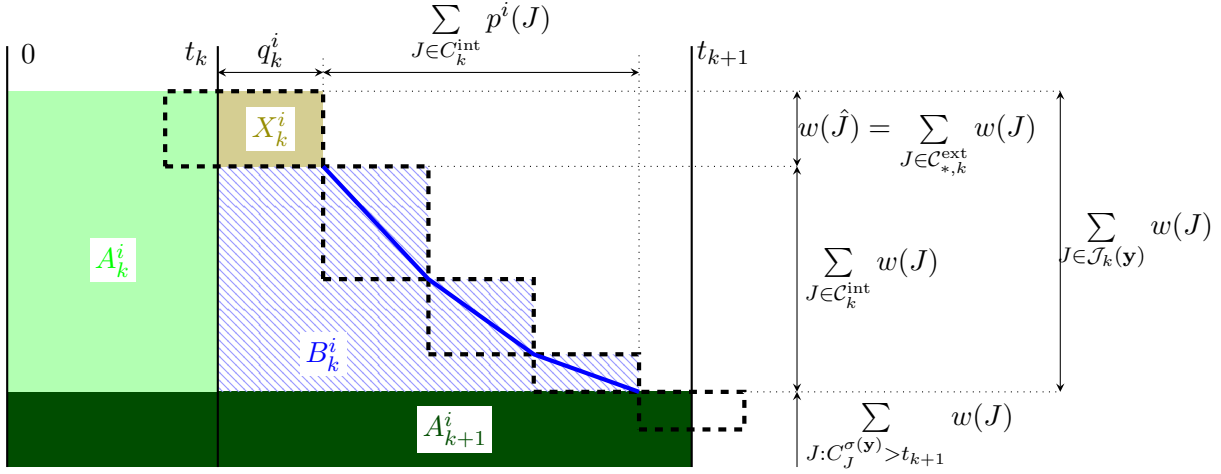
\begin{figure}[bt]
     \usetikzlibrary{patterns}

\begin{tikzpicture}[yscale=.5,xscale=.7]
\tikzstyle{job}=[ultra thick,dashed]
\tikzstyle{fakeJob}=[dotted,red,ultra thick]

\fill[green!30!black] (-4,-1) -- (-4,1) -- (9,1) -- (9,-1) -- cycle;
\fill[green!30] (0,1) -- (0,9) -- (-4,9) -- (-4,1) -- cycle;

\fill[olive!40] (0,9) -- (0,7) -- (2,7) -- (2,9) -- cycle;

\fill[pattern=north west lines,pattern color=blue!40] (0,7) -- (2,7) -- (4,7) -- (4,4) -- (6,4) -- (6,2) -- (8,2) -- (8,1) -- (0,1) -- cycle;

\draw[thick] (0,10.2) -- (0,-1);
\draw[thick] (9,10.2) -- (9,-1);
\draw[thick] (-4,10.2) -- (-4,-1);
\node at (-.4,10) {$t_k$};
\node at (9.6,10) {$t_{k+1}$};
\node at (-3.6,10) {$0$};

\draw[job] (-1,7) -- (-1,9) -- (2,9) -- (2,7) -- cycle;
\draw[job] (2,7) -- (4,7) -- (4,4) -- (2,4) -- cycle;
\draw[job] (4,4) -- (6,4) -- (6,2) -- (4,2) -- cycle;
\draw[job] (6,2) -- (8,2) -- (8,1) -- (6,1) -- cycle;
\draw[job] (8,1) -- (10,1) -- (10,0) -- (8,0) -- cycle;

\draw[blue, ultra thick] ((2,7) -- (4,4) -- (6,2) -- (8,1);

\node[fill=white!10,text=blue,inner sep=2pt] at (2,2) {$B^i_k$};
\node[fill=white!10,text=green!30!black,inner sep=2pt] at (4.5,0) {$A^i_{k+1}$};
\node[fill=white!10,text=green,inner sep=2pt] at (-2,4.5) {$A^i_k$};
\node[fill=white!10,text=olive,inner sep=2pt] at (1,8) {$X^i_k$};

\draw[dotted] (4,7) -- (11,7);
\draw[dotted] (2,9) -- (16,9);
\draw[dotted] (10,1) -- (16,1);

\draw[dotted] (2,9) -- (2,10);
\draw[dotted] (8,2) -- (8,10);

\draw[->,>=stealth] (11,-1) to node[midway,xshift=1.5cm,yshift=-.1cm] {$\sum\limits_{J : C^{\sigma(\vey)}_J > t_{k+1}} w(J)$} (11,1);
\draw[<->,>=stealth] (11,1) to node[midway,xshift=1cm] {$\sum\limits_{J \in \mathcal{C}^{\operatorname{int}}_k} w(J)$} (11,7);
\draw[<->,>=stealth] (11,7) to node[midway,xshift=1.6cm,yshift=-.2cm] {$w(\hat{J}) = \sum\limits_{J \in \mathcal{C}^{\operatorname{ext}}_{*,k}} w(J)$} (11,9);
\draw[<->,>=stealth] (0,9.5) to node[midway,yshift=.3cm] {$q^i_k$} (2,9.5);
\draw[<->,>=stealth] (2,9.5) to node[midway,yshift=.5cm] {$\sum\limits_{J \in C^{\operatorname{int}}_k} p^i(J)$} (8,9.5);
\draw[<->,>=stealth] (16,1) to node[midway,xshift=1cm] {$\sum\limits_{J \in \mathcal{J}_k(\vey)} w(J)$} (16,9);
\end{tikzpicture}
      \caption{\label{fig:sumw_jC_jSegment}%
         A 2D Gantt chart visualization of the decomposition of the objective $\sum w_j C_j$ between two consecutive critical times.
      }
\end{figure}

\paragraph{Objective Function.}
We decompose the objective function (2D Gantt chart) into several areas and show that each of them is expressible as a separable convex function in $\vey$ and auxiliary variables derived from $\vey$.
Before we do so, we formally describe all of the areas.
There are three areas for each interval $\interval{t_{k-1}}{t_k}$.
We shall claim that, for a feasible vector $\vey$, the value of its associated schedule $\sigma(\vey,\sum w_j C_j)$ is
\begin{equation}\label{eq:sumw_jC_j:objectiveToAreasEquation}\tag{objective}
  \sum_{J \in \mathcal{J}(\vex)} w(J) \cdot C^{\sigma(\vey)}_J
  =
  \sum_{k = 1}^{|T^i| - 1} \left( A^i_k(\vey) + B^i_k(\vey) + X^i_k(\vey) \right) \,,
\end{equation}
where \( \mathcal{J}(\vex) \) denotes the set of jobs assigned by $\vex$ to the fixed machine of kind $i$ and \( C^{\sigma(\vey)}_{J} \) is the completion time of job $J$ in the schedule $\sigma(\vey)$.
In what follows, please refer to Figure~\ref{fig:sumw_jC_jSegment}.
Let $\hat{C} \in \mathcal{C}^{\operatorname{ext}}_{*,k}$ be the nonempty external cycle whose last interior critical time in $\sigma(\vey)$ is $t_k$, if such a cycle exists in $\sigma(\vey)$.
Denote $\hat{J}$ the job assigned to the external cycle $\hat{C}$ if $\hat{C}$ is defined and let $\hat{J}$ be any job otherwise.
Define $q^i_k$ to be the amount of time for which the job $\hat{J}$ assigned to an external cycle $\hat{C}$ is processed after the critical time $t_k$ if $\hat{C}$ is defined and $0$ otherwise, that is, we have $q^i_k = C^{\sigma(\vey)}_{\hat{J}} - t_{k}$.
let $\mathcal{J}_k(\vey)$ be the set of jobs with completion time in the interval $\interval[open left]{t_k}{t_{k+1}}$, that is, $\mathcal{J}_k(\vey) = \left\{ J \in \mathcal{J}(\vey) \mid t_k < C^{\sigma(\vey)}_J \le t_{k+1} \right\}$.
Note that $\hat{J} \in \mathcal{J}_k(\vey)$.
Intuitively, we are going to split the contribution of each job towards the total objective exactly at the closest critical time preceeding its completion.\mkcom{I don't understand this sentence.}\dkcom{better now?}
\begin{itemize}
  \item
  The area $A^i_k$ is a rectangle of width $t_{k}$ and height equal to the sum of weights of all jobs $J$ with completion time $C^{\sigma(\vey^i)}_J$ in the interval $\interval[open left]{t_k}{t_{k+1}}$, that is,
  \[
  A^i_k
  =
  t_k \cdot \sum_{J \in \mathcal{J}_k(\vey)} w(J) \,.
  \]
  \item
  The area $B^i_k$ is equal to sum of weighted completion times in an auxiliary scheduling problem, where we have a single machine of kind $i$ and the task is to schedule all jobs contained in the cycle $C^{\operatorname{int}}_k$.
  All jobs are released at time $q^i_k$ and have their original weight.
  Formally, we have
  \[
    B^i_k
    =
    \sum_{J \in C^{\operatorname{int}}_k} w(J) \cdot \left( C^{\sigma(\vey)}_{J} - t_k \right) \,.
  \]
  \item
  Finally, the area $X^i_k$ is a rectangle of width $q^i_k$ and height $w(\hat{J})$.
  That is, \mbox{$X^i_k = \left( C^{\sigma(\vey)}_{\hat{J}} - t_k \right) \cdot w(\hat{J})$}.
  Clearly, we have $X^i_1 = 0$.
\end{itemize}
Combining all the above described areas we obtain equation~\eqref{eq:sumw_jC_j:objectiveToAreasEquation} for the value of the objective function:
\begin{lemma}\label{lem:sumw_jC_j:objectiveToAreasEquationHolds}
  Let $(\vex,\vey,\vez)$ satisfy \eqref{eq:Cmax:sumOfExes}--\eqref{eq:Cmax:incompatibleCyclesDisable} and \eqref{eq:sumw_jC_j:first}--\eqref{eq:sumw_jC_j:cycleVolumeBounds}.
  Then the equation~\eqref{eq:sumw_jC_j:objectiveToAreasEquation} holds for $\sigma(\vey)$ given by Algorithm~\ref{alg:sigmaFromYWithOrder} and gives the objective function value of $\sigma(\vey)$ under $\sum w_j C_j$.
\end{lemma}
\begin{proof}
Fix a machine of kind $i$.
We prove that, for all $k \in [|T| - 1]$, the contribution of the jobs scheduled to be completed in the interval $\interval[open left]{t_k}{t_{k+1}}$ in $\sigma(\vey)$ is equal to $A^i_k(\vey) + B^i_k(\vey) + X^i_k(\vey)$.
Observe that $\mathcal{J}(\vey)$ can be partitioned into $\bigcup_{k=1}^{|T|-1} \mathcal{J}_k(\vey)$ and thus it suffices to show that
\[
  \sum_{J \in \mathcal{J}_k(\vey)} w(J)C^{\sigma(\vey)}_J = A^i_k(\vey) + B^i_k(\vey) + X^i_k(\vey) \,.
\]
Observe that if for a job $J \in \mathcal{J}_k(\vey)$ we split its contribution $w(J) \cdot C_J$ at the critical time $t_k$, we have
\[
  w(J) \cdot C^{\sigma(\vey)}_J
  =
  w(J) \cdot \left( t_k + \left( C^{\sigma(\vey)}_J - t_k \right) \right)
  =
  w(J) \cdot t_k + w(J) \cdot \left( C^{\sigma(\vey)}_J - t_k \right)
  \,.
\]
The left summand is the contribution of the job $J$ in the area $A^i_k(\vey)$, and it remains to show that the right summand is the contribution of $J$ to areas $B^i_k(\vey)$ and $X^i_k(\vey)$.
Observe that, for every $k \in [|T^i| - 1]$, we have \( X^i_k(\vey) = w(\hat{J}) \cdot q^i_k \) if there is a job $\hat{J}$ scheduled to an external cycle in $\mathcal{C}^{\operatorname{ext}}_{*,k}$ in the schedule $\sigma(\vey)$ and otherwise \( X^i_k(\vey) = 0 \).
We compute \( X^i_k(\vey) = w(\hat{J}) \cdot q^i_k = w(\hat{J}) \cdot \left( C^{\sigma(\vey)}_{\hat{J}} - t_k \right) \).
Since we have $\mathcal{J}_k(\vey) = \left\{ \hat{J} \right\} \cup \left\{ J \in C^{\operatorname{int}}_k \right\}$, we are done using the definition of the area $B^i_k(\vey)$.
\end{proof}

Before we are able to express the areas $A$ and $X$ in terms of variables of the model, we introduce auxiliary binary variables: $y^i_{j,C,R,p} \in \{0,1\}$ for each $C \in \mathcal{C}^{\operatorname{ext}}, j \in [d], p \in [p_{\max}]$.
We do this in such a way that they express the product of two binary variables $y^i_{j,C}$ and $y^i_{C,R,p}$ already present in our model (we will also use these later when expressing the area $B$).
This can be done straightforwardly using standard ILP tricks for expressing Boolean connectives (essentially we enforce $y^i_{j,C,R,p} = y^i_{j,C} \land y^i_{C,R,p} = y^i_{j,C} \cdot y^i_{C,R,p}$, where all variables are binary).
\begin{equation}\label{eq:sumw_jC_j:binaryProduct}
\begin{rcases*}
y^i_{j,C,R,p} \ge y^i_{j,C} + y^i_{C,R,p} - 1 & \\
y^i_{j,C,R,p} \le y^i_{j,C} & \\
y^i_{j,C,R,p} \le y^i_{C,R,p} &
\end{rcases*} \forall C \in \CC^{\text{ext}},p \in [p_{\max}], j \in [d]
\end{equation}
It is straightforward to verify that the new variables have the intended meaning.

\begin{lemma}\label{lem:sumw_jC_j:computingAandX}
  Fix a machine of kind $i$ and let $(\vex, \vey, \vez)$ satisfy constraints~\eqref{eq:Cmax:sumOfExes}--\eqref{eq:Cmax:incompatibleCyclesDisable} and \eqref{eq:sumw_jC_j:first}--\eqref{eq:sumw_jC_j:binaryProduct}.
  The functions $A^i_k$ and $X^i_k$ are linear in $\vey$, for all $k \in [|T| - 1]$.
\end{lemma}
\begin{proof}
We divide the proof into two parts, according to the two area types.
\paragraph{Expressing Area $A$.}
Let us express the area of $A^i_k$, which is a rectangle.
As already pointed out, its length is $L = t_k$, which is a constant for fixed $k \in [|T| - 1]$.
Its height~$H$ is the sum of weights of all jobs whose completion time in the schedule $\sigma(\vey)$ is strictly larger than $t_k$ and
at most $t_{k+1}$, that is, the sum of weights of jobs assigned to the cycle $C^{\operatorname{int}}_k$ and the cycles in $\mathcal{C}^{\operatorname{ext}}_{*,k}$.
We naturally decompose these jobs into two groups---the internal cycle and the external cycle (note that only one external cycle from $\mathcal{C}^{\operatorname{ext}}_{*,k}$ can contain a job, since these are mutually incompatible).
We can express $H = \sum_{C \in \mathcal{C}^{\text{ext}}_{*,k}}\sum_{j = 1}^d w_j y_{j,C} + \sum_{j = 1}^d w_j y_{j,C^{\operatorname{int}}_k}$.
It follows that the area $A^i_k$ is $H\cdot L$, which is a linear function in the variables~$\vey$, since $L$ is a constant.
Altogether, we have
\begin{equation}
  A^i_k(\vey)
  =
  t_k \cdot \left( \sum_{C \in \mathcal{C}^{\text{ext}}_{*,k}}\sum_{j = 1}^d w_j y_{j,C} + \sum_{j = 1}^d w_j y_{j,C^{\operatorname{int}}_k} \right) \,.  \tag{$A^i_k$}
\end{equation}

\paragraph{Expressing Area $X$.}
The height of $X^i_k$ is $w_j$ if $y_{j,C} = 1$ for some $C \in \mathcal{C}^{\text{ext}}_{*,k}$, since this the weight of the job scheduled to a cycle in $\mathcal{C}^{\text{ext}}_{*,k}$ if there is such a nonempty cycle in $\sigma(\vey)$.
The width of this area is $p$ if $y_{C,R,p} = 1$.
This implies that $X$ can be expressed with a quadratic term in terms of variables $y_{j,C}$ and $y_{C,R,p}$.
We can linearize this term using the auxiliary variables $y_{j,C,R,p}$:
\begin{equation}
  X^i_k(\vey) = \sum_{j = 1}^d w_j \cdot \left( \sum_{C \in \mathcal{C}^{\text{ext}}_{*,k} } \sum_{p \in [p_{\max}]} p \cdot y_{j,C,R,p} \right) \tag{$X^i_k$} \enspace . \qedhere
\end{equation}
\end{proof}

Turning our attention to area $B$, we compute it using the approach of Knop and Koutecký~\cite[Section~2.4]{KnopKoutecky2017} as outlined above, which uses auxiliary aggregation variables $\alpha^i_{k,j}$ for $j \in [0,d]$ and $k \in [|T^i| - 1]$ whose meaning we will discuss later.
It is worth noting that they used the possibility to express the area of the 2D Gantt chart (see also~\cite{GoemansWilliamson2000}) using so called job slopes which are in close relation to the Smith ordering $\preceq^i_{\mathcal{R},k}$.
Recall that Smith's ordering requires that jobs in an optimal schedule are ordered by the Smith ratio $w_j / p^i_j$.
We call this ratio the \emph{slope} of the job type $j \in [d]$ because in a 2D Gantt chart it corresponds to the slope of the diagonal of the rectangle representing a job of this type.
We first extend the Smith order to a linear order $\prec^i_{\mathcal{R},k}$ which for a fixed $k \in [|T^i|-1]$ and a machine kind $i$ is any fixed linear extension of the Smith's ordering~$\preceq^i_{\mathcal{R},k}$.
We define the \emph{predecessor} operator $\pred^i_{\mathcal{R},k}$ for a job type $j \in [d]$ to be its predecessor in the linear order $\prec^i_{\mathcal{R},k}$ or $0$ if no such predecessor exists.
The \emph{successor} operator $\successor^i_{\mathcal{R},k}$ is the opposite to $\pred^i_{\mathcal{R},k}$, that is, $j = \successor^i_{\mathcal{R},k}(\pred^i_{\mathcal{R},k}(j))$ for all $j \in [d]$.
For the job type $j$ which comes last in $\preceq^i_{\RR,k}$ we define $\rho^i_{k,\successor^i_{\mathcal{R},k}(j)} = 0$ for technical reasons.
The variable $\alpha^i_{k,j}$ is intended to represent the total size of jobs scheduled to be processed in the interval $\interval{t_k}{t_{k+1}}$ (i.e., those assigned to~$C^{\operatorname{int}}_k$) preceding and including type $j$ in the order $\prec^i_{\mathcal{R},k}$ together with the amount of processing time needed to complete the job assigned to an external cycle in $\mathcal{C}^{\operatorname{ext}}_{*,k}$ if such a job exists).

Let $\hat{J}$ be the job assigned to an external cycle in $\mathcal{C}^{\operatorname{ext}}_{*,k}$, if it exists.
The variable $\alpha^i_{k,0}$ is the amount of time $\hat{J}$ is processed after the critical time $t_k$, or $0$ if $\hat{J}$ is not defined, that is:
\begin{equation}\label{eq:sumw_jC_j:alphaZero}
  \alpha^i_{k,0} = \sum_{C \in \mathcal{C}^{\text{ext}}_{*,k}} \sum_{p\in[p_{\max}]} p \cdot y^i_{C,R,p} \,.
\end{equation}
The remaining variables are enforced by constraints
\begin{equation}\label{eq:sumw_jC_j:last}
  \alpha^i_{k,j} = p^i_j y^i_{j,C} + \alpha^i_{k,\pred^i_{\mathcal{R},k}(j)}  \qquad\qquad \forall j \in [d]  \,.
\end{equation}
Note that if an external cycle started before $t_{k}$ and lasted past $t_{k+1}$ in $\sigma(\vey)$ then all of the above variables equal zero.
If this is the case, all three areas $A^i_k, B^i_k$, and $X^i_k$ are empty.

\begin{lemma}\label{lem:sumw_jC_j:objectiveVariablesFollow}
  Let $(\vex,\vey,\vez)$ satisfy \eqref{eq:Cmax:sumOfExes}--\eqref{eq:Cmax:incompatibleCyclesDisable} and \eqref{eq:sumw_jC_j:first}--\eqref{eq:sumw_jC_j:cycleVolumeBounds}.
  Then, there exists $\vey,\vealpha$ satisfying \eqref{eq:sumw_jC_j:binaryProduct}--\eqref{eq:sumw_jC_j:last}.
\end{lemma}
\begin{proof}
  To see this it suffices to follow the intended meaning of the variables $y^i_{j,C,R,p}$ for $C \in \mathcal{C}^{\operatorname{ext}}$ and compute variables $\alpha^i_{k,q}$ as described in \eqref{eq:sumw_jC_j:alphaZero} and \eqref{eq:sumw_jC_j:last}.
  To that end, we set $y^i_{j,C,R,p} = y^i_{j,C} \cdot y^i_{C,R,p}$.
\end{proof}

Let us now move our attention to the notion of slopes.
We denote $\rho^i_{k,j}$ the ratio of the job type $j \in [d]$ on machines of kind $i$, that is, $\rho^i_{k,j} = w_j / p^i_j$ for all $k \in [|T^i|-1]$.
Note that we have $\rho^i_{k,j} \le \rho^i_{k,j'}$ if and only if $j \preceq^i_{\mathcal{R},k} j'$ and if $\rho^i_{k,j} \le \rho^i_{k,j'}$ and $j' \prec^i_{\mathcal{R},k} j$, then $\rho^i_{k,j} = \rho^i_{k,j'}$.
Let us denote $\rho^i_{k,\max} = \max_{j \in [d]} {w_j / p^i_j}$.

We are now ready to give a different description of the area $B^i_k$.
Introduce a new function
\begin{equation}\label{eq:sumw_jC_j:BHatDefinition}\tag{$\hat{B}^i_k$}
  2 \cdot \hat{B}^i_k(\vey,\ve\alpha)
  =
  \sum_{j = 1}^d \left[ \left(\rho^i_{k,j} - \rho^i_{k,\successor^i_{\mathcal{R},k}(j)}\right) \cdot (\alpha^i_{k,j})^2 \right]
  -
  \rho^i_{k,\max} \cdot \left( \alpha^i_{k,0} \right)^2
  +
  \sum_{j=1}^d w_j \cdot p^i_j \cdot y^i_{j,C^{\operatorname{int}}_k}
  \,.
\end{equation}
Before we prove that the function represents $B^i_k$ we show that it is separable convex in $(\vey, \vealpha)$.
\begin{lemma}\label{lem:sumw_jC_j:computingHatB}
  The function $\hat{B}^i_k$ is separable convex in $\vealpha$ and $\vey$ for all \mbox{$k \in [|T^i| - 1]$}.
\end{lemma}
\begin{proof}
Note that in the definition of the function $\hat{B}^i_k(\vealpha)$ by equation~\eqref{eq:sumw_jC_j:BHatDefinition} the middle term is quadratic with a negative coefficient, making it non-convex in terms of $\vealpha$.
The remaining two terms are separable convex: the first sum is a sum of non-negative multiples of convex quadratic functions, since \(\left(\rho^i_{k,j} - \rho^i_{k,\successor^i_{\mathcal{R},k}(j)}\right) \ge 0\) for all $k \in [|T^i| - 1]$ and $j \in [d]$, and the last sum is a linear term.
We rewrite the middle term \(-\rho^i_{k,\max} \cdot \left( \alpha^i_{0,k} \right)^2\) using~\eqref{eq:sumw_jC_j:alphaZero}, arriving at an equivalent term
\[
  - \rho^i_{k,\max} \cdot
  \left(
    \sum_{C \in \mathcal{C}^{\text{ext}}_{*,k}} \sum_{p\in[p_{\max}]} p \cdot y^i_{C,R,p}
  \right)^2\,,
\]
where the variables $y^i_{C,R,p}$ are binary.
Consequently, using the fact that exactly one term of the form \(p \cdot y^i_{C,R,p}\) is non-zero we have that
\[
  - \rho^i_{k,\max} \cdot
  \left(
    \sum_{C \in \mathcal{C}^{\text{ext}}_{*,k}} \sum_{p\in[p_{\max}]} p \cdot y^i_{C,R,p}
  \right)^2
  =
  - \rho^i_{k,\max} \cdot
  \left(
    \sum_{C \in \mathcal{C}^{\text{ext}}_{*,k}} \sum_{p\in[p_{\max}]} p^2 \cdot y^i_{C,R,p}
  \right)
\]
where the right-hand side is a linear term, since $p$ is a constant in the above expression.
Thus, $\hat{B}^i_k$ is separable convex in $\vealpha$ and $\vey$.
\end{proof}

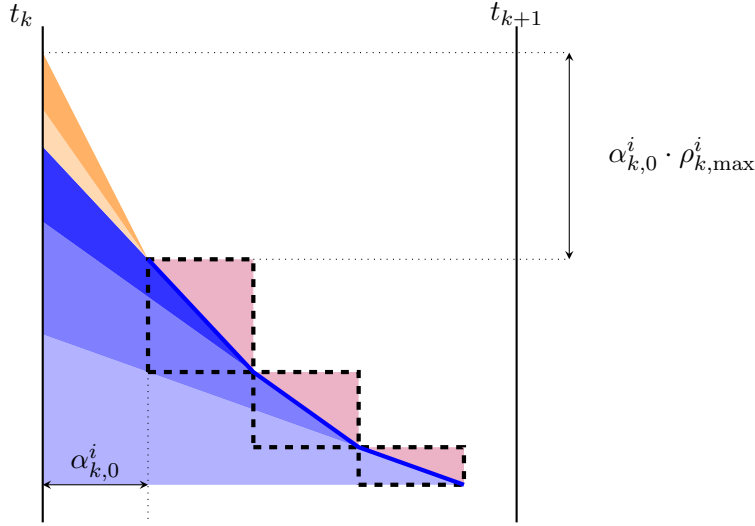
\begin{figure}[bt]
  \begin{minipage}{.65\textwidth}
     \usetikzlibrary{patterns}

\begin{tikzpicture}[yscale=.5,xscale=.7]
\tikzstyle{job}=[ultra thick,dashed]
\tikzstyle{fakeJob}=[dotted,red,ultra thick]

\fill[blue!30] (0,5) -- (8,1) -- (0,1) -- cycle;
\fill[blue!50] (0,8) -- (6,2) -- (0,5) -- cycle;
\fill[blue!80] (0,10) -- (4,4) -- (0,8) -- cycle;

\fill[orange!30] (0,11) -- (2,7) -- (0,10) -- cycle;
\fill[orange!60] (0,12.5) -- (2,7) -- (0,11) -- cycle;

\fill[purple!30] (4,4) -- (4,7) -- (2,7) -- cycle;
\fill[purple!30] (4,4) -- (6,4) -- (6,2) -- cycle;
\fill[purple!30] (6,2) -- (8,2) -- (8,1) -- cycle;

\draw[thick] (0,13.2) -- (0,0);
\draw[thick] (9,13.2) -- (9,0);
\node at (-.4,13.5) {$t_k$};
\node at (9,13.5) {$t_{k+1}$};

\draw[job] (2,7) -- (4,7) -- (4,4) -- (2,4) -- cycle;
\draw[job] (4,4) -- (6,4) -- (6,2) -- (4,2) -- cycle;
\draw[job] (6,2) -- (8,2) -- (8,1) -- (6,1) -- cycle;

\draw[blue, ultra thick] ((2,7) -- (4,4) -- (6,2) -- (8,1);

\draw[dotted] (2,4) -- (2,0);

\draw[dotted] (4,7) -- (10,7);
\draw[dotted] (0,12.5) -- (10,12.5);

\draw[<->,>=stealth] (10,7) to node[midway,xshift=1.5cm] {$\alpha^i_{k,0} \cdot \rho^i_{k,\max}$} (10,12.5);
\draw[<->,>=stealth] (0,1) to node[midway,yshift=.3cm] {$\alpha^i_{k,0}$} (2,1);


\end{tikzpicture}
  \end{minipage}
  \begin{minipage}{.35\textwidth}
      \caption{\label{fig:sumw_jC_jSegmentDetailAreaB}%
         Visualization of the right-hand side of the expression \eqref{eq:sumw_jC_j:BHatDefinition}.
      }
  \end{minipage}
\end{figure}

\begin{lemma}\label{lem:sumw_jC_j:equivalenceOfBs}
  Fix a machine of kind $i \in [\kappa]$ and let $(\vex,\vey,\vez,\vealpha)$ satisfy constraints \eqref{eq:Cmax:sumOfExes}--\eqref{eq:Cmax:incompatibleCyclesDisable} and \eqref{eq:sumw_jC_j:first}--\eqref{eq:sumw_jC_j:last}.
  Then $\hat{B}^i_k(\vealpha)$ expresses the size of area $B^i_k$ in the schedule $\sigma(\vey)$, for all $k \in [|T^i| - 1]$.
\end{lemma}
\begin{proof}
Fix a $k \in [|T^i| - 1]$ and let $C_k = C^{\operatorname{int}}_k$.
We begin by decomposing the area $B^i_k$.
Recall we have $B^i_k = \sum_{J \in C_k} w(J) \cdot \left( C^{\sigma(\vey)}_{J} - t_k \right)$ and in what follows refer to Figure~\ref{fig:sumw_jC_jSegmentDetailSumMember}.
\begin{figure}[bt]
  \begin{minipage}{.7\textwidth}
     \usetikzlibrary{patterns}

\begin{tikzpicture}[yscale=.5,xscale=.7]
\tikzstyle{job}=[ultra thick,dashed]

\draw[pattern=vertical lines, pattern color=blue] (0,8) -- (4,4) -- (0,4) -- cycle;
\draw[pattern=horizontal lines, pattern color=blue] (0,8) -- (6,2) -- (0,2) -- cycle;



\draw[thick] (0,9.2) -- (0,0);
\draw[thick] (9,9.2) -- (9,0);
\node at (-.4,9.5) {$t_k$};
\node at (9,9.5) {$t_{k+1}$};

\draw[job] (2,7) -- (4,7) -- (4,4) -- (2,4) -- cycle;
\draw[job] (4,4) -- (6,4) -- (6,2) -- (4,2) -- cycle;
\draw[job] (6,2) -- (8,2) -- (8,1) -- (6,1) -- cycle;

\draw[blue, ultra thick] (2,7) -- (4,4) -- (6,2) -- (8,1);

\begin{scope}[on background layer]
  \fill[gray!30] (0,2) -- (6,2) -- (4,4) -- (0,4) -- cycle;
\end{scope}

\draw[dotted] (8,2) -- (10,2);
\draw[dotted] (6,4) -- (10,4);
\draw[dotted] (-1,8) -- (10,8);
\draw[dotted] (-1,2) -- (0,2);

\draw[dotted] (4,2) -- (4,0);
\draw[dotted] (6,1) -- (6,0);
\draw[dotted] (6,4) -- (6,7);

\draw[<->,>=stealth] (-.5,2) to node[midway,xshift=-1cm] {$\rho^i_{k,j} \cdot \alpha^i_{k,j}$} (-.5,8);

\draw[<->,>=stealth] (9.5,2) to node[midway,xshift=1cm] {$w_j \cdot y^i_{j,C^{\operatorname{int}}_k}$} (9.5,4);
\draw[<->,>=stealth] (9.5,4) to node[midway,xshift=1.3cm] {$\rho^i_{k,j} \cdot \alpha^i_{k,\pred(j)}$} (9.5,8);

\draw[<->,>=stealth] (0,1.7) to node[midway,yshift=-.4cm] {$\alpha^i_{k,\pred^i_k(j)}$} (4,1.7);
\draw[<->,>=stealth] (0,0) to node[midway,yshift=-.4cm] {$\alpha^i_{k,j}$} (6,0);

\draw[<->,>=stealth] (4,5) to node[midway,yshift=.4cm,xshift=.3cm] {$p^i_j \cdot y^i_{j,C^{\operatorname{int}}_k}$} (6,5);
\end{tikzpicture}
  \end{minipage}
  \begin{minipage}{.29\textwidth}
      \caption{\label{fig:sumw_jC_jSegmentDetailSumMember}%
         An illustration of the areas computed in the proof of Lemma~\ref{lem:sumw_jC_j:equivalenceOfBs}.
         The task is to express the gray area as the difference of the horizontally striped area and the vertically striped area (note that these two are overlapping).
      }
  \end{minipage}
\end{figure}
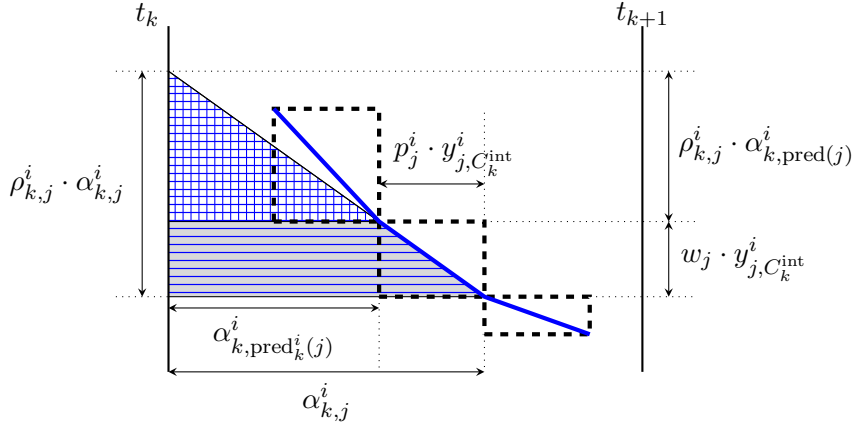
Denote by $\JJ_j$ the set of all jobs of type $j$.
Let us show that for each $j \in [d]$, we have
\[
  \sum_{J \in \mathcal{J}_j \cap C_k} w(J) \cdot \left( C^{\sigma(\vey)}_{J} - t_k \right)
  =
  w_j \cdot y^i_{j,C_k} \cdot \alpha^i_{k, \pred^i_{\mathcal{R},k}(j)} + \frac{w_j}{2} \cdot p^i_j \cdot \left( y^i_{j,C_k}\right)^2 + \frac{w_j}{2} \cdot p^i_j \cdot y^i_{j,C_k}
   \enspace .
\]
This is because
\begin{itemize}
  \item
  the first summand is the total weight of jobs of type $j$ assigned to the cycle $C_k$ (i.e., $w_j \cdot y^i_{j,C_k}$) times the total processing time required by the jobs assigned to the cycle $C_k$ preceding type $j$ in the Smith ordering $\prec^i_{\mathcal{R},k}$ and
  \item
  the remaining contribution of jobs of type $j$ is $w_j \cdot p^i_j \cdot \binom{y^i_{j,C_k}}{2}$, which we have split into two summands for convenience.
\end{itemize}
Now, the last summand is already the same as the last summand in the expression $\hat{B}^i_k$.
Thus, it remains to argue about the first two.
Recall that \( \alpha^i_{k,j} = \alpha^i_{k, \pred^i_{\mathcal{R},k}(j)} + p^i_j \cdot y^i_{j,C_k} \) and we compute
\begin{align*}
  &w_j \cdot y^i_{j,C_k} \cdot \alpha^i_{k, \pred^i_{\mathcal{R},k}(j)} + \frac{w_j}{2} \cdot p^i_j \cdot \left( y^i_{j,C_k}\right)^2 \\
  &=
  w_j \cdot y^i_{j,C_k} \cdot \alpha^i_{k, \pred^i_{\mathcal{R},k}(j)} + \frac{w_j}{2} \cdot p^i_j \cdot \left( y^i_{j,C_k}\right)^2 + \frac{w_j}{2p^i_j}\cdot\left( \alpha^i_{k, \pred^i_{\mathcal{R},k}(j)} \right)^2 - \frac{w_j}{2p^i_j}\cdot\left( \alpha^i_{k, \pred^i_{\mathcal{R},k}(j)} \right)^2 \\
  &=
  \frac{w_j}{2p^i_j} \left( 2p^i_j \cdot y^i_{j,C_k} \cdot \alpha^i_{k, \pred^i_{\mathcal{R},k}(j)} + \left( p^i_j \cdot y^i_{j,C_k} \right)^2 + \left( \alpha^i_{k,\pred(j)} \right)^2 \right) - \frac{w_j}{2p^i_j}\cdot\left( \alpha^i_{k, \pred^i_{\mathcal{R},k}(j)} \right)^2 \\
  &=
  \frac{w_j}{2p^i_j} \left( \alpha^i_{k, \pred^i_{\mathcal{R},k}(j)} + p^i_j \cdot y^i_{j,C_k} \right)^2 - \frac{w_j}{2p^i_j}\cdot\left( \alpha^i_{k, \pred^i_{\mathcal{R},k}(j)} \right)^2 \\
  &=
  \frac{w_j}{2p^i_j} \left( \alpha^i_{k,j} \right)^2 - \frac{w_j}{2p^i_j}\cdot\left( \alpha^i_{k, \pred^i_{\mathcal{R},k}(j)} \right)^2
  = \frac{1}{2} \cdot \rho^i_{k,j} \left[ \left(\alpha^i_{k,j}\right)^2 - \left(\alpha^i_{k, \pred^i_{\mathcal{R},k}(j)}\right)^2 \right] \,,
\end{align*}
where the last equality follows from $\rho^i_{k,j} = w_j / p^i_j$.
It remains to show that indeed we have
\[
  \sum_{j = 1}^d \rho^i_{k,j} \left[ \left(\alpha^i_{k,j}\right)^2 - \left(\alpha^i_{k, \pred^i_{\mathcal{R},k}(j)}\right)^2 \right]
  =
  \sum_{j = 1}^d \left[ \left(\rho^i_{k,j} - \rho^i_{k,\successor^i_{\mathcal{R},k}(j)}\right) \cdot (\alpha^i_{k,j})^2 \right]
  -
  \rho^i_{k,\max} \cdot \left( \alpha^i_{0,k} \right)^2 \,.
\]
To see this we expand the right hand side of the above equation, and denote by $D'$ the set $[d]$ minus the last job type $\hat{j}$ in the linear ordering $\prec^i_{\mathcal{R},k}$:
\begin{align*}
  &\sum_{j = 1}^d \left[ \left(\rho^i_{k,j} - \rho^i_{k,\successor^i_{\mathcal{R},k}(j)}\right) \cdot (\alpha^i_{k,j})^2 \right] - \rho^i_{k,\max} \cdot \left( \alpha^i_{0,k} \right)^2 \\
  &=
  \sum_{j = 1}^d \left( \rho^i_{k,j} (\alpha^i_{k,j})^2 \right) - \sum_{j = 1}^d \left( \rho^i_{k,\successor^i_{\mathcal{R},k}(j)} (\alpha^i_{k,j})^2 \right) - \rho^i_{k,\max} \cdot \left( \alpha^i_{0,k} \right)^2 \\
  &=
  \sum_{j = 1}^d \left( \rho^i_{k,j} (\alpha^i_{k,j})^2 \right) - \sum_{j \in D'} \left( \rho^i_{k,\successor^i_{\mathcal{R},k}(j)} (\alpha^i_{k,j})^2 \right) - \rho^i_{k,\max} \cdot \left( \alpha^i_{0,k} \right)^2 \\
  &=
  \sum_{j = 1}^d \left( \rho^i_{k,j} (\alpha^i_{k,j})^2 \right) - \sum_{j = 1}^d \left( \rho^i_{k,j} (\alpha^i_{k,\pred^i_{\mathcal{R},k}(j)})^2 \right) \,.
\end{align*}
This concludes the proof.
\end{proof}

We conclude that the resulting objective is a positive sum of separable convex functions plus a linear term, thus separable convex.
The correctness of the model defined by constraints~\eqref{eq:Cmax:sumOfExes}--\eqref{eq:Cmax:incompatibleCyclesDisable} and \eqref{eq:sumw_jC_j:first}--\eqref{eq:sumw_jC_j:last} and the objective
\[
  \min f^i(\vey, \vealpha) = \sum_{k = 1}^{|T^i| -1} \left( A^i_k(\vey) + \hat{B}^i_k(\vealpha,\vey) + X^i_k(\vey) \right)
\]
follows from Lemma~\ref{lem:modelPlacesCyclesCorrectly} and Lemmata \ref{lem:cyclesScheduleOrderedObjective}--\ref{lem:sumw_jC_j:equivalenceOfBs}.
Now we conclude with the properties of the model.
\begin{lemma}
\label{lem:MIMOPSumObjectives}
  Let $\mathcal{I}$ be an instance of $R | r^i_j,d^i_j | \sum w_jC_j$ or $R | r^i_j,d^i_j | \sum w_jF_j$ with $m$ machines of $\kappa$ kinds and $d$ job types with maximum job size $p_{\max}$.
  There is a MIMO model $\mathcal{S}$ for $\mathcal{I}$ with extension-separable convex objective functions and parameters
  \begin{tasks}[style=itemize](3)
    \task $\mathcal{S}(\Delta) = p_{\max}$,
    \task $\mathcal{S}(M) = \OhOp{(d)^2 \cdot p_{\max}}$,
    \task $\mathcal{S}(d) = d$,
    \task $\mathcal{S}(d^i) = \OhOp{(d)^2 \cdot p_{\max}}$,
    \task $\mathcal{S}(N) = m$, and
    \task $\mathcal{S}(\tau) = \kappa$.
  \end{tasks}
\end{lemma}
\begin{proof}
  The largest coefficient in the system is $p_{\max}$, as before $\mathcal{S}(N) = \|\vemu\|_1 = m$, and the projections discarding $\vey,\vez,\vealpha$ leave variables $\vex$ of dimension~$d$.
  There are $\kappa$ types in the presented MIMO model, one for each kind of a machine.
  As for the parameter $\mathcal{S}(M)$.
  We have
  \begin{itemize}
    \item $O((d)^2)$ constraints identical with the $R|r^i_j, d^i_j|C_{\max}$ model~\eqref{eq:Cmax:sumOfExes}--\eqref{eq:Cmax:incompatibleCyclesDisable} and
    \item $O((d)^2 p_{\max})$ new constraints~\eqref{eq:sumw_jC_j:first}--\eqref{eq:sumw_jC_j:last}, whose number is dominated by the number of constraints~\eqref{eq:sumw_jC_j:last}.
  \end{itemize}
  In total the number of conditions is governed by the second term $\OhOp{d^2 \cdot p_{\max}}$.
  Regarding the parameter $d^i$, i.e., the dimension of the variables projected out, we have the following.
  The dimension of $\vey$ is $\OhOp{(d)^2 \cdot p_{\max}}$.
  The dimension of $\vez,\vealpha$ is $\OhOp{(d)^2}$.
  These expressions belong to $\OhOp{(d)^2 \cdot p_{\max}}$.
\end{proof}
Before we discuss the algorithmic consequences of the presented MIMO model we show how to adapt the MIMO objective to express the $\sum w_jT_j$ scheduling objective.

\subsubsection{Sum of Weighted Tardiness} \label{sec:tardiness}
Recall that for a job $J$ of type $j$ with completion time $C_J$ its tardiness is $\max\{0, C_J - d^i(J)\}$ if $J$ is scheduled to run on a machine of kind~$i$.
Note that the auxiliary scheduling problem in each interval $\interval{t_k}{t_{k+1}}$ is solved by scheduling jobs in the following order: first come tardy jobs according to their Smith ratio $w_j/p^i(J)$, then other jobs follow in an arbitrary order (since they do not incur any penalty).
(Recall that for the $\sum w_j C_j$ the ordering for this auxiliary problem has been identical for all intervals.)
Again, we fix an arbitrary linear extension $\prec^i_{\mathcal{R},k}$ of this ordering.
See Figure~\ref{fig:sumw_jT_jSegment}.
Note the difference between area $A^i_k$ now and in $\sum w_jC_j$.
In the later objective, this area spans from the critical time $t_k$ to $0$, however, when we measure tardiness each job only contribute its weight times the length of the time interval between $t_k$ and its due date.
This, as we shall see, is not hard to incorporate in 2D Gannt charts and the objective function of our MIMO model.
\begin{figure}[bt]
  \begin{minipage}{.78\textwidth}
    \usetikzlibrary{patterns,decorations.pathreplacing}

\begin{tikzpicture}[yscale=.6,xscale=.85]
\tikzstyle{job}=[ultra thick,dashed]

\fill[green!30] (0,9) -- (-2,9) -- (-2,7) -- (0,7) -- cycle;
\fill[green!30] (0,4) -- (-4,4) -- (-4,2) -- (-2,2) -- (-2,1) -- (0,1) -- cycle;
\node[fill=white!10,text=green,inner sep=2pt] at (-1.5,8) {$A^i_k$};
\node[fill=white!10,text=green,inner sep=2pt] at (-1.5,2.5) {$A^i_k$};

\fill[pattern=north west lines,pattern color=blue!40] (0,7) -- (2,7) -- (4,7) -- (4,4) -- (6,4) -- (6,2) -- (8,2) -- (8,1) -- (0,1) -- cycle;
\node[fill=white!10,text=blue,inner sep=2pt] at (2,2) {$B^i_k$};

\fill[olive!40] (0,9) -- (0,7) -- (2,7) -- (2,9) -- cycle;
\node[fill=white!10,text=olive,inner sep=2pt] at (1,8) {$X^i_k$};

\draw[thick] (0,9.5) -- (0,-.3);
\draw[thick] (11,9.5) -- (11,-.3);
\node at (0,10) {$t_k$};
\node at (11,10) {$t_{k+1}$};

\draw[job] (-1,7) -- (-1,9) -- (2,9) -- (2,7) -- cycle;
\draw[job] (2,7) -- (4,7) -- (4,4) -- (2,4) -- cycle;
\draw[job] (4,4) -- (6,4) -- (6,2) -- (4,2) -- cycle;
\draw[job] (6,2) -- (8,2) -- (8,1) -- (6,1) -- cycle;
\draw[job] (8,1) -- (10,1) -- (10,0) -- (8,0) -- cycle;

\draw[blue, ultra thick] ((2,7) -- (4,4) -- (6,2) -- (8,1);

\draw[decoration={brace}, decorate] (2,7.2)  -- node[midway,yshift=9pt] {tardy jobs} (8,7.2);
\draw[decoration={brace}, decorate] (8,7.2)  -- node[midway,yshift=9pt] {non-tardy jobs} (10,7.2);
\draw[dotted] (8,7.2) -- (8,2);

\end{tikzpicture}
  \end{minipage}
  \begin{minipage}{.2\textwidth}
    \caption{\label{fig:sumw_jT_jSegment}%
    A 2D Gantt chart visualization of the decomposition of the objective $\sum w_jT_j$ between two consecutive critical times $t_k$ and $t_{k+1}$.
    }
  \end{minipage}
\end{figure}

In order to simplify some argumentation we add an auxiliary critical time $t = \sum_{J \in \mathcal{J}} p^i(J)$ to the set $T^i$.
This is to unify our reasoning about jobs that are scheduled after the last due date (which is possible when it comes to tardiness).
Furthermore, when it comes to tardiness, it is possible to schedule a job to run after its due date and therefore we adjust the definition of $\chi$ to capture this
\[
  \chi^i_{j,C} = 1
  \qquad\qquad\text{if and only if}\qquad\qquad
  r^i_j \le \operatorname{left}(C) \,.
\]
Of course, for external cycles we keep the limitations arising from job sizes.

We make no changes to the constraints  \eqref{eq:Cmax:sumOfExes}--\eqref{eq:Cmax:incompatibleCyclesDisable} and \eqref{eq:sumw_jC_j:first}--\eqref{eq:sumw_jC_j:last} and only alter the objective.
First, we give the description of three types of areas, analogously to $\sum w_j C_j$.
Intuitively the area $A^i_k(\vey)$ is the contribution of the jobs in $\JJ_k$ (the set of jobs with completion time in $\interval[open left]{t_k}{t_{k+1}}$) as if all of them were completed at time $t_k$, the area $B^i_k(\vey)$ is the remaining contribution of jobs in $\JJ_k \cap C^{\operatorname{int}}_k$, and finally $X^i_k(\vey)$ is the rest of the contribution of the job $\hat{J}$ scheduled to an external cycle in $\mathcal{C}^{\operatorname{ext}}_{*,k}$.
Here, we define $w^i_{k,j} = 0$ if $t_k \le d^i_j$ and $w^i_{k,j} = w_j$, otherwise; we extend this notion to individual jobs $J$ in a natural way, i.e., we set $w^i_k(J) = w^i_{k,j}$ if $J$ is of type $j$.
Recall that $q^i_k$ is the time that job $\hat{J}$ runs past $t_k$.
The areas are as follows
\begin{itemize}
  \item The area $A^i_k(\vey)$ is $\sum_{J \in \mathcal{J}_k(\vey)} \left( \left( t_k - d^i(J) \right) \cdot w^i_k(J) \right)$,
  \item The area $B^i_k(\vey)$ is $\sum_{J \in \mathcal{J}_k(\vey) \setminus \{ \hat{J} \} } \left( \left( C^{\sigma(\vey)}_J - t_k \right) \cdot w^i_{k,j} \right)$, and
  \item The area $X^i_k(\vey)$ is $q^i_k \cdot w^i_k(\hat{J})$.
\end{itemize}
Now, we arrive at
\begin{equation}\tag{objective}\label{eq:sumw_jT_j:objectiveToAreasEquation}
  \sum_{J \in \mathcal{J}(\vey)} w(J)T^{\sigma({\vey})}_J = \sum_{k = 1}^{|T^i| - 1} \left(A^i_k(\vey) + B^i_k(\vey) + X^i_k(\vey) \right).
\end{equation}

Similarly to Lemma~\ref{lem:sumw_jC_j:objectiveToAreasEquationHolds} the above decomposition of the objective holds for tardiness.
\begin{lemma}
  Let $(\vex,\vey,\vez)$ satisfy constraints \eqref{eq:Cmax:sumOfExes}--\eqref{eq:Cmax:incompatibleCyclesDisable} and \eqref{eq:sumw_jC_j:first}--\eqref{eq:sumw_jC_j:cycleVolumeBounds}.
  Then the equation~\eqref{eq:sumw_jT_j:objectiveToAreasEquation} holds for $\sigma(\vey)$ given by Algortihm~\ref{alg:sigmaFromYWithOrder} and gives the value of $\sigma(\vey)$ under $\sum w_j T_j$.
\end{lemma}
\begin{proof}
Observe that $\mathcal{J}(\vey)$ can be partitioned into $\bigcup_{k=1}^{|T^i|-1} \mathcal{J}_k(\vey)$ and thus it suffices to show that
\[
  \sum_{J \in \mathcal{J}_k(\vey)} w(J)T^{\sigma({\vey})}_J
  =
  A^i_k(\vey) + B^i_k(\vey) + X^i_k(\vey) \,.
\]
To see this observe that for a job $J \in \mathcal{J}_k(\vey)$ we split its contribution $w(J) \cdot T_J$ at time $t_k$ and we have
\[
  w(J)T^{\sigma({\vey})}_J
  =
  w(J) \cdot \max\{0, C^{\sigma(\vey)}_J - d^i(J)\}
  =
  w^i_k(J) \cdot \left( \max\{0, t_k - d^i(J)\} + \max\{0, C^{\sigma(\vey)}_J - t_k \} \right) \,.
\]
It follows from the definition of $w^i_k(J)$ that $\max\{0, t_k - d^i(J)\} = 0$ if and only if $w^i_{k-1}(J)= 0$.
Similarly, $\max\{0, C^{\sigma(\vey)}_J - t_k\} = 0$ if and only if $w^i_k(J) = 0$.
Hence,
\[
  w(J) \cdot \left( \max\{0, t_k - d^i(J)\} + \max\{0, C^{\sigma(\vey)}_J - t_k \} \right)
  =
  w^i_{k-1}(J) \cdot \left( t_k - d^i(J) \right) + w^i_{k}(J) \cdot \left( C^{\sigma(\vey)}_J - t_k \right) \,.
\]
Aggregating over all $\JJ_k(\vey)$, the first summands sum up to $A_k^i(\vey)$, and the second summands sum up to $B_k^i(\vey) + X_k^i(\vey)$, where $B_k^i(\vey)$ is obtained from jobs in $C^{\operatorname{int}}_k$ and $X_k^i(\vey)$ from $\hat{J}$.
\end{proof}

\begin{lemma}\label{lem:sumw_jT_j:computingAandX}
  Fix a machine of kind $i$ and let $(\vex, \vey, \vez)$ satisfy constraints~\eqref{eq:Cmax:sumOfExes}--\eqref{eq:Cmax:incompatibleCyclesDisable} and \eqref{eq:sumw_jC_j:first}--\eqref{eq:sumw_jC_j:binaryProduct}.
  The functions $A^i_k$ and $X^i_k$ are linear in $\vey$, for all $k \in [|T^i| - 1]$.
\end{lemma}
\begin{proof}
We claim that area $A^i_k$ can be expressed as follows:
\begin{equation}\tag{$A^i_k$} \label{eq:Aik_tardiness}
  A^i_k(\vey) = \sum_{j = 1}^d w^i_{k,j} \cdot \left( \left( t_k - d^i_j \right) \cdot \left( y^i_{j,C^{\text{int}}_k} + \sum_{C \in \mathcal{C}^{\text{ext}}_{*,k}} y^i_{j,C} \right) \right) \,.
\end{equation}
The correctness follows from the definition of the weights $w^i_{k,j}$.
Observe that if a job type $j$ is not tardy in the interval $\interval{t_k}{t_{k+1}}$, then its contribution towards the objective should be $0$ and this is enforced by $w^i_{k,j} = 0$.
Thus, it remains to argue about jobs of type $j$ which is already tardy in the interval $\interval{t_k}{t_{k+1}}$.
Let $J$ be such a job and let $j$ be its job type.
We have that $w^i_{k,j} = w(J)$ and $(t_k - d^i_j) \ge 0$.
Now the contribution of the job $J$ to the part $A^i_k$ is by definition $w(J) \cdot (t_k - d^i_j)$, and there are $y^i_{j,C^{\text{int}}_k}$ of such jobs in $C^{\text{int}}_k$, hence the contribution of jobs of type $j$ in $C^{\text{int}}_k$ towards $A^i_k$ is $w^i_{k,j}\cdot (t_k - d^i_j) \cdot y^i_{j,C^{\text{int}}_k}$,
It remains to account for the contribution of the job $\hat{J}$, which is accounted for by the remaining term $w^i_{k,j} \cdot \left(( t_k - d^i_j) \cdot \left( \sum_{C \in \mathcal{C}^{\text{ext}}_{*,k}} y^i_{j,C} \right) \right)$.

Area $X^i_k$ can be expressed as follows
\begin{equation}\tag{$X^i_k$}
  X^i_k(\vey) = \sum_{C \in \mathcal{C}^{\text{ext}}_{*,k}} \sum_{p \in [p_{\max}]} \sum_{j = 1}^d p \cdot w^i_{k,j} \cdot y^i_{j,C,R,p} \,.
\end{equation}
Again if the job $J$ assigned to a cycle $C$ in $\mathcal{C}^{\text{ext}}_{*,k}$ is not tardy (i.e., $t_{k+1} \le d^i(J)$), then the above sum is $0$ and thus it is correct.
It remains to argue the correctness in the case when $J$ is tardy.
Let $j$ be the job type of $J$ and let $p$ be the time $J$ needs to finish after $t_k$ in $\sigma(\vey)$.
It follows that the only nonzero summand in the above expression is $p \cdot w^i_{k,j} \cdot y^i_{j,C,R,p}$, where we have $w^i_{k,j} = w(J)$ and $y^i_{j,C,R,p} = 1$.
The correctness thus follows.
\end{proof}

The function $B^i_k(\vey, \vealpha)$ is identical to the one for $\sum w_jC_j$ except for the segment-specific weights defined above, hence it is separable convex in $(\vey, \vealpha)$.
The parameters of the resulting MIMO model are the same as in Lemma~\ref{lem:MIMOPSumObjectives}, since it only differs from it in the objective function.

\subsubsection{Algorithm for Objectives $\sum w_j C_j$, $\sum w_j F_j$, and $\sum w_j T_j$}
Let us summarize the properties of the MIMO models we have constructed:
\begin{lemma}
\label{lem:MIMORSumObjectives}
  Let $\mathcal{I}$ be an instance of $R | r^i_j,d^i_j | \mathcal{R}$ with $\mathcal{R} \in \left\{ \sum w_j C_j, \sum w_j F_j, \sum w_j T_j \right\}$ with $m$ machines of $\kappa$ kinds, $d$ job types, and denote $p_{\max} = \max_{1 \le j \le d} p^i_j$.
  There is a MIMO model $\mathcal{S}$ for $\mathcal{I}$ with extension-separable convex objective functions and parameters
  \begin{tasks}[style=itemize](3)
    \task $\mathcal{S}(\Delta) = p_{\max}$,
    \task $\mathcal{S}(M) = \OhOp{(d)^2 \cdot p_{\max}}$,
    \task $\mathcal{S}(d) = d$,
    \task $\mathcal{S}(d^i) = \OhOp{(d)^2 \cdot p_{\max}}$,
    \task $\mathcal{S}(N) = m$, and
    \task $\mathcal{S}(\tau) = \kappa$. \qed
  \end{tasks}
\end{lemma}

The following theorem then follows by the application of parts 1, 3 and 4 of Theorem~\ref{thm:implicitMIMO}.
\begin{theorem}
  Fix a scheduling objective $\mathcal{R} \in \left\{ \sum w_jC_j, \sum w_jT_j, \sum w_jF_j \right\}$.
  The problem $R | r^i_j,d^i_j | \mathcal{R}$ with~$m$ machines of~$\kappa$ kinds and $d$ job types with maximum job size $p_{\max}$ admits a fixed-parameter algorithm for parameters
\begin{itemize}
  \item $m + d + p_{\max}$, and,
  \item $d + p_{\max}$. \qed
\end{itemize}
\end{theorem}

\subsubsection{Introducing Speeds into the Model}
Recall that with the $C_{\max}$ objective we used the right-hand sides to introduce speeds to the model, and this was sufficient because for the objective $C_{\max}$ the ``relative `position of a critical time during processing of a job assigned to a realization of an external cycle'' is irrelevant.
In the case of ordered objectives this changes.
Now, when speeds come to play it is not possible to assume a job in an external cycle is split by the corresponding critical time into integral chunks ($\mathbb{N}$-regularity).
The weaker notion of $\frac{\mathbb{N}}{s}$-regularity is not sufficient, as we want to bound the number of variables in our model independently of the value $s$.
The good news is that there is only a limited number of (fractional parts of the) ``shifts'', i.e., possibilities of how a unit of job size can be split by a critical time in an optimal schedule, as we are about to see.
On the other hand, we cannot afford to guess these shifts beforehand and thus we have to be able to express the shift corresponding to a schedule encoded by $\vey$.
In order to do so we introduce the notion of a \emph{busy period} of a nonempty (external) cycle on a machine.
Fix a machine of kind~$i$ and let $s$ be its speed.
Let $C$ be a nonempty external cycle in some schedule $\sigma$ and its decomposition $\DD$ on this machine, and let $J$ be the job assigned to~$C$.
We say that a \emph{busy period} of $C$ starts at time $t(J)$ (cf. Lemma~\ref{lem:leftAlignedScheduleCMax}); recall that $t(J)$ is the latest critical time preceding the execution of $J$ before which the machine is idle (or $0$ if the machine is not idle prior to the execution of~$J$).
Observe now that if $t_k$ is a critical time contained in $C$, then $t_k$ splits $J$ into two parts such that $\alpha + \beta$ of $J$'s size is processed at time $t_k$, where $\alpha \in \N$ and $0 \le \beta = \fract(s \cdot (t_k - t(J))) < 1$.
This is because $s \cdot (t_k - t(J))$ is the total size of jobs a machine processes after $t(J)$ and since the sizes of all jobs ($p^i_j$) is a positive integer.
Note that the value of $\beta$ depends only on the alignment of $t(J)$ with respect to $t_k$; consequently, if the fractional part of $s \cdot (t_k - t(J))$ is the same as the fractional part of $s \cdot (t_k - t_\ell)$ the busy period for $J$ might as well started at $t_\ell$ and the result would be the same.

Observe that the size of a job $J$ assigned to an external cycle $C = C^{\operatorname{ext}}_{\ell,k}$ is split naturally into three parts as follows.
\begin{itemize}
  \item The first part is the amount of the size of $J$ which is already processed at time $t_\ell$; we denote this $\bar{\alpha} + \bar{\beta}$ and we stress that $\bar{\alpha} \in \mathbb{N}$ and $0 \le \bar{\beta} < 1$.
  \item The second part is possibly empty and is of length $s \cdot (t_k - t_\ell)$, that is, the amount of $J$'s size $p^i(J)$ processed between the critical times defining the cycle $C$.
  \item The last part of the total size of $J$ processed after $t_k$; again denote $\alpha + \beta$ with $\alpha \in \mathbb{N}$ and $0 \le \beta < 1$.
\end{itemize}

When it comes to the computation of the value of the objective we observe that the area $A^i_k$ remains completely the same.
Furthermore, the areas $B^i_k$ and $X^i_k$ remain nearly the same (essentially, we only stretch our arguments by the speed factor).
Notice that for computation of $B^i_k$ and $X^i_k$ we needed for each job type its \emph{exact} completion time after we subtract $t_k$.
This is again possible using the shifts $\beta$ and thanks to fractional MIMO objective.

Note that the number of different possible shifts is of order $d^2$, since it is a function only of the two critical times---the one starting the busy period and the last one in the external cycle preceding the current time interval $\interval{t_\ell}{t_{\ell+1}}$ for some $\ell \in [|T| - 1]$ (i.e., a possibly nonempty internal cycle).
Let $\beta_{\hat{\ell}, \ell}$ denote the fractional part of the job size to be processed after the critical time $t_{\ell}$ (which we call \emph{$t_\ell$ busy shift}) if the busy period started at time $t_{\hat{\ell}}$, that is,
\[
  \beta_{\hat{\ell}, \ell} = 1 - \fract(s \cdot (t_\ell - t_{\hat{\ell}})) \,.
\]
Let $\mathcal{B}^{i,s}$ be the set of all possible shifts with respect to machine kind $i$ and speed $s$.
Moreover, we make $\BB^{i,s}$ ``symmetric'' by adding to it $(1-\beta)$ for every $\beta \in \BB^{i,s}$, and also add $0$ to it.
This allows us to give an analogue of Lemma~\ref{lem:cyclesScheduleOrderedObjective}, which says that it is sufficient to restrict our attention to $(\N+\BB^{i,s})$-regular schedules, where by $\N+\BB^{i,s}$ we denote the set $\left\{ \alpha + \beta \mid \alpha \in \N, \beta \in \BB^{i,s} \right\}$.
\begin{lemma}
	\label{lem:cyclesScheduleOrderedObjectiveSpeeds}
	Fix a machine of kind $i$ with speed $s$ and an ordered objective $\mathcal{R}$, and let $\vex \in \N^d$.
	If there is a schedule $\sigma'$ of $\vex$ with value $\varphi$ under $\mathcal{R}$, then there is an $(\mathbb{N}+\BB^{i,s})$-regular schedule $\sigma$ of $\vex$ with value at most~$\varphi$ under $\mathcal{R}$ such that each cycle of $\sigma$ is a realization of some potential cycle $C \in \CC$, and the jobs in each internal cycle $C = C^{\operatorname{int}}_k$ of $\sigma$ are ordered by $\preceq^i_{\mathcal{R},k}$.
\end{lemma}
\begin{proof}
The procedure is identical to the one described in the proof of Lemma~\ref{lem:cyclesScheduleOrderedObjective}.
It suffices to observe that the left-aligning of cycles results in a schedule which satisfies $\lambda(C) \in \N+\BB^{i,s}$ for all $C \in \DD'$.
Note that this is the point where we are using the symmetry of $\BB^{i,s}$, since if $J$ is a job in an external cycle which is split by $t_k$ and $J$ is processed for $\alpha+\beta$ size units after $t_k$ with $\alpha \in \N$, $0 \leq \beta <1$, it means that $J$ is processed for some $\bar{\alpha}+\bar{\beta}$ size units before $t_k$ where $\bar{\alpha} \in \N$ and $\bar{\beta} = 1-\beta$.
\end{proof}

Before we proceed let us briefly mention the basis of the MIMO model for polynomial objectives and input given in an $\ves$-representation.
As the basis we use the constraints \eqref{eq:Cmax:sumOfExes}--\eqref{eq:Cmax:incompatibleCyclesDisable} and \eqref{eq:sumw_jC_j:second}--\eqref{eq:sumw_jC_j:fourth}.

Now, our aim is to design binary variables $y^{i,s}_{C,\hat{\ell},\ell}$ with $C = C^{\operatorname{ext}}_{k, \ell}$ for some $1 \le k \le \ell$ and every $1\le\hat{\ell}\le k$ taking value $1$ if and only if the external cycle $C$ is nonempty and the busy shift for $C$ is $\beta_{\hat{\ell},\ell}$.
To that end we add the following constraints to our MIMO model.
We now use variables $\vey^{i,s}$ instead of $\vey^i$ in order to clearly keep track of the current machine kind~$i$ and its speed~$s$; similarly for $\vez^{i,s}$.
\begin{align}
  \sum_{1 < \hat{\ell} < \ell} y^{i,s}_{C,\hat{\ell},\ell} &= z^{i,s}_C & \forall \ell: t_\ell \in T, \forall C \in \mathcal{C}^{\operatorname{ext}}_{*, \ell} \label{eq:sumw_jC_jSpeeds:1}\\
  z^{i,s}_{\ell} &= \sum_{C \in \mathcal{C}^{\operatorname{ext}}: C \text{ contains } t_\ell}  z^{i,s}_C  & \forall \ell: t_\ell \in T \label{eq:sumw_jC_jSpeeds:2}\\
  y^{i,s}_{C,\hat{\ell},\ell} &\le z^{i,s}_k  & \forall (k, \hat{\ell}, \ell) : t_\ell \in T \land 1 < \hat{\ell} < k < \ell \label{eq:sumw_jC_jSpeeds:3}\\
  y^{i,s}_{C,\hat{\ell},\ell} &\le 1-z^{i,s}_{\hat{\ell}}  & \forall (\hat{\ell}, \ell) : t_\ell \in T \land 1 < \hat{\ell} < \ell \label{eq:sumw_jC_jSpeeds:4}\\
  0 \le y^{i, s}_{C,\hat{\ell}, \ell} &\le 1  & \forall (\hat{\ell}, \ell): t_\ell \in T, \forall C \in \mathcal{C}^{\operatorname{ext}}_{*, \ell}, 1 < \hat{\ell} < \ell \label{eq:sumw_jC_jSpeeds:5}
\end{align}

We are going to see that the variables $y^{i,s}_{C,\hat{\ell},\ell}$ characterize the busy shift of the external cycle $C$ at time $t_\ell$.
However, before we do so, we adjust the constraints \eqref{eq:sumw_jC_j:first} that we used to indicate a correct split of processing of a job assigned to an external cycle.
Now, we take the busy shift variables $y^{i,s}_{C,\hat{\ell},\ell}$ into account.
\begin{equation}\label{eq:sumw_jC_jSpeeds:6}
  \left( \sum_{p \in [p_{\max}] } p y^{i,s}_{C,L,p} \right) + \left( \sum_{1 < \hat{\ell} < \ell \text{ with } \beta_{\hat{\ell},{\ell}} \neq 0} y^{i,s}_{C,\hat{\ell},\ell} \right) + \left( \sum_{p \in [p_{\max}]} p y^{i,s}_{C,R,p} \right) = \sum_{j \in [d]} p^i_j y^{i,s}_{j,C}    \qquad \forall \ell \in [|T^i|], \forall C \in \mathcal{C}^{\text{ext}}_{*, \ell} \,.
\end{equation}
We stress here that we have added new variables $y^{i,s}_{C,L,0}, y^{i,s}_{C,R,0}$.
Furthermore, for technical reasons, we add the following constraints to ensure that if a job is assigned to an external cycle $C^{\operatorname{ext}}_{k,l}$, then a nonzero fraction of its size is executed prior to the critical time $t_k$.
Note that this is a natural assumption, since otherwise one can use a different cycle decomposition of an identical schedule.
Up until now this property was implied (note that $p \ge 1$ in all sums involved in \eqref{eq:sumw_jC_j:first}), however, in order to be able to prove that the variables $y^{i,s}_{C,\hat{\ell},\ell}$ are set correctly, we have to ensure this.
To that end we add the constraints
\begin{align}
  \sum_{p \in [p_{\max}]  \text{ with } p > s \cdot (t_k - t_{\ell})} y^{i,s}_{C,L,p} + \sum_{\hat{\ell} \in [\ell] \text{ with } \beta_{\hat{\ell},\ell} \neq 0} y^{i,s}_{C,\hat{\ell},k}   &\ge z^{i,s}_{C}
  && \forall k \in [|T^i|], \forall 1 < \ell \le k, C = C^{\text{ext}}_{\ell,k} \label{eq:sumw_jC_jSpeeds:7}
  \\
  \sum_{p \in [p_{\max}]} y^{i,s}_{C,R,p} + \sum_{\hat{\ell} \in [\ell] \text{ with } \beta_{\hat{\ell},k} \neq 0} y^{i,s}_{C,\hat{\ell},k}   &\ge z^{i,s}_{C}
  && \forall k \in [|T^i|], \forall 1 < \ell \le k, C = C^{\text{ext}}_{\ell,k}
  \label{eq:sumw_jC_jSpeeds:8} \,.
\end{align}

As usually, the last type of conditions we discuss now are the constraints constraining the total ``volume'' of the jobs assigned to be scheduled in between two critical times.
This time we make rather straightforward combination of conditions \eqref{eq:CmaxSpeed:cycleVolumeBounds} and \eqref{eq:sumw_jC_j:cycleVolumeBounds} as follows
\begin{equation}\label{eq:sumw_jC_jSpeeds:cycleVolumeBounds}
  \sum_{j \in [d]} \sum_{t_\ell \lhd C \lhd t_k} p^i_j \cdot y^{i,s}_{j,C}
  +
  \sum_{C \in \mathcal{C}^{\operatorname{ext}}_{> \ell,k}} \sum_{p \in [p_{\max}]} p \cdot y^{i,s}_{C, L, p}
  \le \left\lfloor s \cdot (t_k - t_\ell) \right\rfloor
  \qquad \forall k \in [|T^i|], 1 \le \ell \le k \,.
\end{equation}

\begin{lemma}\label{lem:sumw_jC_jSpeeds:betaAlignLemma}
  Let $C = C^{\operatorname{ext}}_{k,\ell}$ be a nonempty external cycle.
  If the busy period for $C$ starts in the critical time $t_{\bar{\ell}}$, then the variable $y^{i,s}_{C,\hat{\ell},\ell}$ is set to $1$ for exactly one $\bar{\ell} \le \hat{\ell} < \ell$ with $\beta_{\hat{\ell}, \ell} = \beta_{\bar{\ell}, \ell}$; otherwise all of these variables are set to $0$.
  Furthermore, we have $z^{i,s}_k = 1$ if and only if there exists a nonempty external cycle $C$ containing $t_k$.
\end{lemma}
\begin{proof}
We begin the proof by observing the auxiliary variables $z^{i,s}_k$, i.e., we prove that $z^{i,s}_k = 1$ if and only if there exists a nonempty external cycle $C$ containing $t_k$.
In order to see this recall first that for an external cycle $C$ the variable $z^{i,s}_C$ is set to $1$ if and only if $C$ is nonempty in the current assignment (this follows from~\eqref{eq:Cmax:cycle_bounds}).
Furthermore, there is at most one nonempty (external) cycle containing any critical time due to their mutual incompatibility, by~\eqref{eq:Cmax:incompatibleCyclesDisable}.
It follows that the right-hand side of \eqref{eq:sumw_jC_jSpeeds:2} is $0$ if and only if all of the summands are $0$ (i.e., when none of the potential external cycles containing $t_k$ is realized) while it is $1$ if and only if some potential external cycle containing $t_k$ is realized.

Now constraints \eqref{eq:sumw_jC_jSpeeds:1} imply that exactly one variable $y^{i,s}_{C, \hat{\ell}, \ell}$ is set to $1$ if a potential external cycle $C$ is realized.
From \eqref{eq:sumw_jC_jSpeeds:4} and \eqref{eq:sumw_jC_jSpeeds:3} it follows that $y^{i,s}_{C, \hat{\ell}, \ell}$ can be set to $1$ if and only if
\begin{itemize}
  \item no potential external cycles containing $t_{\hat{\ell}}$ is realized, and
  \item every critical time between $t_{\hat{\ell}}$ and $t_{\ell}$ is contained in some nonempty external cycle.
\end{itemize}
Note that if the busy period for $C$ starts at $t_{\bar{\ell}}$, then no potential external cycle containing $t_{\bar{\ell}}$ is realized due to the discussion following \eqref{eq:sumw_jC_jSpeeds:7}.
However, the converse is not true!
This can happen for a critical time $t$ if the last job of the schedule produced by our routine ends exactly at time $t$.
We claim that if the busy period for $C$ starts at $t_{\bar{\ell}}$ and every external cycle containing a critical time $t_{\tilde{\ell}}$ is empty, then $\beta^{i,s}_{\hat{\ell}, \ell} = \beta^{i,s}_{\tilde{\ell}, \ell}$ (and $\beta_{\bar{\ell},\hat{\ell}} = 0$).
Note that if this is true, then we are done, since \eqref{eq:sumw_jC_jSpeeds:4} and \eqref{eq:sumw_jC_jSpeeds:3} imply that $y^{i,s}_{C, \hat{\ell}, \ell}$ can be set to $1$ only if $t_{\hat{\ell}}$ is the last critical time prior to $t_\ell$ which does not contain any nonempty external cycle.

Let $t_{\hat{\ell}} = t_{\hat{\ell}_1}, \ldots, t_{\hat{\ell}_q}$ be the critical times between $t_{\hat{\ell}}$ and $t_\ell$ such that no potential external cycle containing a critical time $\bar{t}_{\hat{\ell}_p}$ is realized for $1 \le p \le q$.
Now, the above claim is equivalent to $\beta^{i,s}_{\hat{\ell}, \ell} = \beta^{i,s}_{\hat{\ell}_p, \ell}$ for all $p = 1, \ldots, q$.
Suppose $1 < p < q$.
Since the busy period for $C$ does not start in $t_{\hat{\ell}_p}$, in an optimal schedule there is no idle time between $t_{\hat{\ell}_1}$ and $t_{\hat{\ell}_p}$.
It follows that the size of all of the jobs assigned to cycles between $t_{\hat{\ell}_1}$ and $t_{\hat{\ell}_p}$ completely fills the available processing time on the machine, that is,
\[
  \sum_{t_{\hat{\ell}} \lhd C \lhd t_{\hat{\ell}_p}} \sum_{j = 1}^d p^i_j \cdot y^{i,s}_{j,C} = s \cdot \left( t_{\hat{\ell}_p} - t_{\hat{\ell}} \right) \,.
\]
Consequently, $\fract(s \cdot (t_\ell - t_{\hat{\ell}})) = \fract(s \cdot (t_\ell - t_{\hat{\ell}_p}))$, since the left-hand side of the above expression is integral.
The claim follows and so does the lemma.
\end{proof}

\begin{figure}[bt]

\begin{tikzpicture}

\fill[red!20,draw=black] (4,.1) -- (5.4,.1) -- (5.4,4) -- (4,4) -- cycle;
\draw[draw=black,dotted] (3,.1) -- (5.4,.1) -- (5.4,2) -- (3,2) -- cycle;
\fill[blue!40,draw=black] (4,2) -- (4.6,2) -- (4.6,4) -- (4,4) -- cycle;
\draw[dotted] (3.8,.1) -- (3.8,2);
\draw[dotted] (4.6,.1) -- (4.6,4);

\draw [decorate,decoration={brace,amplitude=3pt,mirror}] (3,.1) -- (3.8,.1);
\node (unitLabel) at (1.2,-.5) {\footnotesize size unit};
\draw[->,>=stealth] (unitLabel.east) to[out=0,in=270] (3.4,-.1);

\draw [decorate,decoration={brace,amplitude=3pt,mirror}] (4,.1) -- (4.6,.1);
\node[text width=3cm] (shiftLabel) at (7,-.5) {\footnotesize critical time busy period shift \mbox{$1 - \operatorname{frac}(s \cdot (t_k - t_\ell))$}};
\draw[->,>=stealth] (shiftLabel.west) to[out=180,in=270] (4.3,-.1);

\fill[gray!60] (0,3) -- (4,3) -- (4,3.7) -- (0,3.7) -- cycle;
\node at (2,3.3) {busy period};

\draw[thick] (0,0) to (0,5) node[yshift = .2cm] {$s \cdot t_{\ell}$};
\draw[thick] (4,0) to (4,5) node[yshift = .2cm] {$s \cdot t_{k}$};;

\end{tikzpicture}
  \caption{\label{fig:criticalTimeBusyPeriodShift}%
  A visualization of a busy period shift.
  }
\end{figure}
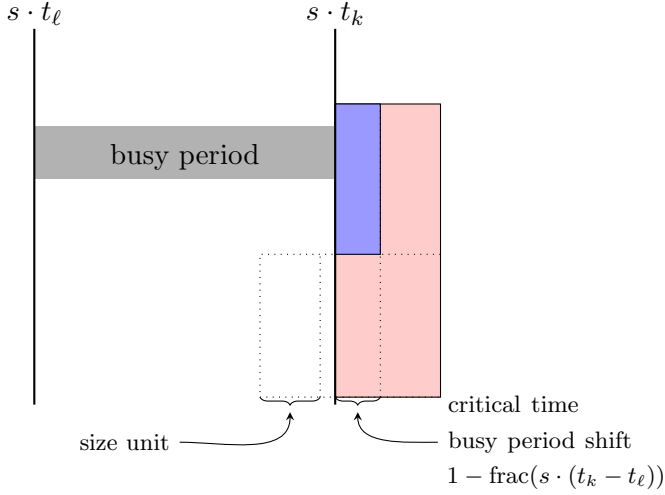

We are now going to use the new variables to extend the previous definition of the objective functions to the setting with speeds.

\paragraph{Completion Time.}
In the following discussion refer to the segment of the 2D Gannt chart for the $\sum w_jC_j$ objective in Figure~\ref{fig:sumw_jC_jSegmentWithBusyShift} (which itself is a simple combination of Figures~\ref{fig:sumw_jC_jSegment} and \ref{fig:criticalTimeBusyPeriodShift}).
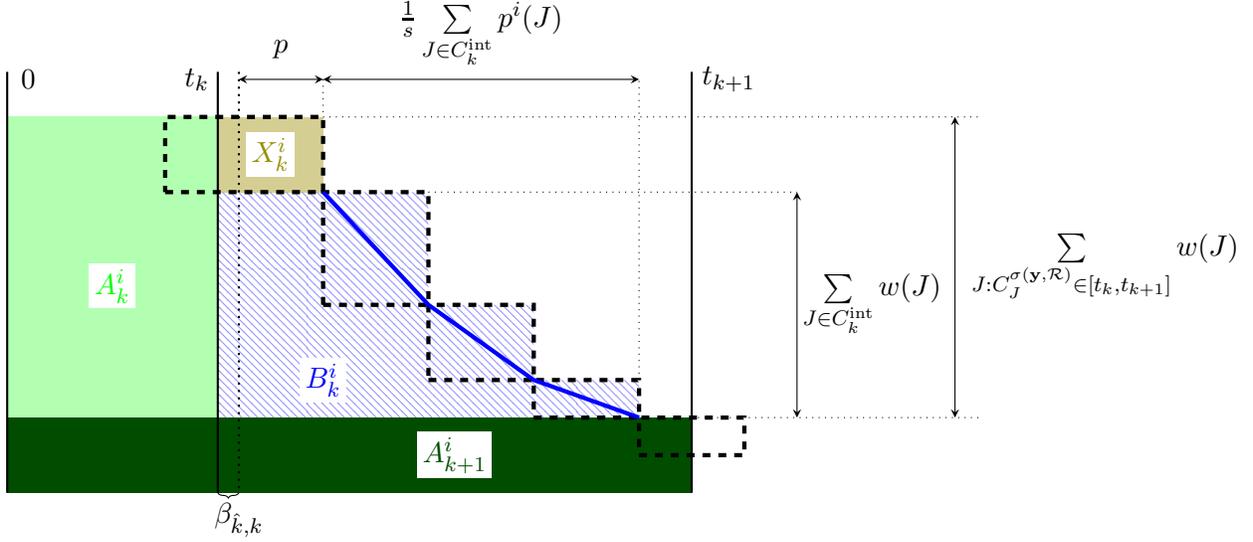
\begin{figure}[bt]
  \usetikzlibrary{patterns}

\begin{tikzpicture}[yscale=.5,xscale=.7]
\tikzstyle{job}=[ultra thick,dashed]
\tikzstyle{fakeJob}=[dotted,red,ultra thick]

\fill[green!30!black] (-4,-1) -- (-4,1) -- (9,1) -- (9,-1) -- cycle;
\fill[green!30] (0,1) -- (0,9) -- (-4,9) -- (-4,1) -- cycle;

\fill[olive!40] (0,9) -- (0,7) -- (2,7) -- (2,9) -- cycle;

\fill[pattern=north west lines,pattern color=blue!40] (0,7) -- (2,7) -- (4,7) -- (4,4) -- (6,4) -- (6,2) -- (8,2) -- (8,1) -- (0,1) -- cycle;

\draw[thick,dotted] (0.4,10.2) -- (0.4,-1);
\draw[thick] (0,10.2) -- (0,-1);
\draw[thick] (9,10.2) -- (9,-1);
\node at (-.4,10) {$t_k$};
\node at (9.7,10) {$t_{k+1}$};
\draw[thick] (-4,10.2) -- (-4,-1);
\node at (-3.6,10) {$0$};

\draw [decorate,decoration={brace,amplitude=2pt,mirror}] (0,-1) -- (.4,-1) node[below] {$\beta_{\hat{k},k}$};

\draw[dotted] (2,10) -- (2,9);
\draw[dotted] (8,2) -- (8,10);
\draw[<->,>=stealth] (0.4,10) to node[midway,yshift=.4cm] {$p$} (2,10);
\draw[<->,>=stealth] (2,10) to node[midway,yshift=.6cm] {$\frac{1}{s} \sum\limits_{J \in C^{\operatorname{int}}_k} p^i(J)$} (8,10);
\draw[<->,>=stealth] (11,7) to node[midway,xshift=1cm] {$\sum\limits_{J \in C^{\operatorname{int}}_k} w(J)$} (11,1);
\draw[<->,>=stealth] (14,9) to node[midway,xshift=2cm] {$\sum\limits_{J : C^{\sigma(\vey,\mathcal{R})}_J \in \interval{t_k}{t_{k+1}}} w(J)$} (14,1);

\draw[job] (-1,7) -- (-1,9) -- (2,9) -- (2,7) -- cycle;
\draw[job] (2,7) -- (4,7) -- (4,4) -- (2,4) -- cycle;
\draw[job] (4,4) -- (6,4) -- (6,2) -- (4,2) -- cycle;
\draw[job] (6,2) -- (8,2) -- (8,1) -- (6,1) -- cycle;
\draw[job] (8,1) -- (10,1) -- (10,0) -- (8,0) -- cycle;

\draw[dotted] (4,7) -- (11,7);
\draw[dotted] (9,1) -- (14.5,1);
\draw[dotted] (2,9) -- (14.5,9);

\draw[blue, ultra thick] ((2,7) -- (4,4) -- (6,2) -- (8,1);

\node[fill=white!10,text=blue,inner sep=2pt] at (2,2) {$B^i_k$};
\node[fill=white!10,text=green!30!black,inner sep=2pt] at (4.5,0) {$A^i_{k+1}$};
\node[fill=white!10,text=green,inner sep=2pt] at (-2,4.5) {$A^i_k$};
\node[fill=white!10,text=olive,inner sep=2pt] at (1,8) {$X^i_k$};

\end{tikzpicture}
  \caption{\label{fig:sumw_jC_jSegmentWithBusyShift}%
  An effect of the busy period shift on the expression of $\sum w_jC_j$ for the segment between $t_k$ and $t_{k+1}$.
  }
\end{figure}
First observe that the area marked $A$ remains the same.
On the other hand, areas $B$ and $X$ are indeed affected by the current shift.
We now show that still their area (i.e., the respective contribution to the objective) can be expressed with only a minor change using the newly added variables $y^{i,s}_{C,\hat{\ell},\ell}$.
Before we do so, we add further auxiliary variables $y^{i,s}_{k,p,\beta,j}$, for each $t_k \in T, p \in [p_{\max}], \beta \in \mathcal{B}^{i,s}, j \in [d]$, which express the overall shift at $t_k$ using the variables $y^{i,s}_{C,\hat{\ell},\ell}$.
To that end we would like to enforce the nonlinear constraints
\[
  y^{i,s}_{k,p,\beta,j} = \sum_{C \in \mathcal{C}^{\operatorname{ext}}_{*,k}} \sum_{p \in [p_{\max}]} \sum_{\hat{k}: \beta = \beta_{\hat{k},k}}  y^{i,s}_{j,C,R,p} \cdot y^{i,s}_{C,\hat{k},k}
  \qquad\qquad
  \forall k : t_k \in T \,.
\]
Fortunately, the multiplication in the above constraint can be linearized, since only binary variables are involved in it (via a simple trick we have already seen).
Indeed adding the following equivalent system of linear constraints (for every $k \in [2,|T|]$) does the job:\dkcom{either this of the former four constraints---needs to be numbered as it is used in proofs}
\begin{equation}\label{eq:sumw_jC_jSpeeds:9}
\begin{rcases}
  y^{i,s}_{C,k,p,\beta,j} &\ge - 1 + y^{i,s}_{j,C,R,p} + y^{i,s}_{C,\hat{k},k} \\
  y^{i,s}_{C,k,p,\beta,j} &\le y^{i,s}_{j,C,R,p} \\
  y^{i,s}_{C,k,p,\beta,j} &\le y^{i,s}_{C,\hat{k},k}  \\ 
  y^{i,s}_{k,p,\beta,j} &= \sum_{C \in \mathcal{C}^{\operatorname{ext}}_{*,k}} y^{i,s}_{C,k,p,\beta,j} \\
  y^{i,s}_{k,\beta} &= \sum_{j = 1}^d \sum_{p \in [p_{\max}]} y^{i,s}_{k,p,\beta,j}
\end{rcases}
\forall C \in \mathcal{C}^{\operatorname{ext}}_{*,k}, \forall 1 \le \hat{k} < k, \forall \beta \in \mathcal{B}^{i,s} : \beta = \beta_{\hat{k},k}, \forall p \in [p_{\max}] \forall j \in [d]
\end{equation}
It is straightforward to verify that the first three constraints assure that $y^{i,s}_{C,k,p,\beta,j} = y^{i,s}_{j,C,R,p} \cdot y^{i,s}_{C,\hat{k},k}$.
We conclude that $y^{i,s}_{k,p,\beta,j} = 1$ if and only if there is a job of type $j$ scheduled to an external cycle $C \in \mathcal{C}^{\operatorname{ext}}_{*,k}$ and the size to be processed after $t_k$ is exactly $p + \beta$.
Finally, we have that $y^{i,s}_{k,\beta} = 1$ if and only if there exists $j \in [d]$ and $p \in [p_{\max}]$ such that $y^{i,s}_{k,p,\beta,j} = 1$ (and $y^{i,s}_{k,\beta} = 0$ otherwise).

\begin{algorithm}[bt]
	\SetKwProg{Def}{def}{:}{}
	\SetKwFunction{cycleHandler}{handleCycle}
	\SetKwFunction{scheduleEnd}{endOf}
	\DontPrintSemicolon
	\Def{\cycleHandler{$i, \sigma, C, \vey_C$}}{
		\If{$C \in \mathcal{C}^{\operatorname{int}}$}{
			\nlset{Order}\label{alg:sigmaFromYWithOrderAndSpeeds:leftCritical}\ForEach{$j \in [d]$ in order $\preceq^i_{\mathcal{R},\leftCritical(C)}$}{
				\For{$\ell = 1$ \KwTo $\vey_{j,C}$}{
					$t \leftarrow \max($\scheduleEnd{$\sigma$}$, \leftCritical(C))$ \;
					$\sigma \leftarrow \sigma \cup \left\{ \left( j, \interval{t}{t + \frac{p^i_j}{s^i}} \right) \right\}$ \;
				}
			}
		}
		\Else{
      Let $k$ be such that $C \in \mathcal{C}^{\operatorname{ext}}_{*,k}$ \;
			Let $p,\beta,j$ be such that $y_{k,p,\beta,j} = 1$ \;
			\nlset{ext}\label{alg:sigmaFromYWithOrderAndSpeeds:externalCritical}$\sigma(\vey) \leftarrow \sigma(\vey) \cup \left\{ \left( j, \interval{t_k - \frac{p^i_j - p - \beta}{s^i}}{t_k + \frac{p + \beta}{s^i}} \right) \right\}$ \;
		}
	}

	\caption{\label{alg:sigmaFromYWithOrderAndSpeeds}
		We only redefine the \texttt{handle\_cycle} function, the rest of the algorithm is identical to Algorithm~\ref{alg:sigmaFromY}.
		The algorithm computes a left aligned schedule from a vector $\vey$ by placing jobs in internal cycles in the orders $\preceq^i_{\mathcal{R},k}$ and furthermore takes speeds into account.
	}
\end{algorithm}

\begin{lemma}\label{lem:sumw_jC_jSpeeds:sigmaFromY}
  Fix a machine of kind $i \in [\kappa]$ and speed $s$ and let $\vex \in \mathbb{N}^d$.
  There exists an \mbox{$(\mathbb{N} + \mathcal{B}^{i,s})$-regular} schedule of $\vex$ if and only if there exist $\vey,\vez$ such that all of the constraints \eqref{eq:Cmax:sumOfExes}--\eqref{eq:Cmax:incompatibleCyclesDisable}, \eqref{eq:sumw_jC_j:second}--\eqref{eq:sumw_jC_j:fourth}, and \eqref{eq:sumw_jC_jSpeeds:1}--\eqref{eq:sumw_jC_jSpeeds:9} are satisfied.
  Moreover, the schedule $\sigma(\vey)$ given by Algorithm~\ref{alg:sigmaFromYWithOrderAndSpeeds} is one such schedule.
  Furthermore, $y^{i,s}_{C,k,p,\beta,j} = 1$ for a potential external cycle $C = C^{\operatorname{ext}}_{\ell,k}$ for $\ell,k \in \left\{ 2, \ldots, |T^i| - 1 \right\}$ with $\ell \le k$
  if and only if
  there is a job $J$ of type $j$ assigned to a realization of $C$ in $\sigma(\vey)$ such that at time $t_k$ exactly $\frac{p^i_j - p - \beta}{s}$ units of the total processing time of the job $J$ on the fixed machine are processed.
\end{lemma}
\begin{proof}
Suppose there exists an $(\mathbb{N} + \mathcal{B}^{i,s})$-regular schedule of $\vex$ and let $\mathcal{D}$ be its cycles (i.e., the potential cycles realized in the $(\mathbb{N} + \mathcal{B}^{i,s})$-regular schedule).
Now, we use the arguments of Lemma~\ref{lem:leftAlignedScheduleCMax} to derive the vector $\vex$ and some entries in the vector $\vey$ that together satisfy conditions \eqref{eq:Cmax:sumOfExes}--\eqref{eq:Cmax:incompatibleCyclesDisable}.
The rest of the $\vey$ variables is centered around external cycles realized in~$\mathcal{D}$.
Let $C$ be an external cycle realized in~$\mathcal{D}$ and assume it is a realization of the potential cycle in~$\mathcal{C}^{\operatorname{ext}}_{*,k}$.
Let $J$ be the job assigned to $C$ in the assumed $(\mathbb{N} + \mathcal{B}^{i,s})$-regular schedule of $\vex$.
Let $p_L$ be the amount of the size of $J$ that is already processed at time $t_k$; note that $p_L < p^i(J) / s$ and let $p_R = (p^i(J) / s) - p_L$.
Formally, we have
\begin{equation}\label{eq:sumw_jC_jSpeeds:settingPLandPR}\tag{$p_L, p_R$}
  p_R = \frac{C(J) - t_k}{s}
  \qquad\text{and}\qquad
  p_L = \frac{p^i(J) - s \cdot (C(J) - t_k)}{s} = \frac{p^i(J)}{s} - (C(J) - t_k) \,,
\end{equation}
where $C(J)$ is the completion time of the job $J$.
Note that we have $0 < p_L, p_R < p^i(J) / s$.
In order to satisfy conditions \eqref{eq:sumw_jC_j:second}--\eqref{eq:sumw_jC_j:fourth} we have to set one variable $y^i_{C,L,p}$ to one (i.e., for some $p \in \{ 0, \ldots, p_{\max} \}$) as well as one variable $y^i_{C,R,p}$, since we have $z^i_C = 1$.
Let $p = \lfloor s \cdot p_L \rfloor$ and $q = \lfloor s \cdot p_R \rfloor$ and let us set
\[
  y^i_{C,L,p} = 1 \qquad\qquad\text{and}\qquad\qquad y^i_{C,R,q} = 1
\]
and all other variables $y^i_{C,L,\hat{p}},y^i_{C,R,\hat{p}}$ we set to $0$.
Clearly, such an assignment satisfies conditions \eqref{eq:sumw_jC_j:second}--\eqref{eq:sumw_jC_j:fourth}.
Now, it remains to argue about conditions \eqref{eq:sumw_jC_jSpeeds:1}--\eqref{eq:sumw_jC_jSpeeds:cycleVolumeBounds}.
Satisfying \eqref{eq:sumw_jC_jSpeeds:2} is straightforward---we have to set the variable $z^{i,s}_\ell$ to $1$ if and only if in~$\mathcal{D}$ there exists an external cycle containing the critical time $t_\ell$.
We set the variable $y^{i,s}_{C,\hat{\ell},k}$ to $1$ (recall that $C$ is an external cycle realized in~$\mathcal{D}$ and $t_k$ is the last critical time contained in its interior) if and only if
\begin{enumerate}[label=(\emph{\alph*})]
  \item\label{it:sumw_jC_jSpeeds:settingYLK:1}
  the critical time $t_{\hat{\ell}}$ is not contained in (the interior of) an external cycle in~$\mathcal{D}$ and
  \item\label{it:sumw_jC_jSpeeds:settingYLK:2}
  all critical times $t_\ell$ with $\hat{\ell} < \ell < k$ are contained in (the interior of) an external cycle in~$\mathcal{D}$,
\end{enumerate}
otherwise we set it to $0$.
Clearly, the conjunction of the the two hold for exactly one $\hat{\ell}$ (thus, we have verified \eqref{eq:sumw_jC_jSpeeds:1}, since we have $z^{i,s}_C = 1$).
The condition \eqref{eq:sumw_jC_jSpeeds:4} is fulfilled by~\ref{it:sumw_jC_jSpeeds:settingYLK:1}, since by the above setting we have $z^{i,s}_{\hat{\ell}} = 0$ as $t_{\hat{\ell}}$ is not contained in an external cycle in~$\mathcal{D}$.
The condition \eqref{eq:sumw_jC_jSpeeds:3} is fulfilled by~\ref{it:sumw_jC_jSpeeds:settingYLK:2}, since by the above setting we have $z^{i,s}_{\ell} = 1$ for all $\ell$ for which the critical time $t_\ell$ is contained in the interior of an external cycle in~$\mathcal{D}$.
Clearly, the above setting assigns the discussed variables to either $0$ or $1$ and thus~\eqref{eq:sumw_jC_jSpeeds:5} follows.
Now, we verify~\eqref{eq:sumw_jC_jSpeeds:6}.
First observe that by the above setting of $p,q$ we have that $p^i(J) - 1 \le p + q \le p^i(J)$.
Furthermore, we have that $p + q = p^i(J)$ if and only if $t_k$ and $t(J)$ (the critical time in which the busy period of $J$ starts) align, that is, when $\operatorname{frac}(s \cdot (t_k - t(J))) = 0$.
Thus, the condition~\eqref{eq:sumw_jC_jSpeeds:6} can be verified for $C$, since using the above assignment, it boils down to verify \(p + q + (\operatorname{frac}(s \cdot (t_k - t(J)))) = p^i(J)\).
Verification of~\eqref{eq:sumw_jC_jSpeeds:7} and~\eqref{eq:sumw_jC_jSpeeds:8} goes back to their intended meaning.
First, the left-hand side of~\eqref{eq:sumw_jC_jSpeeds:7} is at least one if either
\begin{itemize}
  \item $p > s \cdot (t_k - t_\ell)$, where $\ell$ is chosen such that $C = C^{\operatorname{ext}}_{\ell,k}$,
  or
  \item $\beta_{\hat{\ell},\ell} \neq 0$, where $\ell$ is chosen such that $C = C^{\operatorname{ext}}_{\ell,k}$ and $t_{\hat{\ell}} = t(J)$.
\end{itemize}
In order to see that at least one of the above conditions holds, we observe that, since $C$ is a realization of $C^{\operatorname{ext}}_{\ell,k}$, we have $\frac{p^i(J)}{s} - (C(J) - t_\ell) > 0$, where $C(J)$ is the completion time of $J$.
This implies that
\[
  0 < \frac{p^i(J)}{s} - (C(J) - t_k) - (t_k - t_\ell) = p_L - (t_k - t_\ell)
\]
and consequently
\[
  p_L > t_k - t_\ell \,.
\]
Now, since $p = \lfloor s \cdot p_L \rfloor$, we get that $p \ge \lfloor s \cdot (t_k - t_\ell) \rfloor$.
Thus, either $p > s \cdot (t_k - t_\ell)$ in which case the left hand side of~\eqref{eq:sumw_jC_jSpeeds:7} amounts to at least one, as the sum contains the variable $y^{i,s}_{C,L,p}$; or $p - s \cdot (t_k - t_\ell) < 0$.
In the later case, we get $t(J)$ and $t_\ell$ do not align, since if that is the case, then $t_\ell$ is not contained in the interior of $C$ and thus we obtain a contradiction with $C$ being a realization of $C^{\operatorname{ext}}_{\ell,k}$.
As now $t(J)$ and $t_\ell$ do not align, we conclude that $y^{i,s}_{C,\hat{\ell},k}$ is contained in the second sum in~\eqref{eq:sumw_jC_jSpeeds:7} and the above constructed assignment sets it to $1$.
Thus, we have verified the condition~\eqref{eq:sumw_jC_jSpeeds:7} for $C$.
One can verify~\eqref{eq:sumw_jC_jSpeeds:8} by a similar argument.
We know that $z^{i,s}_{C} = 1$ and thus $t_k$ is in the interior of the cycle $C$.
Suppose now that $y^{i,s}_{C,R,p} = 0$ for all $p \in [p_{\max}]$, that is, $0 < p_R < 1$.
We know by Lemma~\ref{lem:sumw_jC_jSpeeds:betaAlignLemma} that $p_R = \beta_{\hat{\ell},k}$ for some $\hat{\ell} \in [\ell]$.
Furthermore, since $p_R > 0$ we get that $\beta_{\hat{\ell},k} \neq 0$ and thus \eqref{eq:sumw_jC_jSpeeds:8} is satisfied for $C$.

Now, we set all other variables (e.g., those associated with a potential external not realized in~$\mathcal{D}$) to $0$.
This way we satisfy conditions \eqref{eq:sumw_jC_jSpeeds:1}--\eqref{eq:sumw_jC_jSpeeds:8}:
All equations evaluate both left and right hand sides to $0$ thus we clearly satisfy conditions \eqref{eq:sumw_jC_jSpeeds:1} and \eqref{eq:sumw_jC_jSpeeds:5}--\eqref{eq:sumw_jC_jSpeeds:8}.
Observe further that \eqref{eq:sumw_jC_jSpeeds:2} sets the variable $z^{i,s}_\ell$ to $0$ if and only if all external cycles containing $t_\ell$ are not realized in $\mathcal{D}$; otherwise we have set it to $1$ (note that there is at most one external cycle containing the critical time $t_\ell$, since these cycles are incompatible).
Consequently, we have verified conditions~\eqref{eq:sumw_jC_jSpeeds:2}.
Moreover, we know that $z^{i,s}_\ell \in \{0,1\}$ for all $\ell \in [|T|]$.
Thus, the right hand-sides in \eqref{eq:sumw_jC_jSpeeds:3} and \eqref{eq:sumw_jC_jSpeeds:4} are in $\{0,1\}$ and thus satisfied for $C$ not realized in $\mathcal{D}$, i.e., with $y^{i,s}_{C, \hat{\ell}, \ell} = 0$.

It remains to verify~\eqref{eq:sumw_jC_jSpeeds:cycleVolumeBounds}.
This directly follow from $\mathcal{D}$ being a cycle decomposition of an admissible ($(\mathbb{N} + \mathcal{B}^{i,s})$-regular) schedule for $\vex$.
Fix $\ell,k$ such that $1 \le \ell < k \le |T^i|$.
First, if no cycle in~$\mathcal{D}$ is a realization of a potential cycle in $\mathcal{C}^{\operatorname{ext}}_{\ge,\ell}$, then the left hand side of~\eqref{eq:sumw_jC_jSpeeds:cycleVolumeBounds} is the total size of all jobs scheduled for processing between $t_\ell$ and $t_k$.
Let $a$ be the value of the left hand side.
We know that $a / s \le t_k - t_\ell$, since~$\mathcal{D}$ is a cycle decomposition of an admissible schedule.
This is equivalent to $a \le s \cdot (t_k - t_\ell)$ and, since $a \in \mathbb{N}$, even to $a \le \lfloor s \cdot (t_k - t_\ell) \rfloor$ and thus~\eqref{eq:sumw_jC_jSpeeds:cycleVolumeBounds} follows in this case.
Second, suppose there is a cycle in~$\mathcal{D}$ realizing a potential external cycle $C^{\operatorname{ext}}_{\hat{\ell},k}$ for some $\hat{\ell}$ with $\ell \le \hat{\ell} \le k$.
Again let $a$ denote the total size of all the jobs $\mathcal{D}$ assigns to cycles between $t_\ell$ and $t_k$ (i.e., $a$ is the value of the first sum in~\eqref{eq:sumw_jC_jSpeeds:cycleVolumeBounds}).
Let $p$ be such that we have set the value of $y^{i,s}_{C,L,p}$ to $1$.
Now, again we have $(a+p)/s \le (t_k - t_\ell)$ which yields $a+p \le s \cdot (t_k - t_\ell)$ and, since $a,p \in \mathbb{N}$, we get $a+p \le \lfloor s \cdot (t_k - t_\ell) \rfloor$.
This finishes the first part of the proof.

Now, suppose we are given vectors $\vex,\vey,\vez$ that together fulfill the presented model, i.e., satisfy the conditions \eqref{eq:Cmax:sumOfExes}--\eqref{eq:Cmax:incompatibleCyclesDisable}, \eqref{eq:sumw_jC_j:second}--\eqref{eq:sumw_jC_j:fourth}, and \eqref{eq:sumw_jC_jSpeeds:1}--\eqref{eq:sumw_jC_jSpeeds:cycleVolumeBounds}.
Observe that it is possible to extend $\vey$ so that it satisfies \eqref{eq:sumw_jC_jSpeeds:9} (as already discussed, the first three conditions are equivalent to a product of binary variables and the last condition can be directly fulfilled using the new variable it introduces).
As we already know, conditions \eqref{eq:Cmax:sumOfExes}--\eqref{eq:Cmax:incompatibleCyclesDisable} ensure that
\begin{itemize}
  \item
  $x_j$ is the total number of jobs of type $j$ to be processed on the fixed machine,
  \item
  $z_C = 1$ for an external cycle $C$ if $\sum_{j \in [d]} y_{j,C} = 1$ (that is, if exactly one job is assigned to the potential cycle $C$), and
  \item
  if for a potential external cycle $C$ we have $y_{j,C} = 1$ for some $j \in [d]$, then $y_{j,C,L,q} = 1$ and $y_{j,C,R,p} = 1$ for some $p,q \in \left\{ 0,\ldots, p^i_j \right\}$.
\end{itemize}
It is straightforward to verify that the conditions \eqref{eq:sumw_jC_jSpeeds:cycleVolumeBounds} applied for $\ell \in [|T| - 1]$ and $k = \ell + 1$ implies that for a potential internal cycle $C_\ell \in \mathcal{C}^{\operatorname{int}}$ we have
\[
 \sum_{j \in [d]} p^i_j \cdot y_{j,C} \le \lfloor s \cdot (t_{\ell + 1} - t_\ell) \rfloor
 \qquad\text{consequently}\qquad
 \sum_{j \in [d]} \frac{p^i_j}{s} \cdot y_{j,C} \le t_{\ell + 1} - t_\ell \,.
\]
It follows from our discussion in the first part of the proof that if
$\vey$ satisfies \eqref{eq:sumw_jC_jSpeeds:1}--\eqref{eq:sumw_jC_jSpeeds:5} and Lemma~\ref{lem:sumw_jC_jSpeeds:betaAlignLemma}, then for a potential external cycle $C \in \mathcal{C}^{\operatorname{ext}}_{*,\ell}$ to which a job $J$ is assigned we have $y_{C,\hat{\ell},\ell} = 1$ for some $\hat{\ell} \le \ell$ with $\beta_{\bar{\ell},\hat{\ell}} = 1$, where $t_{\bar{\ell}} = t(J)$.

Let $\sigma(\vey, \mathcal{R})$ be the schedule output by Algorithm~\ref{alg:sigmaFromYWithOrderAndSpeeds}.
We claim that $\sigma(\vey, \mathcal{R})$ is an admissible schedule.
Note that each job scheduled in $\sigma(\vey, \mathcal{R})$ can only be processed after its release time.
Let $J$ be a job assigned to a potential cycle $C$ in $\vey$ and let $t_k = \operatorname{right}(C)$ and $t_\ell = t(J)$.
  Suppose that $C$ is a potential internal cycle.
  Then, \eqref{eq:sumw_jC_jSpeeds:cycleVolumeBounds} when applied to $t_\ell$ and $t_k$ together with the above observation implies that $J$ is assigned to a realization of $C$ in $\sigma(\vey, \mathcal{R})$.
  Consequently, $J$ must finish prior to $t_k \le d^i(J)$ (which holds due to \eqref{eq:chiDefinition}).
  Now suppose $C$ is a potential external cycle $C = C^{\operatorname{ext}}_{\hat{\ell},k}$ and let $\hat{\mathcal{C}}$ be the collection of potential cycles $\hat{C}$ with $\sum_{j \in [d]} y_{j,\hat{C}} \ge 1$ and $t_\ell \le \operatorname{left}(\hat{C})$ and $\hat{\ell} \le \operatorname{right}(\hat{C})$.
  Suppose further that $y_{k,p,\beta,j} = 1$ and note that in this case we have $y_{k,p,\beta,j} = y_{C,k,p,\beta,j}$.
  Now, \eqref{eq:sumw_jC_jSpeeds:cycleVolumeBounds} applied to $t_{k}$ and $t_\ell$ implies that
  \[
    \sum_{j \in [d]} \left( \sum_{C \in \hat{\mathcal{C}}} p^i_j \cdot y_{j,C} \right) + p \le \lfloor s \cdot (t_{k} - t_\ell) \rfloor \,.
  \]
  As a consequence we get
  \[
    \frac{\left( \sum_{j \in [d]} \sum_{C \in \hat{\mathcal{C}}} p^i_j \cdot y_{j,C} \right) + p}{s} - (t_{k} - t_\ell) = \beta_{\ell,k} + q \,,
  \]
  for some $q \in \mathbb{N}$.
  Finally, Lemma~\ref{lem:sumw_jC_jSpeeds:betaAlignLemma} implies that $\beta = \beta_{\ell,k}$.
  We conclude that $\sigma(\vey, \mathcal{R})$ is an admissible schedule and exactly $\frac{p + \beta}{s}$ units of $p^i(J)$ are already processed at time $t_{k}$.
\end{proof}

Now, we know our new model yields (via $\sigma(\vey, \mathcal{R})$) a feasible schedule with stronger properties for jobs assigned to a (realization of) potential external cycle.
We are going to use these properties to yet again compute the objective---sum of weighted completion times.
This we do similarly to before, since we know that
\[
  \sum_{J \in \mathcal{J}(\vex)} w(J) \cdot C^{\sigma(\vey,\mathcal{R})}_J
  =
  \sum_{k \in [|T^i|]} \left( A^{i,s}_k(\vey) + B^{i,s}_k(\vey) + X^{i,s}_k(\vey) \right) \,.
\]
The areas (quantities) $A^{i,s}_k, B^{i,s}_k, X^{i,s}_k$ are defined as in the previous section.
Again we first begin with computation of simpler linear parts of the objective and later we move our attention to the convex nonlinear part ($B^{i,s}_k$).
\begin{lemma}\label{lem:sumw_jC_jSpeeds:computingAandX}
  Fix a machine of kind $i$ with speed $s$ and let $(\vex, \vey, \vez)$ satisfy constraints~\eqref{eq:Cmax:sumOfExes}--\eqref{eq:Cmax:incompatibleCyclesDisable}, \eqref{eq:sumw_jC_j:second}--\eqref{eq:sumw_jC_j:fourth}, and \eqref{eq:sumw_jC_jSpeeds:1}--\eqref{eq:sumw_jC_jSpeeds:9}.
  The functions $A^{i,s}_k$ and $X^{i,s}_k$ are linear in $\vey$, for all $k \in [|T^i| - 1]$.
\end{lemma}
\begin{proof}
  First note that by the definition of $A^{i,s}_k$ this part is equal to the $t_k$ multiple of sum of weights of all the jobs assigned to $C^{\operatorname{int}}_k$ and $\mathcal{C}^{\operatorname{ext}}_{*,k}$.
  Thus, following the same arguments as in Lemma~\ref{lem:sumw_jC_j:computingAandX} we arrive to
  \begin{equation}\tag{$A^{i,s}_k$}
    A^{i,s}_k
    =
    t_k \cdot \left( \sum_{j \in [d]} w_j \cdot y_{j, C^{\operatorname{int}}_k} + \sum_{C \in \mathcal{C}^{\operatorname{ext}}_{*,k}}\sum_{j \in [d]} w_j \cdot y_{j,C} \right) \,.
  \end{equation}
  That is, this part remains unaffected by the presence of speed.

  The part $X^{i,s}_k$ is going to be fractional.
  Suppose there is a job $J$ of type~$\hat{j}$ assigned to a potential cycle in $\mathcal{C}^{\operatorname{ext}}_{*,k}$, since otherwise we have $X^{i,s}_k = 0$ (and we will show that in such a case our formula is correct as well).
  We already know (Lemma~\ref{lem:sumw_jC_jSpeeds:sigmaFromY}) that the completion time of $J$ in $\sigma(\vey, \mathcal{R})$ is exactly \( t_k + \frac{p+\beta}{s} \) if and only if $y_{k,p,\beta,j} = 1$.
  Thus, in such a case we have $X^{i,s}_k = \frac{p+\beta}{s} \cdot w(J)$.
  We set
  \begin{equation}\tag{$X^{i,s}_k$}
    X^{i,s}_k
    =
    \frac{1}{s} \cdot \sum_{p \in [p_{\max}]}\sum_{\beta \in \mathcal{B}^{i,s}}\sum_{j \in [d]} \left( w_j \cdot \left(p+\beta\right) \cdot y_{k,p,\beta,j} \right) \,.
  \end{equation}
  Clearly, the above expression is a linear function in $\vey$.
  If there is no job assigned to an external cycle in $\mathcal{C}^{\operatorname{ext}}_{*,k}$, then we have $y_{k,p,\beta,j} = 0$ for all $p \in [p_{\max}], \beta \in \mathcal{B}^{i,s}, j\in [d]$.
  Thus, the right hand side of the above expression is $0$ and so in such a case it is a correct expression as well.
  Now, if there is a job $J$ (of type~$\hat{j}$) assigned to a realization of an external cycle in $\mathcal{C}^{\operatorname{ext}}_{*,k}$, then exactly one variable $y_{k,\hat{p},\hat{\beta},\hat{j}} = 1$ (and the rest of $y_{k,p,\beta,j}$ is set to $0$, due to~\eqref{eq:sumw_jC_jSpeeds:8}).
  Consequently, the right hand side becomes \(\frac{1}{s} \cdot w_{\hat{j}} \cdot (\hat{p} + \hat{\beta})\) which is correct, since we know that $\frac{\hat{p}+\hat{\beta}}{s}$ is the difference between the completion time of $\hat{J}$ in $\sigma(\vey,\mathcal{R})$ and the critical time $t_k$.
\end{proof}

Now, to keep our argumentation as simple as possible and close to the arguments given in the previous section, we in fact compute the $s$ multiple of the part $B^{i,s}_k$.
This allows us to keep the idea behind the $\ve\alpha^{i,s}$ variables as before; we only introduce a specialized variable for every $\beta \in \mathcal{B}^{i,s}$.
In fact, we redefine $\alpha^{i,s}_{k,\beta,0}$ as follows
\begin{equation}\label{eq:sumw_jC_jSpeeds:alphaZero}
  \alpha^{i,s}_{k,\beta,0}
  =
  \sum_{j \in [d]} \sum_{p \in [p_{\max}]} p \cdot y^{i,s}_{k,p,\beta,j} \qquad\qquad \forall k \in \{2,\ldots,|T^i|\}, \, \forall \beta \in \mathcal{B}^{i,s}
\end{equation}
and similarly for the other variables, that is, we have
\begin{equation}\label{eq:sumw_jC_jSpeeds:last}
  \alpha^{i,s}_{k,\beta,j}
  =
  p^i_j\cdot y^{i,s}_{j,C} + \alpha^{i,s}_{k,\beta,\pred^i_{k,\mathcal{R}}(j)} \qquad\qquad \forall j \in [d], \, \forall \beta \in \mathcal{B}^{i,s} \,,
\end{equation}
where $C = C^{\operatorname{int}}_k$.
Although now $\ve\alpha$ is defined differently, the idea behind its definition remains the same.
Observe that $\alpha^{i,s}_{k,\beta,j}$ is in fact the $s$-multiple of the completion time in $\sigma(\vey,\mathcal{R})$ minus $\beta$ of the last job assigned to $C^{\operatorname{int}}_k$.

Similarly to Lemma~\ref{lem:sumw_jC_j:objectiveVariablesFollow} we observe that if there is a solution for our model without the variables we added in order to compute the objective, then it is possible to compute the values satisfying constraints \eqref{eq:sumw_jC_jSpeeds:8}--\eqref{eq:sumw_jC_jSpeeds:last}.
\begin{lemma}\label{lem:sumw_jC_jSpeeds:objectiveVariablesFollow}
  It for a vector $\vex \in \mathbb{N}^d$ there exist integral vectors $\vey,\vez$ satisfying constraints~\eqref{eq:Cmax:sumOfExes}--\eqref{eq:Cmax:incompatibleCyclesDisable}, \eqref{eq:sumw_jC_j:first}--\eqref{eq:sumw_jC_j:fourth}, and \eqref{eq:sumw_jC_jSpeeds:1}--\eqref{eq:sumw_jC_jSpeeds:9}, then there exists an integral $\vey,\ve\alpha$ satisfying~\eqref{eq:Cmax:sumOfExes}--\eqref{eq:Cmax:incompatibleCyclesDisable}, \eqref{eq:sumw_jC_j:first}--\eqref{eq:sumw_jC_j:fourth}, and \eqref{eq:sumw_jC_jSpeeds:1}--\eqref{eq:sumw_jC_jSpeeds:last}.
  \qed
\end{lemma}

Recall that $\rho^i_{k,j}$ is the slope of job type $j$ without taking the speed of the particular machine into account, that is, $\rho^i_{k,j} = w_j / p^i_j$.
This remains correct now, since we are going to express the $s$ multiple of $B^{i,s}_k$ (i.e., we locally stretch the time in order to work directly with the integral job sizes of the stretched time axis).
Now, we are ready to define the sought alternative view on $B^{i,s}_k$ that allows us to compute it; to that end we set $2 \cdot \frac{1}{s} \cdot \hat{B}^{i,s}_k(\vey,\ve\alpha) = $
  \begin{equation*}
    \sum_{\beta \in \mathcal{B}^{i,s}} \left[
    \sum_{j = 1}^d \left[ \left(\rho^i_{k,j} - \rho^i_{k,\successor^i_{\mathcal{R},k}(j)}\right) \cdot (\alpha^{i,s}_{k,\beta,j} + \beta \cdot y^{i,s}_{k,\beta})^2 \right]
    -
    \rho^i_{k,\max} \cdot \left( \alpha^{i,s}_{k,0} + \beta \cdot y^{i,s}_{k,\beta} \right)^2
    \right] \\
    +
    \sum_{j=1}^d w_j \cdot p^i_j \cdot y^{i,s}_{j,C^{\operatorname{int}}_k}
    \,.
  \end{equation*}
Observe that we have
\begin{align*}
  \left( \alpha^{i,s}_{k,\beta,j} + \beta \cdot y^{i,s}_{k,\beta} \right)^2
  &=
  \left( \alpha^{i,s}_{k,\beta,j} \right)^2
  +
  2 \cdot \alpha^{i,s}_{k,\beta,j} \cdot \beta \cdot y^{i,s}_{k,\beta}
  +
  \left( \beta \cdot y^{i,s}_{k,\beta} \right)^2 \\
  &=
  \left( \alpha^{i,s}_{k,\beta,j} \right)^2
  +
  2 \cdot \beta \cdot \alpha^{i,s}_{k,\beta,j}
  +
  \beta^2 \cdot y^{i,s}_{k,\beta}
\end{align*}
for all $j = 0, \ldots, d$.
This follows easily from the fact that $y^{i,s}_{k,\beta} \in \{0,1\}$ and furthermore $y^{i,s}_{k,\beta} = 1$ if and only if $\alpha^{i,s}_{k,\beta,j} > 0$.
Note that the above definition of $\hat{B}^{i,s}_k(\vey,\ve\alpha)$ differs from the one without speed only by a (fractional) constant multiple and additional linear terms.
Consequently, it remains separable convex in variables $\vey$ and $\ve\alpha$ (cf. Lemma~\ref{lem:sumw_jC_j:computingHatB}) and it is equivalent to the expression $B^{i,s}_k$ (by the same arguments as in Lemma~\ref{lem:sumw_jC_j:equivalenceOfBs}).
\dkinline{is a concluding lemma needed for $\hat{B}^{i,s}_k$?}

This finishes the description of all model changes needed.

\paragraph{Tardiness.}
We take the same approach as previously, that is, we first compute the penalty for jobs scheduled after the critical time $t_k$ that are already late.
That is, we define $w^i_{k,j} = 0$ if $t_k \le d_j$ and $w^i_{k,j} = w_j$, otherwise.
Now, the green area can be computed as follows
\begin{equation}\tag{$A^{i,s}_k$}
  A^{i,s}_k(\vey)
  =
 \sum_{j = 1}^d w^i_{k,j} \left( y^i_{j,C^{\text{int}}_k} + \sum_{C \in \mathcal{C}^{\text{ext}}_{*,k}} y^i_{j,C} \right) (t_k - d^i_j) \,.
\end{equation}
Finally, $B^{i,s}_k(\vey, \vealpha), X^{i,s}_k(\vey)$ remain nearly the same as in the previous section.
The only subtle difference is that now we use weights $w^i_{k,j}$ instead of $w^i_j$; note that we use these weights when defining the slopes $\rho^i_{k,j}$.

\subsubsection{Handling Fractionality of Objective Functions} \label{sec:fractionality}
In Section~\ref{sec:preliminaries} we have defined MIMO only for objective functions which are integer-valued over integer points, i.e., $\forall \vex \in \Z^{d+d^i}\colon f^i(\vex) \in \Z$ for each $i \in [\tau]$.
This condition is not satisfied in the aforementioned models with speeds because some terms of the objective function are scaled by a fractional factor $1/s$ (or $1/s^2$ in the case of $\sum w_jC_j$ objective in the $\ves$-representation).
Hence, to satisfy the integrality condition, it suffices to multiply all the objective functions by the least common multiple of the denominators $q$ of the fractions (one for each distinct speed) $p/q=1/s$ (with $p,q$ coprime).
This increases the maximum value of the objective function at most by a multiplicative factor of $\Oh(2^{\langle s \rangle \cdot \tau})$ and hence the factor $\log f_{\max}$ in the time complexity of solving MIMO increases at most by a multiplicative factor of $\langle s \rangle \cdot \tau$, that is, polynomially in the input encoding.

We summarize the parameters of the above three models.
In fact, we have added $O((d)^4)$ new variables ($\ve\alpha$); thus, we conclude the following MIMO parameters.
\begin{lemma}
\label{lem:MIMORSumObjectivesSpeeds}
  Let $\mathcal{I}$ be an instance of $R | r^i_j,d^i_j| \mathcal{R}$ in its $\ves$-representation with $\mathcal{R} \in \left\{ \sum w_j C_j, \sum w_j F_j, \sum w_j T_j \right\}$ with $m$ machines of $\kappa$ kinds, speed vectors $\ves_1, \ldots, \ves_{\bar{\tau}}$, hence $\tau = \kappa \cdot \bar{\tau}$ types of machines, and $d$ job types, and denote $p_{\max} = \max_{1 \le j \le d} (\max p_{ij})$.
  There is a MIMO model $\mathcal{S}$ for $\mathcal{I}$ with extension-separable convex objective functions and
  \begin{tasks}[style=itemize](3)
    \task $\mathcal{S}(\Delta) = p_{\max}$,
    \task $\mathcal{S}(M) = \OhOp{(d)^4 \cdot p_{\max}}$,
    \task $\mathcal{S}(d) = d$,
    \task $\mathcal{S}(d^i) = \OhOp{(d)^4 \cdot p_{\max}}$,
    \task $\mathcal{S}(N) = m$, and
    \task $\mathcal{S}(\tau) = \tau = \kappa \cdot \bar{\tau}$. \qed
  \end{tasks}
\end{lemma}
Again, applying parts 1, 3 and 4 of Theorem~\ref{thm:implicitMIMO} we obtain the following.
\begin{theorem}
  Fix an ordered objective $\mathcal{R} \in \left\{ \sum w_jC_j, \sum w_jT_j, \sum w_jF_j \right\}$.
  Problem $R | r^i_j,d^i_j| \mathcal{R}$ in $\ves$-representation with $m$ machines of $\tau$ types, $d$ job types with maximum job size $p_{\max}$, and speeds, admits a fixed-parameter algorithm for parameters
\begin{itemize}
  \item $m + d + p_{\max}$ and
  \item $d + p_{\max}$ \qed
\end{itemize}
\end{theorem}

\subsection{Bin Packing}
In the following subsection we show that several variants of \textsc{Bin Packing} in the high-multiplicity setting can be modeled as MIMO.

\subsubsection{Problem Definitions}
\label{ssec:applications:definitions}
In the {\sc Multiple Knapsack} problem we are given positive integers $B,d,d', \tau$, vectors $\veb^i \in \N^{d'}$ of knapsack capacities and $\mu^i$ of knapsack multiplicities for each knapsack type $i \in [\tau]$, and vectors $\ves_j$ of item sizes for $j \in [d]$.
The task is to partition the $n = \sum_{j=1}^d n_j$ items into $B = \sum_{i=1}^\tau \mu^i$ bins such that sum of items (which are vectors) packed into a bin does not exceed its capacity (in any dimension).
\textsc{Bin Packing} is the case when the dimension of an item $d'$ is equal to $1$, the number of bin types $\tau$ is $1$, and we do binary search over the number $n$ of bins necessary to pack all the items in order to find the smallest such value.
Another variant of \textsc{Bin Packing} generalized by \textsc{Multiple Knapsack} is \textsc{Bin Packing with Cardinality Constraints}, where items are one-dimensional but each bin additionally has a limit on the \emph{number} of items it can pack.
This is modeled as \textsc{Multiple Knapsack} by representing each item as a $2$-dimensional vector, with the first dimension being the item's size, and the second dimension being $1$, and with the second coordinate of the knapsack capacity tuple representing the limit on the number of items it can pack, hence $d=2$, $\tau=1$, and we find the smallest necessary number of bins using binary search.
\textsc{Cutting stock} is the a related problem where $d'=1$, $\tau$ is the number of roll lengths, $\mu^i=n$ and with the $i$-th objective function being a fixed-charge objective incurring a cost $c^i \in \N$ for each bin of size $i$ which is used.

In the {\sc Bin Packing with General Cost Structures} (GCBP)~\cite{AnilyEtAl1994,EpsteinLevin2012} problem we are given~$n$ items with integer sizes $s_1, \dots, s_n$ and a monotonically non-decreasing concave function $f \colon [n] \to \R_{\geq 0}$ with $f(0) = 0$.
The cost of a bin containing $p$ items is $f(p)$.
The task is to find a packing of all items into (at most) $B$ bins each of which contains items of total size at most the common integer capacity such that the total cost is minimized.

\subsubsection{Multiple Knapsack}
The following constraints define the polytope $P^i$ of possible configurations of items in a knapsacks of type $i \in [\tau]$:
\begin{align*}
  \sum_{j = 1}^d s_{j, \delta} x_j &\le b_{\delta}^i & \forall \delta \in \left[ d' \right] \\
  x_j &\ge 0 & \forall j \in [d]
\end{align*}
There is no objective since we only have to decide whether a packing into $B$ knapsacks of a given type exists.
In summary, the properties of the model are:
\begin{lemma}
  Let $\II$ be an instance of~\textsc{Multiple Knapsack}, and let $\sigma=\max_{i \in [d]}(\|\ves_i\|_{\infty})$.
  There exists a MIMO model of $\II$ with an empty objective function and with parameters
  \begin{tasks}[style=itemize](3)
    \task $\mathcal{S}(\Delta) = \sigma$,
    \task $\mathcal{S}(M) = d'$,
    \task $\mathcal{S}(d) = d$,
    \task $\mathcal{S}(d^i) = 0$,
    \task $\mathcal{S}(N) = B$, and
    \task $\mathcal{S}(\tau) = \tau$. \qed
  \end{tasks}
\end{lemma}
Applying parts 1, 3 and 4 of Theorem~\ref{thm:implicitMIMO} gives:
\begin{theorem}\label{thm:vector_bin_packing}
  Let $\sigma=\max_{i \in [d]}(\|\ves_i\|_{\infty})$.
  \textsc{Multiple Knapsack} is fixed-parameter tractable parameterized by
  \begin{itemize}
  \item $d+d'+ B$,
  \item $d+d'+\sigma$, and,
  \item $d+d'+\tau$ if $\sigma$ is given in binary. \qed
  \end{itemize}
\end{theorem}

\subsubsection{Bin Packing with General Cost Structures}
Let us turn to the \textsc{Bin Packing with General Cost Structures} problem.
There is only one type of bin, hence $\tau=1$.
The polytope $P^1$ describing the set of configurations of items of a bin is given simply by the knapsack constraint
\begin{align*}
  \sum_{j = 1}^d s_{j} x_j &\le b \\
                       x_j &\ge 0 & \forall j \in [d] \enspace.
\end{align*}
The objective function is $f^1(\vex) = f(\sum_{j=1}^d x_j)$, and it is in general concave, since if it is linear the problem is trivial.
Summarizíng the properties of the above model:
\begin{lemma}
  Let $\II$ be an instance of \textsc{Bin Packing with General Cost Structures}. There is a MIMO model of $\II$ with a concave objective and parameters
  \begin{tasks}[style=itemize](3)
    \task $\mathcal{S}(\Delta) = \sigma$,
    \task $\mathcal{S}(M) = 1$,
    \task $\mathcal{S}(d) = d$,
    \task $\mathcal{S}(d^i) = 0$,
    \task $\mathcal{S}(N) = B$, and
    \task $\mathcal{S}(\tau) = 1$. \qed
  \end{tasks}
\end{lemma}
Now Part~\ref{thm:implicitMIMO:ch} of Theorem~\ref{thm:implicitMIMO}, the only part applicable to a MIMO with concave objectives, yields the following theorem.
\begin{theorem} \label{thm:bin_packing_gcs}
  \textsc{Bin Packing with General Cost Structures} is fixed-parameter tractable when parameterized by $d$ and all $s_i$ are given in unary.
\qed
\end{theorem}

\subsection{Surfing}
\label{sec:surfing}
Here we discuss the \textsc{Surfing} problem, which demonstrates the utility of MIMO beyond scheduling and bin packing, and particularly the usefulness of part~\ref{thm:implicitMIMO:huge} of Theorem~\ref{thm:implicitMIMO} in handling many types of locations.

In the \textsc{Surfing} problem we have $d'$ \emph{commodities} and $d''$ \emph{servers} (``service providers''), with each server $j \in [d'']$ declaring a \emph{supply vector} $\ves_j \in \N^{d'}$ indicating how much of each commodity it is capable of supplying.
Moreover, $N$ is a large number of \emph{surfers} of $\tau$ types such that there are $\mu^i \in \N$ surfers of type $i$ and $N = \mu^1 + \cdots + \mu^\tau$, and for each type $i \in [\tau]$ there is a \emph{demand vector} $\vedelta^i \in \N^{d'}$ with respect to the commodities, a \emph{capacity vector} $\vegamma^i \in \N^{d''}$ with respect to the servers, and a cost vector $\vecc^i \in \N^{d' \cdot d''}$ with respect to commodity-server pairs.
For example, $\delta^i_j$ is the demand of a surfer of type $i \in [\tau]$ for commodity $j \in [d']$, $\gamma^i_j$ is a bound on the total amount of all commodities received by a surfer of type $i \in [\tau]$ from the server $j \in [d'']$, and $c^i_{j',j''}$ is the cost per one unit of commodity $j' \in [d']$ received from server $j'' \in [d'']$ by a surfer of type $i \in [\tau]$.
The task is to determine, for each surfer, how much commodity they should buy from each server so as to satisfy the surfer's demand for each commodity, stay within capacity bounds, and minimize total cost, while also staying within each server's supply.
We denote $\vedelta = (\vedelta^1, \dots, \vedelta^\tau)$, $\vegamma = (\vegamma^1, \dots, \vegamma^\tau)$, and $\vecc = (\vecc^1, \dots, \vecc^\tau)$.

We model \textsc{Surfing} as MIMO in the following way.
We let $d = d' \cdot d''$ and we let $\ven = (\ves_1, \dots, \ves_{d''})$.
Let $x_{jk}$ be a variable describing how much commodity $j$ a given surfer is buying from the server $k$.
The polytope $P^i$ describing assignments satisfying the surfer's demands and staying within bounds is given by the following constraints:
\begin{align*}
\sum_{k=1}^{d''} x_{jk} &= \delta^i_j & \forall j \in [d'] \\
\sum_{j=1}^{d'} x_{jk} &\leq \gamma^i_k & \forall k \in [d''] \\
\vex & \geq \vezero\,, &
\end{align*}
and the objective function of a surfer of type $i \in [\tau]$ is $f^i(\vex) = \sum_{j=1}^{d'} \sum_{k=1}^{d''} c_{jk}^i x_{jk}$.
It remains to deal with the fact that we do not have to use up all the available supply, so we introduce a ``slack'' surfer type $\tau+1$ with only non-negativity constraints and no objective.
\begin{lemma}
There is a MIMO model $\SSS$ for \textsc{Surfing} with parameters
\begin{tasks}(3)
\task $\SSS(\Delta) = 1$,
\task $\SSS(M) = d'+ d''$,
\task $\SSS(d) = d = d' \cdot d''$,
\task $\SSS(d^i) = 0$,
\task $\SSS(N) = N$,
\task $\SSS(\tau) = \tau+1$. \qedhere
\end{tasks}
\end{lemma}
Applying part~\ref{thm:implicitMIMO:huge} of Theorem~\ref{thm:implicitMIMO} then yields that
\begin{theorem} \label{thm:surfing}
\textsc{Surfing} can be solved in time $(d^2)^{\Oh(d^2)} \la N, \ven, \vedelta, \vegamma, \vecc \ra^{\Oh(1)}$, i.e., single-exponential in $d=d' \cdot d''$ and polynomial in the binary encoding of the rest of the input data. \qed
\end{theorem}

\section{Research Directions} \label{sec:researchdirections}
Our work poses several research directions.
Due to Part~\ref{thm:implicitMIMO:gr} of Theorem~\ref{thm:implicitMIMO} and the work of Goemans and Rothvo{\ss}~\cite{GoemansRothvoss2014} we now have a fairly good understanding of the parameterized complexity of many high-multiplicity scheduling problems in the setting with few job types and unary job sizes.
However, a glaring large class where we do not know anything beyond small job sizes are scheduling problems with objectives from $\mathfrak{C}_{\text{poly}}$, the simplest ones being $P || \sum w_j C_j$ or $P || \ell_2$.
More generally, could Part~\ref{thm:implicitMIMO:gr} of Theorem~\ref{thm:implicitMIMO} be extended to the optimization of separable convex objectives?
On the negative side, our hardness results hint at (parameterized) intractability of scheduling with few job types but many machine types.
This leads us to ask whether $R || C_{\max}$ with many machine types is \NP-hard for a constant number of job types~$d$, or perhaps \Wh{1} parameterized by $d$?

\section*{Acknowledgments}
Part of this work was carried out during the workshop ``Scheduling Meets Fixed-Parameter Tractability'' held at the Lorenz Center in Leiden in February 2018.
We thank the Lorenz Center for providing a welcoming and stimulating environment.

\clearpage
\pagebreak
\appendix
\section{Appendix: Omitted Proofs} \label{sec:technical}

\subsection{Proofs for Section~\ref{sec:preliminaries}}
\begin{proof}[Proof of Lemma~\ref{lem:MIMO_as_nfold}]
	Let $E_1 = (I~\vezero) \in \Z^{d \times (d + D + M)}$ where $I$ is the $(d\times d)$-identity matrix and $\vezero$ is a $(d \times (D+M))$-all-zero matrix, let $E_1^i = E_1$ for each $i \in [\tau]$, and let $\veb^0 = \ven$.
	The last $M$ coordinates of each brick will play the role of slack variables in order to model inequalities in the system $A^i \vex \leq \veb^i$.
	For each $i \in [\tau]$, obtain $E_2^i$ from $A^i$ by adding $M-m^i$ zero rows and $D-d^i$ zero columns, and then appending from the right the $M \times M$ identity matrix, ensuring $E_2^i$ has $M$ rows and $d+D+M$ columns, and append $M-m^i$ zeroes to $\veb^i$.
	Formally extend the objective function $f^i$ to $d+D+M$ dimensions by making it ignore the last $M+D-d^i$ dimensions.
	For each $i \in [\tau]$, define $\vel^i = \{-\infty\}^{d+d^i} \times \{0\}^{M+D-d^i}$ and $\veu^i = \{+\infty\}^{d+d^i} \times \{0\}^{D-d^i} \times \{+\infty\}^{M}$.
	Let $\mu^i$ be the number of bricks of type $i$, for each $i \in [\tau]$.

	It is easy to check that, for each brick $j \in [N]$ of type $i \in [\tau]$ of the resulting huge $N$-fold formulation, $\vex^j$ restricted to the first $d+d^i$ coordinates can take on exactly the values of $P^i \cap \Z^{d+d^i}$.
	Moreover, the objective value of the brick is exactly the objective value of corresponding point of $P^i \cap \Z^{d+d^i}$ in the MIMO problem.
	Finally, by the definition of $E_1$, the sum of the restrictions of all bricks to the first $d$ coordinates is exactly $\ven$.
	This shows that we have reduced the MIMO instance to a huge $N$-fold IP instance, and the bounds are clearly as stated in the Lemma.
\end{proof}

\subsection{Proofs for Section~\ref{sec:part4}}
We build on the following theorem of Goemans and Rothvo{\ss}~\cite{GoemansRothvoss2014}:

\begin{proposition}[{Structure Theorem~\cite{GoemansRothvoss2014}}]
\label{prop:structurethm}
  Let $P = \{\vex \mid A\vex \leq \veb\} \subseteq \R^d$ be a polytope with $A \in \Z^{m \times d}$ and $\veb \in \Z^m$ such that all coefficients are bounded by $\Delta$ in absolute value.
  Then there exists a set $Y \subseteq P \cap \Z^d$ of size $|Y| \leq S := m^d d^{\Oh(d)} (\log \Delta)^d$ that can be computed in time $S^{\Oh(1)}$ with the following property.
  For every vector $\ven = \sum_{\vex \in P \cap \Z^d} \lambda_\vex \vex$ with $\velambda \in \N^{P \cap \Z^d}$, there exists a vector $\hat{\velambda} \in \N^{P \cap \Z^d}$ such that $\ven = \sum_{\vex \in P \cap \Z^d} \hat{\lambda}_\vex \vex$, and,
  \begin{tasks}(3)
    \task $\hat{\lambda}_{\vex} \in \{0,1\} \, \forall \vex \not\in Y$,
    \task $|\suppo(\hat{\velambda}) \cap Y| \leq 2^{2d}$,
    \task $|\suppo(\hat{\velambda}) \setminus Y| \leq 2^{2d}$.
  \end{tasks}
\end{proposition}

We start by showing a natural extension of the Structure Theorem (Proposition~\ref{prop:structurethm}) to the multitype setting.

\begin{lemma}[Multitype Structure Lemma]
\label{lem:structure}
  Let $P^1, \dots, P^\tau$ be a $P$-representation of $X^1, \dots, X^\tau$.
  Then, for each $i \in [\tau]$, there exists a set $Y^i \subseteq P^i \cap \Z^{d+d^i}$ of size $|Y^i| \leq S^i := (m^i)^{(d+d^i)} (d+d^i)^{\Oh(d+d^i)} (\log \Delta)^{d+d^i}$ that can be computed in time $(S^i)^{\Oh(1)}$ with the following property.
  For every vector $\ven = \sum_{i=1}^\tau \sum_{(\vex, \vex') \in P^i \cap \Z^{d+d^i}} \lambda^i_{(\vex, \vex')} \vex$ with non-negative integral $\velambda$, there exists a non-negative integral vector $\hat{\velambda} $ such that $\ven = \sum_{i=1}^\tau \sum_{(\vex, \vex') \in P^i \cap \Z^{d+d^i}} \hat{\lambda}^i_{(\vex, \vex')} \vex$, and, for each $i \in [\tau]$,
  \begin{tasks}(2)
    \task $\hat{\lambda}^i_{(\vex, \vex')} \in \{0,1\} \, \forall (\vex, \vex') \not\in Y^i$,
    \task $|\suppo(\hat{\velambda}^i) \cap Y^i| \leq 2^{2(d+d^i)}$,
    \task $|\suppo(\hat{\velambda}^i) \setminus Y^i| \leq 2^{2(d+d^i)}$,
    \task \label{lem:structure:1norm} $\|\hat{\velambda}^i\|_1 = \|\velambda^i\|_1$.
  \end{tasks}
\end{lemma}
\begin{proof}
  First, extend each $P^i$ with a coordinate which is always $1$, i.e., replace $P^i$ with $\{(1, \vex, \vex') \mid (\vex, \vex') \in P^i\}$.
  This only increases the dimension by $1$ and requires an additional equality constraint.
  Then, apply the Structure Theorem to each $P^i$ individually.
  Let $\mu^i = \|\velambda^i\|_1$ and $N = \|\velambda\|_1$.
  Now observe that if $(N,\ven) = \sum_{i=1}^\tau \sum_{(1, \vex, \vex') \in P^i \cap \Z^{d+d^i}} \lambda^i_{(1, \vex, \vex')} (1,\vex)$, then there is a decomposition of $(N, \ven)$ into $\tau$ summands $(\mu^i, \ven^i) = \sum_{(1, \vex, \vex') \in P^i \cap \Z^{d+d^i}} \lambda^i_{(1, \vex, \vex')} (1,\vex)$ to which the statement of the Structure Theorem applies directly, and we obtain all points except for point~\ref{lem:structure:1norm}.
  To argue this last point, observe that both decompositions $\velambda^i$ and $\hat{\velambda}^i$ decompose $\ven^i$ into $\mu^i$ points of $P^i \cap \Z^{d+d^i}$ and the claim holds.
  Thus, for each $i \in [\tau]$, $Y^i$ is obtained by applying Proposition~\ref{prop:structurethm} to the extended polytope $P^i$ and then projecting out the first coordinate of each element of the computed set.
\end{proof}

\begin{proof}[Proof of Theorem~\ref{thm:implicitMIMO}, part~\ref{thm:implicitMIMO:gr}]
  First, compute the sets $Y^1, \dots, Y^\tau$ from Lemma~\ref{lem:structure}.
  Our goal now is to set up an ILP in small dimension which will correspond to an optimal solution $\velambda$ with the properties of Lemma~\ref{lem:structure} (i.e., most of its support lies in $\bigcup_i Y^i$).
  Fix such an optimal $\velambda$.
  For each $i \in [\tau]$, guess a subset $Z^i \subseteq Y^i$ which satisfies $|Z^i| \leq 2^{2(d+d^i)}$ and $|\suppo(\velambda^i) \setminus Z^i| \leq 2^{2(d+d^i)}$.
  Also guess the number $\bar{\mu}^i = |\suppo(\velambda^i \setminus Z^i)|$.
  Note that there are $\prod_{i=1}^\tau (S^i)^{\Oh(2^{2(d+D)})}$ choices.
  For each guess, we construct an IP with few variables and a linear objective function, solve it, and then pick the best obtained value across all guesses and transform it into a solution of MIMO.

  Introduce a variable $\lambda^i_{(\vez, \vez')}$ for each $i \in [\tau]$ and each $(\vez, \vez') \in Z^i$.
  Additionally, for each $i \in [\tau]$, introduce $\bar{\mu}^i$ vectors of variables $(\vex, \vex')_j^i$.
  Then, consider the following constraints:
  \begin{align}
    A^i (\vex, \vex')_j^i &\leq \veb^i & \forall i \in [\tau], \, j \in [\bar{\mu}^i] \label{eq:mimo:gr:polytope}\\
    \sum_{i=1}^{\tau} \left[ \sum_{(\vez, \vez') \in Z^i} \lambda^i_{(\vez, \vez')} \vez  + \sum_{j=1}^{\bar{\mu}^i} \vex_j^i \right] & = \ven \label{eq:mimo:gr:decomp}\\
    \sum_{(\vez, \vez') \in Z^i} \lambda^i_{(\vez, \vez')} &= \mu^i - \bar{\mu}^i & \forall i \in [\tau] \text{ with } \vezero \not\in X^i \label{eq:mimo:gr:nozero}\\
    \sum_{(\vez, \vez') \in Z^i} \lambda^i_{(\vez, \vez')} &\leq \mu^i - \bar{\mu}^i & \forall i \in [\tau] \text{ with } \vezero \in X^i \label{eq:mimo:gr:zero} \\
    \lambda_{(\vez, \vez')}^i & \in \N & \forall i \in [\tau], \, \forall (\vez, \vez') \in Z^i \label{eq:mimo:gr:confs}\\
    (\vex, \vex')_j^i & \in \Z^{d+d^i} & \forall i \in [\tau], \, \forall j \in [\bar{\mu}^i], \label{eq:mimo:gr:polytope2}
  \end{align}
  and, depending on the objective of the MIMO instance, solve with one of the objectives
  \begin{align*}
    \text{linear}(\velambda, \vex, \vex') &= \sum_{i=1}^{\tau} \left[ \sum_{(\vez, \vez') \in Z^i} \lambda^i_{(\vez, \vez')} (\vew^i \vez)  + \sum_{j=1}^{\bar{\mu}^i} \vew^i \vex_j^i \right], \, \text{ or}, \\
    \text{fixed-charge}(\velambda, \vex, \vex') &= \sum_{i=1}^{\tau} \left[ \sum_{(\vez, \vez') \in Z^i} \lambda^i_{(\vez, \vez')} c^i  + \bar{\mu}^i c^i \right] \enspace .
\end{align*}
The constraints~\eqref{eq:mimo:gr:polytope} and~\eqref{eq:mimo:gr:polytope2} ensure that the variable vectors $(\vex, \vex')_j^i$ assume values from $P^i \cap \Z^{d+d^i}$, for each $i \in [\tau]$, enforcing the meaning that these variables represent the part of solution $\hat{\velambda}$ whose support does not lie in $Y^i$ (cf. Lemma~\ref{lem:structure}).
The constraint~\eqref{eq:mimo:gr:decomp} ensures that the solution indeed corresponds to a decomposition of $\ven$ into points from $\bigcup (P^i \cap \Z^{d+d^i})$.
The constraints~\eqref{eq:mimo:gr:nozero}--\eqref{eq:mimo:gr:zero} ensure that the number of non-zero configurations of type $i$ is at most $\mu^i$.

Let $(\velambda, \vex, \vex')$ be an optimum of the model above, computed using an algorithm for ILP in small dimension (cf.~\cite{FrankTardos1987,Kannan1987}).
We construct a solution $\velambda^*$ of MIMO as follows.
For each $i \in [\tau]$ and each $\vez \in \Z^d$ such that $\vez \in \pi^i(P^i)$, let
$$\lambda^*(i, \vez) = \left(\sum_{\substack{\vez' \in \Z^{d^i}:\\ (\vez, \vez') \in P^i}} \lambda(i, \vez, \vez') + \sum_{\substack{j \in [\bar{\mu}^i],\\ (\vex, \vex')_j^i = (\vez, \vez')}} 1 \right),$$
and let $\lambda^*(i, \mathbf{0}) = \mu^i - \sum_{\vez \in (\pi^i(P^i) \setminus \mathbf{0}) \cap \Z^d} \lambda^*(i, \vez)$.
We argue that $\velambda^*$ is an optimal solution.

First, consider a linear objective.
Observe that any solution of MIMO induces a decomposition $\ven = \sum_{i=1}^\tau \ven^i$, and that this decomposition fully determines the objective function, which becomes $\sum_{i=1}^\tau \vew^i \ven^i$.
Furthermore, part~\ref{lem:structure:1norm} of Lemma~\ref{lem:structure} guarantees that we can (almost) restrict our attention to the special sets $Y^i$ without ruling out any decomposition of $\ven$ into $\ven^i$.
Second, considering a fixed-charge objective, observe that the constraint~\eqref{eq:mimo:gr:zero} and our separate handling of $\lambda^*(i, \mathbf{0})$ encodes this objective appropriately.

Regarding runtime, the number of times we solve the ILP constructed above is equal to the number of guesses of the sets $Z^i$ and the numbers $\bar{\mu}^i$, which is
\begin{align*}
    \prod_{i=1}^\tau (S^i)^{\Oh(2^{2(d+D)})} & \leq \left((M)^{(d+d^{\max})} (d+D)^{\Oh(d+D)} (\log \Delta)^{(d+D)}) \right)^{\tau 2^{\Oh(d+D)}} \\
                                             & \leq \left( (M + d + D + \log \Delta)^{d+d_{\max}} \right)^{\tau 2^{\Oh(d+D)}} \\
                                             & \leq \left(M + d + D + \log \Delta\right)^{\tau^{(d+D)^{\Oh(1)}}} \enspace .
\end{align*}
The ILP we have constructed has dimension $$p=\sum_{i=1}^{\tau} (|Z^i| + \bar{\mu}^i(d+d^i)) \leq \tau 2^{2(d+D} + 2^{2(d+D}(d+D) \leq (\tau + d + D) 2^{2(d+D)},$$
and can be solved in time $p^{\Oh(p)} \hat{L}^{\Oh(1)}$ by Kannan's algorithm~\cite{Kannan1987} (recall that $\hat{L} = \la N, \ven, \veb^1, \dots, \veb^\tau, \Delta, f_{\max} \ra$).
so the total runtime is upper bounded by $(d D M \log \Delta)^{\tau^{(d+D)^{\Oh(1)}}} \hat{L}^{\Oh(1)}$.
\end{proof}

\subsection{Hardness}
\begin{repproposition}{prop:hardness}
   Solving MIMO systems
  \begin{enumerate}
    \item \label{prop:hardness:wh1} is \Wh{1} parameterized by $d$ only, even if $\|X\|_\infty$ is given in unary;
    \item \label{prop:hardness:fixedcharge} with a fixed-charge objective is \NPh even with $d=1$ and $\|X\|_\infty = 1$ (but with large penalties);
    \item \label{prop:hardness:concave} with a separable concave quadratic objective is \textbf{a)} \NPh even with $d=2$ and $\|X\|_\infty = 1$, and \textbf{b)} \Wh{1} parameterized by $d$ even when the largest coefficient of the objective is given in unary and $\|X\|_\infty = 1$.
  \end{enumerate}
\end{repproposition}
\begin{proof}
  \noindent \textbf{Part~\ref{prop:hardness:wh1}.}
  Consider the \textsc{Unary Bin Packing} problem, in which we have $n$ items of sizes $o_1, \dots, o_n \in \N$ with $\max_{i \in [n]} o_i = O \leq \poly(n)$, a capacity $B \in \N$, and an integer $k \in \N$, and we ask whether the items can be packed into $k$ bins of capacity $B$.
  Jansen et al.~\cite{JansenEtAl2013} have shown that \textsc{Unary Bin Packing} is \Wh{1} parameterized by $k$, even for \emph{tight} instances where $\sum_{i=1}^n o_i = kB$.

  We shall construct a MIMO instance with $n$ types.
  We let $P^i$, for each $i \in [n]$, be defined by the system
  \begin{align}
    \sum_{j=1}^k x_j & = o_i,            & \label{eq:hardness:sum} \\
    \sum_{j=1}^k y_j & = 1,              & \notag\\
                 x_j & \leq O y_j,       & \forall j \in [k], \label{eq:hardness:bigm} \\
                \vex & \geq \mathbf{0},  & \label{eq:hardness:nonneg}
  \end{align}
  and we let $d=k$, $d^i = k$, and $\mu^i = 1$, for each $i \in [n]$, and define the $\vey$ variables to be the one which are discarded by the projection $\pi^i$.
  Finally, we let $\ven$ be a $k$-dimensional vector of all $B$.

  It is easy to see that $\pi^i(P^i) \cap \Z^k = \{(o_i, 0, \dots, 0), (0, o^i,0, \dots, 0), \dots, (0, \dots, 0, o_i)\}$ and each element encodes the bin into which item $i$ is assigned.
  Thus, the MIMO instance is feasible if and only if there exists an assignment of items to bins such that the sum of item sizes of each bin is $B$.

  \medskip

  \noindent \textbf{Part~\ref{prop:hardness:concave}b).}
  We will continue working with the MIMO instance constructed above.
  Recall that minimizing a concave function is equivalent with maximizing a convex function, which is the perspective we shall take here.
  Our goal now is to model the constraint~\eqref{eq:hardness:bigm} which involves the big coefficient $O$ by the objective.
  Let each $P^i$ now be only defined by the constraints~\eqref{eq:hardness:sum} and~\eqref{eq:hardness:nonneg}, so $d=k$ and $d^i = 0$ for each $i \in [n]$.
  Let $f^i(\vex) = \sum_{j=1}^k f_j^i(x_j)$ with $f^i_j(x_j) = (x_j - \frac{o_i}{2})^2$.
  This means each $f^i_j$ is maximized exactly at the endpoints of the feasible interval $[0,o_i]$.
  Thus, a solution with value $\sum_{i=1}^n (\frac{o_i}{2})^2$ must have each $x_j \in \{0,o_i\}$ and thus corresponds to a bin packing; it is easy to check that no better value is attainable, which finishes the proof.

  \medskip

  \noindent \textbf{Part~\ref{prop:hardness:fixedcharge}.}
  In the \textsc{Partition} problem we are again given $n$ numbers $o_1, \dots, o_n$, however now their size can be large.
  The task is to decide whether there is a subset $I \subseteq [n]$ of indices such that $\sum_{i \in I} o_i = \sum_{i \not\in I} o_i$.
  We again construct a MIMO instance with $n$ types.
  This time, each $P^i$ is simply a segment defined by $0 \leq x \leq o_i$, and let $\mu^i = 1$ for each $i \in [n]$.
  Let $a = \frac{1}{2} \sum_{i=1}^n o_i$.
  Then, for each $i \in [n]$, let $f^i(x) = o_i$ if $x \neq 0$.
  We claim that there is a solution of value $a$ if and only if the \textsc{Partition} instance was a ``yes'' instance.

  In one direction, let $I$ be a solution of the \textsc{Partition} instance.
  Then setting $\lambda(i, o_i)=1$ if $i \in I$ and $\lambda(i, 0)=1$ otherwise clearly defines a decomposition of $a$ into elements $o_i$ and has objective value $a$.
  In the other direction, assume for contradiction that there was a solution of value $a$ but the instance was a ``no'' instance.
  Let $I$ be the set of indices of the types $i \in [n]$ for which $\lambda(i,0) = 0$.
  By our definition of the $f^i$'s it must hold that $\sum_{i \in I} o_i = a$, which means that $I$ certifies the instance was a ``yes'' instance, a contradiction.

  \medskip

  \noindent \textbf{Part~\ref{prop:hardness:concave}a).}
  We continue with the \textsc{Partition} problem.
  This time we let $P^i$ be defined by $x_1 + x_2 = o_i$ and $x_1, x_2 \geq 0$, and let $\mu^i = 1$ for each $i$.
  Again, let $a = \frac{1}{2} \sum_{i=1}^n o_i$, and let $\ven=(a,a)$.
  Our goal now is to use a separable concave quadratic objective to enforce that either $x_1 = o_i$ or $x_2 = o_i$.
  This is done similarly as in the case of \textsc{Unary Bin Packing}: we let $f^i_1(x_1) = (x_1 - \frac{o_i}{2})^2$ and similarly $f^i_1(x_2) = (x_2 - \frac{o_i}{2})^2$.
  It is easy to check that the only way to get $2 (\frac{o_i}{2})^2$ is if either $x_1$ or $x_2$ is $o_i$.
  Thus, it is enough to check whether a solution exists with value $2 \sum_{i=1}^n (\frac{o_i}{2})^2$.
  By the arguments above, this is the case if and only if the \textsc{Partition} instances was a ``yes'' instance.
\end{proof}

Recall Lemma~\ref{lem:MIMO_as_nfold}, which states that MIMO reduces to huge $N$-fold IP.
So far we have used it to obtain positive results by encoding MIMO as $N$-fold IP and then applying various \FPT algorithms to the obtained IPs.
Now it will be useful to obtain hardness of $N$-fold IP from Proposition~\ref{prop:hardness}.
Specifically, applying Lemma~\ref{lem:MIMO_as_nfold} to the hardness instances of the previous proposition gives the existence of hard $N$-fold instances with parameters $r=M$, $s=d$, $t=d+D+M$, $\|E\|_\infty = \Delta$, and $N = \|\vemu\|_1\|_1$, which implies that
\begin{corollary}
  Solving $N$-fold IPs
  \begin{enumerate}
    \item is \Wh{1} parameterized by $r,s$, and $t$ when $\|E\|_\infty$ is unary.
    \item is \NPh with a fixed-charge objective even with $r=s=t=1$ and $\|X\|_\infty = 1$ (but with large penalties).
    \item with a separable concave quadratic objective is \textbf{a)} \NPh even with $r=t=2$ and $s=\|E\|_\infty = 1$ (but with large coefficients in the objective), and \textbf{b)} \Wh{1} parameterized by $r$ and $t$ even when $s=\|E\|_\infty = 1$ and when the largest coefficient of the objective is given in unary.
  \end{enumerate}
\end{corollary}

\bibliographystyle{abbrvnat}
\bibliography{scheduling}

\end{document}